\documentclass[12pt]{article}
\usepackage{graphicx,amsfonts,amsmath,amssymb,mathrsfs, comment}

\title{The Point Processes of the GRW Theory\\ of Wave Function Collapse\footnote{A version of this work has been submitted as a Habilitation thesis to the Mathematics Institute of Eberhard-Karls-Universit\"at T\"ubingen, Germany. The main difference between the thesis and the present version is that the proof of Theorem \ref{thm:Kol} (a Kolmogorov extension theorem for POVMs) was included in the thesis but not here, as it has been submitted for publication separately \cite{Tum07c}.}}
\author{
Roderich Tumulka\footnote{Department of Mathematics,
	Rutgers University, 110 Frelinghuysen Road, Piscataway, 
	NJ 08854-8019, USA. E-mail:
     tumulka@math.rutgers.edu}
}
\date{October 31, 2007}

\newcommand{\Hilbert}{\mathscr{H}}

\newcommand{\conf}{\mathcal{Q}}
\newcommand{\Q}{\conf}
\renewcommand{\Re}{\mathrm{Re}}

\newcommand{\D}{\mathrm{d}}
\newcommand{\E}{\mathrm{e}}
\newcommand{\I}{\mathrm{i}}
\newcommand{\be}{\begin{equation}}
\newcommand{\ee}{\end{equation}}

\newcommand{\PPP}{\mathbb{P}}
\newcommand{\RRR}{\mathbb{R}}
\newcommand{\CCC}{\mathbb{C}}
\newcommand{\QQQ}{\mathbb{Q}}
\newcommand{ \ZZZ}{\mathbb{Z}}
\newcommand{\NNN}{\mathbb{N}}
\newcommand{\scp}[2]{\langle #1|#2 \rangle}
\newcommand{\Bscp}[2]{\Bigl\langle #1 \Big| #2 \Bigr\rangle}

\newcommand{\Cauchy}{\mathcal{C}}
\newcommand{\Hyper}{\mathcal{H}}
\newcommand{\hyper}{\mathbb{H}}
\newcommand{\Dirac}{\mathscr{D}}
\newcommand{\ini}{\mathcal{I}}
\newcommand{\future}{J^+}
\newcommand{\past}{J^-}

\newcommand{\Bdd}{\mathscr{B}}
\newcommand{\salg}{\mathcal{A}}
\newcommand{\Lab}{\mathscr{L}}
\newcommand{\law}{\mathscr{L}}
\newcommand{\povm}{G}
\newcommand{\Borel}{\mathcal{B}}
\newcommand{\Power}{\mathcal{P}}

\newcommand{\ceme}{\diamond}
\newcommand{\tdist}{\tau} 
\newcommand{\sdist}{\mathrm{dist}} 
\newcommand{\profile}{\ell}

\newtheorem{thm}{Theorem}

\newtheorem{cor}{Corollary}
\newtheorem{defn}{Definition}
\newtheorem{lemma}{Lemma}
\newtheorem{ass}{Assumption}
\newenvironment{proof}{\noindent\textit{Proof. }}{\hfill$\square$ \medskip }
\newenvironment{proofthm}[1]{\noindent\textit{Proof of Theorem \ref{#1}. }}{\hfill$\square$ \medskip }
\newcounter{ex}
\newenvironment{ex}{\bigskip\noindent\textbf{Example \theex\ }\addtocounter{ex}{1}}{\hfill$\square$ \bigskip}

\begin{document}
\maketitle
\begin{abstract}
The Ghirardi--Rimini--Weber (GRW) theory is a physical theory that, when combined with a suitable ontology, provides an explanation of quantum mechanics. The so-called collapse of the wave function is problematic in conventional quantum theory but not in the GRW theory, in which it is governed by a stochastic law. A possible ontology is the flash ontology, according to which matter consists of random points in space-time, called flashes. The joint distribution of these points, a point process in space-time, is the topic of this work. The mathematical results concern mainly the existence and uniqueness of this distribution for several variants of the theory. Particular attention is paid to the relativistic version of the GRW theory that I developed in 2004.
\medskip

\end{abstract}

\tableofcontents

\section{Introduction}

This work concerns the foundations of quantum mechanics. 
The Ghirardi--Rimini--Weber (GRW) theory is a proposal for a precise definition of quantum mechanics, intended to replace the conventional rules of quantum mechanics (as formulated by, e.g., Dirac and von Neumann) and to overcome the certain vagueness and imprecision inherent in these rules. This vagueness and imprecision arise from the situation that these rules specify what a macroscopic observer will see when measuring a certain observable, but leave unspecified exactly which systems should be counted as macroscopic, or as observers, and exactly which physical processes should be counted as measurements. The GRW theory, as proposed in 1986 by Ghirardi, Rimini, and Weber \cite{GRW86} and Bell \cite{Bell87}, solves this problem for the entire realm of non-relativistic quantum mechanics, and a key role in this theory is played by a stochastic law according to which wave functions collapse at random times, rather than at the intervention of an observer. It is a ``quantum theory without observers'' \cite{Gol98}.

After the success of this approach with non-relativistic quantum mechanics, the question arises whether and how the GRW theory can be extended to quantum field theory, to relativistic quantum mechanics, and to relativistic quantum field theory. This question has been worked on intensely over the past 20 years, but not completely and finally answered. The first seriously relativistic theories of the GRW type, and in fact the first seriously relativistic quantum theories without observers, were developed in 2002 by Dowker and Henson \cite{Fay02} (on a discrete space-time) and in 2004 by myself \cite{Tum04} (on a flat or curved Lorentzian manifold). A major part of this work (Section~\ref{sec:defrGRWf}) consists of a study of the model I have proposed. This model, which I will abbreviate rGRWf for ``relativistic GRW theory with flash ontology,'' uses some elements that were suggested for this purpose already in 1987 by Bell \cite{Bell87}, in particular the ``flash ontology,'' which corresponds to a point process in space-time. Since the flash ontology is incompatible with the standard way of extending the GRW theory to quantum field theory---the CSL (continuous spontaneous localization) approach pioneered particularly by Pearle \cite{Pe89} and employing diffusion processes in Hilbert space---, I developed in \cite{Tum05} a different way of extending the GRW theory to quantum field theories, suitable for flashes. A key element of this extension is an abstract scheme generalizing the original GRW theory (which applies to non-relativistic quantum mechanics), in which the theory is defined by specifying the \emph{Hamiltonian operator} (as in conventional quantum theory) and the \emph{flash rate operators}. This scheme is directly applicable to quantum field theories. A major part of this work (Sections~\ref{sec:scheme} and \ref{sec:schemer}) consists of a description, further generalization and mathematical analysis of this scheme, including existence theorems providing exact conditions for the existence of the relevant point processes. The further generalization is necessary to include the process of the rGRWf theory. The goal of this work is to provide a firm mathematical basis for the GRW theories with flash ontology.

It lies in the nature of the topic that this work must be a mixture of mathematics, physics, and philosophy. The theorems and proofs I present appear here for the first time, while the physical (and philosophical) considerations I report about have been published before \cite{Fay02,Tum04, Tum05, Tum06,Tum07,AGTZ06}. 
The relevant mathematical considerations involve concepts and results from several fields, including stochastic processes; operators in Hilbert space; and differential geometry of Lorentzian manifolds. The main results of this work are existence proofs for the relevant point processes. An existence question arises in many physical theories and is often remarkably difficult. For example, the existence of Newtonian trajectories with Coulomb interaction (for almost all initial conditions) is still an open problem for more than 4 particles. For existence results about other quantum theories without observers, see \cite{bmex,GT05a,TT05}. A simple introduction to rGRWf I have given in \cite{Tum06}; discussions of rGRWf can also be found in \cite{ADLZ05,AGTZ06,Jad06,Mau05, Mau07,Ghi07}.

\subsection{Physical Motivation}

When the standard quantum formalism utilizes the concept of collapse of the wave function, it does so in a rather ill-defined way, introducing a collapse whenever ``an observer'' intervenes. This is replaced by a concept of objective collapse, or spontaneous collapse, in GRW-type theories. These theories replace the unitary Schr\"odinger evolution of the wave function by a nonlinear, stochastic evolution, so that the Schr\"odinger evolution remains a good approximation for microscopic systems while superpositions of macroscopically different states (such as Schr\"odinger's cat) quickly collapse into one of the contributions. The GRW theory \cite{GRW86,Bell87,BG03} is the simplest and best-known theory of this kind, another one the Continuous Spontaneous Localization (CSL) approach \cite{Pe89,BG03}. These theories, when combined with a suitable ontology, provide  paradox-free versions of quantum mechanics and possible explanations of the quantum formalism in terms of objective events, and thus ``quantum theories without observers.''

\label{sec:motivation}

Quantum theory is conventionally formulated as a \emph{positivistic theory}, i.e., as a set of rules predicting what an observer will see when performing an experiment (more specifically, predicting which are the possible outcomes of the experiment, and which are their probabilities), also called the \emph{quantum formalism}. Many physicists have felt it desirable to formulate quantum theory instead as a \emph{realistic theory}, i.e., one describing (a model of) reality, independently of the presence of observers; in other words, describing all events that actually happen. This idea was most prominently advocated by Einstein \cite{Ein49}, Bell \cite{Bell90}, Schr\"odinger \cite{Schr35}, de Broglie \cite{deB}, Bohm \cite{Bohm52}, and Popper \cite{Pop}. Realistic theories have come to be known as \emph{quantum theories without observers} (QTWO) \cite{Gol98}. Since in a QTWO also the observer and her experiments are contained as special cases of matter and events, the quantum formalism remains valid but is \emph{a theorem and not an axiom}, that is, a consequence of the QTWO and not its basic postulate. Conversely, a QTWO provides an \emph{explanation} of the quantum formalism, describing how and why the outcomes specified by the formalism come about with their respective probabilities.

There are two examples of QTWO that work in a satisfactory way (as pointed out by, e.g., Bell \cite{Bell86b}, Goldstein \cite{Gol98}, and Putnam \cite{putnam}): Bohmian mechanics \cite{Bohm52,Bell66,survey} and GRW theory \cite{GRW86,Bell87,BG03}, as well as variants of these two theories. (It may or may not be possible that also other approaches, such as the ``many worlds'' view or the ``decoherent histories'' program, can be developed into satisfactory QTWOs \cite{Gol98, AGTZ06}.) 

Among the variants of GRW theory (i.e., among the mathematical theories of spontaneous wave function collapse besides the original GRW model), the most notable is the continuous spontaneous localization (CSL) theory of Pearle \cite{Pe89}; similar models were considered by Di\'osi \cite{Dio88}, Belavkin \cite{Be89}, Gisin \cite{Gi89}, and Ghirardi, Pearle, and Rimini \cite{GPR90}. Aside from explicit mathematical models, the idea that the Schr\"odinger equation might have to be replaced by a nonlinear and stochastic evolution has also been advocated by such distinguished theoretical physicists as Penrose \cite{Pen04} and Leggett \cite{leggett}.

\subsection{A Philosophical Aspect}\label{sec:phil}

A crucial part of QTWOs is the so-called \emph{primitive ontology} \cite{AGTZ06}. This means variables describing the distribution of matter in space and time. Here are four examples of primitive ontology:
\begin{itemize}
\item \textit{Particle ontology.} Matter consists of point particles, mathematically represented by a location $Q_t$ in physical 3-space for every time $t$, or, equivalently, by a curve in space-time called the particle's world line. This is the primitive ontology both of Bohmian and classical mechanics. One should imagine that each electron or quark is one point, so that a macroscopic object consists of more than $10^{23}$ particles.
\item \textit{String ontology.} Matter consists of strings, mathematically described by a curve in physical 3-space (or possibly another dimension of physical space), or, equivalently, by a 2-surface in space-time called the world sheet. One should imagine that each electron consists of one or more strings.
\item \textit{Flash ontology.} Matter consists of discrete points in space-time, called world points or flashes. One should imagine that a solid object consists of more than $10^6$ flashes per cubic centimeter per second. More flashes means more matter.
\item \textit{Matter density ontology.} Matter is continuously distributed in space, mathematically described by a density function $m(q,t)$, where $q$ is the location in physical 3-space and $t$ the time.
\end{itemize}
A QTWO needs a primitive ontology to give physical meaning to the mathematical objects considered by the theory \cite{AGTZ06,Mau07}. The role of the wave function then is ``to tell the matter how to move'' \cite{AGTZ06}, that is, to govern the primitive ontology (in a stochastic way). The theory we are mainly considering here, rGRWf, uses the flash ontology, which was first proposed for the original (non-relativistic) GRW model by Bell \cite{Bell87} and adopted in \cite{kent, Gol98, Tum05}. 

\begin{figure}[ht]
\begin{center}
\includegraphics[width=.4 \textwidth]{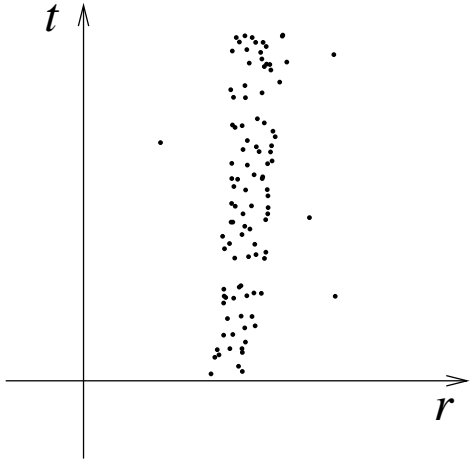}
\caption{A typical pattern of flashes in space-time ($r$ = space, $t$ = time), and thus a possible world according to the GRW theory with the flash ontology.}
\end{center}
\label{flashes}
\end{figure}

Interestingly, the (non-relativistic) GRW evolution of the wave function can reasonably be combined with the matter density ontology as well \cite{BGG95, Gol98, AGTZ06}; thus, there are two different GRW theories, called GRWm and GRWf, with the same wave function but different ontologies \cite{AGTZ06}. However, it is not known how GRWm could be made relativistic.

Likewise, it is not known how Bohmian mechanics could be made relativistic. More precisely, there does exist a natural and convincing way of defining Bohmian world lines on a relativistic space-time \cite{HBD,Tum06b}, but it presupposes the existence of a preferred slicing of space-time into spacelike 3-surfaces, called the \emph{time foliation}. The time foliation may itself be given by a Lorentz-invariant law, but still it seems against the spirit of relativity because it defines a notion of absolute simultaneity. This does not mean that this theory is wrong; it means that if it is right then we will have to adopt a different understanding of relativity. I have given an overview of recent research about Bohmian mechanics and relativity in \cite[Section~3.3]{Tum06b}.

We introduce some notation. Throughout this work, $\Hilbert$ will always be a separable complex Hilbert space. The adjoint of an operator $T$ on $\Hilbert$ is denoted $T^*$. The Borel $\sigma$-algebra of a topological space $X$ will be denoted $\Borel(X)$.

\section{Scheme of GRW Theories with Flash Ontology}
\label{sec:scheme}

This chapter is of a physical character. It provides an overview of GRW theories with flash ontology (hereafter, GRWf theories). The mathematical considerations in this chapter are not intended to be rigorous (unless when stated otherwise). For example, we will pretend that functions are differentiable or operators invertible whenever that is useful.

I describe a general scheme of GRWf theories (including, but more general than, the scheme described in \cite{Tum05}). We begin with a simple special case and increase generality step by step, finally arriving at the general version that contains also rGRWf. Given the scheme, a particular GRWf theory can be defined by specifying certain operators. This situation is roughly analogous to the general Schr\"odinger equation
\begin{equation}\label{Schr}
  i\hbar\frac{\D\psi_t}{\D t} = H \psi_t\,,
\end{equation}
which becomes a concrete evolution equation only after specifying the self-adjoint operator $H$, called the \emph{Hamiltonian}.

\subsection{The Simplest Case of GRWf}
\label{sec:simpleGRWf}

We take physical space to be $\RRR^3$ and the time axis to be $\RRR$. To specify the probabilistic law for the flash process, we specify the \emph{rate density} $r(q,t)$ at time $t \in\RRR$ for a flash to occur at $q\in\RRR^3$, which means, roughly speaking, that the probability of a flash in an infinitesimal volume $\D q$ around $q$ between $t$ and $t+\D t$, conditional on the flashes in the past of $t$, equals $r(q,t) \, \D q \, \D t$. The first basic equation of GRWf says that the flash rate density is given by
\begin{equation}\label{simpleflashrate}
  r(q,t) = \scp{\psi_t}{\Lambda(q)\, \psi_t}\,.
\end{equation}
Here $\Lambda(q)$ is a positive operator, called the \emph{flash rate operator}, which must be specified to define the theory, and $\psi_t \in \Hilbert$ is called the \emph{wave function} or \emph{state vector} at time $t$, which fulfills $\|\psi_t\|=1$ and evolves according to the following two evolution laws. When a flash occurs at time $T$ and location $Q$, the wave function changes discontinuously according to the second basic equation,
\begin{equation}\label{simplecollapse}
  \psi_{T+} = \frac{\Lambda(Q)^{1/2} \psi_{T-}}
  {\|\Lambda(Q)^{1/2} \psi_{T-}\|}\,.
\end{equation}
Here, $\psi_{T+} = \lim_{t\searrow T} \psi_t$ and $\psi_{T-} = \lim_{t\nearrow T} \psi_t$. 
This is called the \emph{collapse of the state vector at time $T$ and location $Q$}.
Between the flashes, the wave function evolves according to the Schr\"odinger equation \eqref{Schr}.

Once the operators $H$ and $\Lambda(q)$ are specified, the equations are intended to define the flash process
\begin{equation}\label{Fdef}
  F= \Bigl( (T_1,Q_1), (T_2,Q_2), \ldots\Bigr)\,,
\end{equation}
as follows: Choose, at an ``initial time'' $t_0$ the initial state vector $\psi_{t_0} \in\Hilbert$ with $\|\psi_{t_0}\|=1$, and evolve it using the Schr\"odinger equation \eqref{Schr} up to the time $T_1>t_0$ at which the first flash occurs; let $Q_1$ be the location of the first flash; collapse the state vector at time $T_1$ and location $Q_1$; continue with the collapsed state vector. (In the more general variants of the GRWf scheme, it can happen that the sequence $F$ ends after finitely many flashes if the rate is very low; in the simple variant we are presently considering, this does not happen, as we will see.)

\begin{ex}
The original 1986 GRW model \cite{GRW86,Bell87} is designed for non-relativistic quantum mechanics of $N$ particles; for $N=1$ it fits the above scheme as follows: $\Hilbert = L^2(\RRR^3)$; $H$ is the usual Hamiltonian of non-relativistic quantum mechanics, a self-adjoint extension of
\begin{equation}\label{HSchr1}
  H\psi = -\frac{\hbar^2}{2m} \nabla^2 \psi + V \psi
\end{equation}
for $\psi\in C_0^\infty(\RRR^3)$, where $m$ is the particle's mass and $V$ the potential; finally, the flash rate operators are multiplication operators by a Gaussian,
\begin{equation}\label{simpleGRWLambda}
  \Lambda(q) \, \psi(r) = \frac{\lambda}{(2\pi\sigma^2)^{3/2}}
  \E^{-(r-q)^2/2\sigma^2} \psi(r)\,,
\end{equation}
where $\lambda$ and $\sigma$ are new constants of nature with suggested values $\lambda \approx 10^{-15} \, \text{s}^{-1}$ and $\sigma \approx 10^{-7}\, \text{m}$. Since
\begin{equation}\label{LambdaR3I}
  \int_{\RRR^3} \Lambda(q) \, \D q = \lambda \, I\,,
\end{equation}
where $I$ is the identity operator on $\Hilbert$, the total flash rate
\begin{equation}
  r(\RRR^3,t) = \int_{\RRR^3} r(q,t) \, \D q = \lambda
\end{equation}
is independent of the state vector and constant in time. Thus, the flash times $T_1,T_2,\ldots$ form a Poisson process with intensity $\lambda$ (while the locations $Q_1,Q_2,\ldots$ do depend on $\psi$).
\end{ex}

\begin{ex}
A version of the GRW model advocated by Dove and Squires \cite{DS95} and myself \cite{Tum05} corresponding to non-relativistic quantum mechanics of $N$ \emph{identical} particles fits into the scheme as follows: $\Hilbert=S_\pm L^2(\RRR^{3})^{\otimes N}$ with $S_+$ the symmetrizer and $S_-$ the anti-symmetrizer, i.e., $\Hilbert$ is the space of symmetric (for bosons) respectively anti-symmetric (for fermions) $L^2$ functions on $\RRR^{3N}$; $H$ is the usual Hamiltonian, a self-adjoint extension of
\begin{equation}\label{Hidentical}
  H\psi = -\sum_{i=1}^N \frac{\hbar^2}{2m} \nabla_i^2\psi + V\psi\,,
\end{equation}
for $\psi\in C_0^\infty(\RRR^{3N}) \cap \Hilbert$; finally, the flash rate operators are
\begin{equation}\label{identicalGRWLambda}
  \Lambda(q) \, \psi(r_1,\ldots, r_N) = \frac{\lambda}{(2\pi\sigma^2)^{3/2}}
  \sum_{i=1}^N \E^{-(r_i-q)^2/2\sigma^2} \psi(r_1,\ldots, r_N)
\end{equation}
with the same constants as before. Then \eqref{LambdaR3I} holds with $N\lambda$ instead of $\lambda$, and hence the total flash rate is larger by a factor $N$,
\be\label{Nlambda}
r(\RRR^3,t) = N\lambda\,.
\ee
\end{ex}

The condition \eqref{LambdaR3I} plays a role to ensure the important property that \textit{the distribution of $F$ is given by a POVM}, i.e., there is a POVM (positive-operator-valued measure, see Section~\ref{sec:POVM}) $\povm(\cdot)$ on the history space $\Omega = (\RRR^4)^\NNN$, called the \emph{history POVM}, 
such that for $A\subseteq \Omega$
\begin{equation}\label{FPOVM}
  \PPP(F\in A) = \scp{\psi}{\povm(A) \, \psi}
\end{equation}
with $\psi = \psi_{t_0}$ the initial state vector.
(A physical consequence of this property is the impossibility of superluminal communication by means of entanglement.) We can come close to an explicit expression for the history POVM $\povm(\cdot)$ by providing an explicit expression for its marginal $\povm_n(\cdot)$ for the first $n$ flashes, which we obtain by a formal calculation \cite{Tum05} from \eqref{Schr}, \eqref{simpleflashrate}, \eqref{simplecollapse} and \eqref{LambdaR3I}, writing $X$ for the space-time point $(Q,T)$ (and $x=(q,t)$, $\D x = \D q\, \D t$):
\begin{align}\label{FPOVMn}
  \PPP(X_1 \in \D x_1,\ldots, X_n \in \D x_n) &=
  \scp{\psi}{\povm_n(\D x_1 \times \cdots \times \D x_n) \,\psi} =\\
  & =\scp{\psi}{L_n^* L_n \,\psi} \, \D x_1 \cdots \D x_n
\end{align}
with 
\begin{multline}\label{Lndef}
  L_n(x_1,\ldots,x_n) = 1_{t_0<t_1< \ldots < t_n} \,\E^{-\lambda(t_n-t_0)/2} \:\times \\
  \times \: \Lambda(q_n)^{1/2} \, \E^{-\I H(t_n-t_{n-1})/\hbar} 
  \Lambda(q_{n-1})^{1/2} \, \E^{-\I H(t_{n-1}-t_{n-2})/\hbar}
  \cdots \Lambda(q_1)^{1/2} \, \E^{-\I H(t_1-t_0)/\hbar} \,.
\end{multline}
Here, $1_{t_0<t_1< \ldots < t_n}$ means the characteristic function of the set $\{(x_1,\ldots,x_n) \in (\RRR^4)^n: x_k = (q_k,t_k), t_0<t_1< \ldots < t_n\}$.
These formulas we use for the rigorous definition of the GRWf flash process in Section~\ref{sec:simplestr}.

\subsection{Labeled Flashes}
\label{sec:labeled}

In some models we want the flashes to wear \emph{labels}, expressed by a mapping $\{X_1,X_2,\ldots\} \to \Lab$, where $\Lab$ is the set of all possible labels, which must be specified to define the theory. One can think of these labels as different types of flashes; for example,  electron flashes might be metaphysically different from muon flashes. (I like to imagine this situation as flashes of different color.) The set $\Lab$ can be finite or infinite. The flash process with labels can also be thought of as a point process in $\RRR^4 \times \Lab$ instead of $\RRR^4$. 

We write $I_n$ for the label of the $n$-th flash and $Z_n$ for the pair $(X_n,I_n)$ or the triple $(Q_n,T_n,I_n)$; thus $F= \bigl( Z_1,Z_2, \ldots \bigr)$. Similarly, we write $i$ for an element of $\Lab$, $z$ for the pair $(x,i)\in \RRR^4 \times \Lab$ or the triple $(q,t,i)$, and $\D z$ for a volume element in $\RRR^4 \times \Lab$, i.e., $\D z = \D x \times \{i\}$. We also write $F_n$ for the first $n$ flashes, $F_n= \bigl(Z_1,\ldots,Z_n\bigr)$, similarly $f_n=(z_1,\ldots,z_n)$, and $\D f_n = \D z_1\cdots \D z_n$.

To adapt the defining equations of GRWf, the rate density of flashes of type $i\in\Lab$ is
\begin{equation}\label{labelflashrate}
  r_i(q,t) = \scp{\psi_t}{\Lambda_i(q)\, \psi_t}\,,
\end{equation}
which means we assume separate flash rate operators for every type; the new collapse rule prescribes that if a flash of type $I$ occurs at time $T$ and location $Q$ then
\begin{equation}\label{labelcollapse}
  \psi_{T+} = \frac{\Lambda_I(Q)^{1/2} \psi_{T-}}
  {\|\Lambda_I(Q)^{1/2} \psi_{T-}\|}\,.
\end{equation}
Concerning the total flash rate operator, we assume, instead of \eqref{LambdaR3I},
\begin{equation}\label{labelLambdaR3I}
  \sum_{i\in\Lab} \int_{\RRR^3} \Lambda_i(q)\, \D q = \lambda I\,.
\end{equation}
(It should not lead to confusion that a capital $I$ is sometimes used for the identity operator and sometimes for a random label.)

As a consequence, the joint distribution of the first $n$ flashes together with their labels is
\begin{equation}\label{labelFPOVMn}
 \PPP(F_n \in \D f_n) = 
  \PPP(X_1 \in \D x_1, I_1 = i_1, \ldots, X_n \in \D x_n, I_n = i_n) =
  \scp{\psi}{L_n^* L_n \,\psi} \, \D x_1 \cdots \D x_n
\end{equation}
with 
\begin{multline}\label{labelLndef}
  L_n = L_n(x_1,i_1,\ldots, x_n,i_n) 
  = 1_{t_0<t_1< \ldots < t_n} \E^{-\lambda(t_n-t_0)/2} \:\times \\
  \times \: \Lambda_{i_n}(q_n)^{1/2} \, \E^{-\I H(t_n-t_{n-1})/\hbar} 
  \Lambda_{i_{n-1}}(q_{n-1})^{1/2} \, \E^{-\I H(t_{n-1}-t_{n-2})/\hbar}
  \cdots \Lambda_{i_1}(q_1)^{1/2} \, \E^{-\I H(t_1-t_0)/\hbar} \,.
\end{multline}

\begin{ex}
The original GRW model (corresponding to non-relativistic quantum mechanics of $N$ distinguishable particles) fits this scheme as follows: $\Hilbert= L^2(\RRR^{3N})$; $\Lab = \{1,\ldots,N\}$; $H$ is the usual Hamiltonian of non-relativistic quantum mechanics, a self-adjoint extension of
\begin{equation}\label{HSchrN}
  H\psi = -\sum_{i=1}^N \frac{\hbar^2}{2m_i} \nabla_i^2 \psi + V\psi
\end{equation}
for $\psi\in C_0^\infty(\RRR^{3N})$, where $m_i$ is the mass of particle $i$; finally, the flash rate operators are
\begin{equation}\label{GRWLambda}
  \Lambda_i(q) \, \psi(r_1,\ldots, r_N) = \frac{\lambda_i}{(2\pi\sigma^2)^{3/2}}
  \E^{-(r_i-q)^2/2\sigma^2} \psi(r_1,\ldots, r_N)\,.
\end{equation}
One easily checks \eqref{labelLambdaR3I} with $\lambda = \lambda_1 + \ldots + \lambda_N$.
\end{ex}

\subsection{Variable Total Flash Rate}

We now stop assuming that the total flash rate operator $\sum_i \int \D q\, \Lambda_i(q)$ is a multiple of the identity; that is, we drop \eqref{LambdaR3I} and \eqref{labelLambdaR3I}. As pointed out in \cite{Tum05}, this situation naturally arises for a GRWf process appropriate for quantum field theory (corresponding to a variable number of particles), as already suggested by the fact, see \eqref{Nlambda}, that the total flash rate is proportional to the number of particles. A stochastic wave function evolution very similar to the one discussed here was proposed by Blanchard and Jadczyk \cite{BJ95} as a model of a quantum system interacting with a classical one; see \cite{Jad06} for a discussion of the commonalities.

We can keep the same formulas, \eqref{labelflashrate} for the flash rate and \eqref{labelcollapse} for the collapse, but need to modify the Schr\"odinger equation \cite{Tum05,BJ95}:
\begin{equation}\label{totalevol}
 \I\hbar \frac{\D\psi_t}{\D t} = H\psi_t - \tfrac{\I\hbar}{2} \Lambda(\RRR^3) \, \psi_t 
 + \tfrac{\I\hbar}{2} \scp{\psi_t}{\Lambda(\RRR^3)\,\psi_t} \, \psi_t\,,
\end{equation}
where
\begin{equation}\label{totalLambdaint}
  \Lambda(\RRR^3) = \sum_{i\in\Lab} \int_{\RRR^3} \Lambda_i(q)\, \D q \,.
\end{equation}
Note that if, as assumed so far, $\Lambda(\RRR^3) = \lambda I$ then \eqref{totalevol} reduces to the Schr\"odinger equation \eqref{Schr}. Note further that \eqref{totalevol} formally preserves $\|\psi_t\|$ if $\|\psi\|=1$ initially:
\be
  \frac{\D}{\D t} \|\psi_t\|^2 = 2 \Re \Bscp{\psi_t}{\frac{\D\psi_t}{\D t}} =
\ee
\be
= 2\Re \Bigl(-\tfrac{\I}{\hbar} \scp{\psi_t}{H\psi_t} - 
  \tfrac{1}{2} \scp{\psi_t}{\Lambda(\RRR^3) \, \psi_t} +
  \tfrac{1}{2} \scp{\psi_t}{\Lambda(\RRR^3) \, \psi_t} \scp{\psi_t}{\psi_t} \Bigr) =
\ee
\begin{equation}
  = (\|\psi_t\|^2-1) \scp{\psi_t}{\Lambda(\RRR^3) \, \psi_t} = 0\,.
\end{equation}

Next, we want to obtain formulas analogous to \eqref{labelFPOVMn} and \eqref{labelLndef} for the distribution of the first $n$ flashes from the flash rate \eqref{labelflashrate}, the collapse law \eqref{labelcollapse} and the between-flashes evolution \eqref{totalevol}. However, since the total flash rate is not constant any more, it need not be bounded from below, and, as a consequence, it can have positive probability that only finitely many flashes occur. The appropriate history space space is therefore not $(\RRR^4)^\NNN$ but the space of all \emph{finite or infinite} sequences,
\begin{equation}\label{totalOmega}
  \Omega = \bigcup_{m=0}^\infty (\RRR^4)^m \cup (\RRR^4)^\NNN
\end{equation}
(where $(\RRR^4)^0$ has a single element, the empty sequence). Another method of representing finite-or-infinite sequences is based on introducing a formal symbol $\ceme$ (``cemetery state'') and writing a finite sequence in $\RRR^4$, such as $(x_1,\ldots,x_n)$, as the infinite sequence $(x_1,\ldots,x_n, \ceme, \ceme, \ldots)$ in $\RRR^4 \cup \{\ceme\}$. Then $\Omega$ can be understood as
\be\label{cemeOmega}
\Omega = \Bigl\{(z_1,z_2,\ldots) \in (\RRR^4 \cup\{\ceme\})^\NNN: z_n =\ceme \Rightarrow z_{n+1} = \ceme \Bigr\}\,,
\ee
the set of sequences for which $\ceme$ is ``absorbing'', i.e., the sequence cannot leave the cemetery state once it is reached. In this representation, the \emph{number of flashes} $\# F$ in a sequence $F=(z_1,z_2,\ldots) \in (\RRR^4 \cup\{\ceme\})^\NNN$ has to be defined as 
\be
\# F = \inf \{n\in\NNN: z_n = \ceme\}-1
\ee
(with the understanding $\inf \emptyset = \infty$).

By a formal computation, one obtains the following expression for the probability of existence of $n$ flashes and their joint distribution:
\begin{equation}\label{totalFPOVMn}
  \PPP(\# F \geq n, Z_1 \in \D z_1,\ldots, Z_n \in \D z_n) =
  \scp{\psi}{L_n^* L_n \,\psi} \, \D x_1 \cdots \D x_n
\end{equation}
with 
\begin{multline}\label{totalLndef}
  L_n = L_n(x_1,i_1,\ldots, x_n,i_n)  =  \\
  \Lambda_{i_n}(q_n)^{1/2} \, W_{t_n-t_{n-1}} \, 
  \Lambda_{i_{n-1}}(q_{n-1})^{1/2} \, W_{t_{n-1}-t_{n-2}}
  \cdots \Lambda_{i_1}(q_1)^{1/2} \, W_{t_1-t_0} \,,
\end{multline}
where
\begin{equation}\label{Wdef}
  W_t = \E^{-\frac{1}{2} \Lambda(\RRR^3)t - \frac{\I}{\hbar} Ht} 
  \text{ for }t\geq 0\,, \quad W_t = 0 \text{ for } t<0 \,.
\end{equation}

Another formal computation yields the following probability that the process stops after $n$ flashes:
\be
  \PPP(\# F = n, Z_1 \in \D z_1,\ldots, Z_n \in \D z_n) =
  \scp{\psi}{L_n^* \bigl(\lim_{t\to\infty} W_t^* W_t  \bigr) L_n \,\psi} 
  \, \D x_1 \cdots \D x_n\,.
\ee

\begin{ex}\label{ex:qft}
The following version of GRWf corresponding to a non-relativistic quantum field theory (i.e., quantum mechanics with a variable number of particles) I have proposed in \cite{Tum05}. The labels correspond to different particle species (electron, quark, \ldots); 
$\Hilbert$ is a product of spaces corresponding to the particle species,
\be\label{qftH}
\Hilbert = \bigotimes_{i\in\Lab} \Hilbert_i\,,
\ee
where $\Hilbert_i$ is a copy of the (bosonic or fermionic) Fock space, i.e.,
\begin{equation}\label{qftHi}
  \Hilbert_i = \bigoplus_{N=0}^\infty \Hilbert_i^{(N)} = 
  \bigoplus_{N=0}^\infty S_\pm L^2(\RRR^{3},\CCC^{2s_i+1})^{\otimes N}
\end{equation}
with $s_i\in\{0,\tfrac12,1,\tfrac32,\ldots\}$ the spin of species $i$;
a typical Hamiltonian consists of a contribution like \eqref{Hidentical} on every $\Hilbert_i^{(N)}$ plus contributions creating and annihilating particles (see, e.g., \cite{CDT05} for concrete examples); finally, the flash rate operators are given by \eqref{identicalGRWLambda} on every $\Hilbert_i^{(N)}$. As a consequence,
\be
\int \Lambda_i(q) \, \D q = \lambda \hat{N}_i\,,
\ee
where $\hat{N}_i$ is the particle number operator of species $i$,
\be
\hat{N}_i\psi = N_i\psi \quad\text{for}\quad \psi\in\Hilbert_i^{(N)}\otimes \bigotimes_{i'\neq i} \Hilbert_{i'}\,,
\ee
which is unbounded. Indeed, $\Lambda_i(q)$ is nothing but the particle number density operator $\hat{N}_i(q)$ of species $i$ (which actually is an operator-valued distribution) convolved with the Gaussian of width $\sigma$. Conversely, $\Lambda_i(q)$ could be \emph{defined} as $\hat{N}_i(q)$ convolved with the Gaussian of width $\sigma$, also if a given quantum field theory is not of the structure \eqref{qftH}--\eqref{qftHi}.
\end{ex}

\subsection{Time-Dependent Operators}

Suppose now that the relevant operators are explicitly time dependent: the Hamiltonian $H(t)$ and the flash rate operators $\Lambda_i(q,t)$. 

It is straightforward to adapt the basic equations of GRWf to this situation. We rewrite the flash rate density as
\begin{equation}\label{timeflashrate}
  r_i(q,t) = \scp{\psi_t}{\Lambda_i(q,t)\, \psi_t}\,,
\end{equation}
the collapse law as
\begin{equation}\label{timecollapse}
  \psi_{T+} = \frac{\Lambda_I(Q,T)^{1/2} \psi_{T-}}
  {\|\Lambda_I(Q,T)^{1/2} \psi_{T-}\|}\,, 
\end{equation}
and the between-flashes evolution law as
\begin{equation}\label{timeevol}
  \I\hbar \frac{\D\psi_t}{\D t} = H(t)\psi_t - \tfrac{\I\hbar}{2} \Lambda(\RRR^3,t) \, \psi_t 
  + \tfrac{\I\hbar}{2} \scp{\psi_t}{\Lambda(\RRR^3,t)\,\psi_t} \, \psi_t\,,
\end{equation}
where
\begin{equation}\label{timeLambdaint}
  \Lambda(\RRR^3,t) = \sum_{i\in\Lab} \int_{\RRR^3} \Lambda_i(q,t)\, \D q \,.
\end{equation}

These equations reduce to \eqref{labelflashrate}, \eqref{labelcollapse}, and \eqref{totalevol} if $H(t)$ and $\Lambda_i(q,t)$ are constant as functions of $t$. We also write $\Lambda(z)$ instead of $\Lambda_i(q,t)$, where $z=(q,t,i)$ is a labeled flash. From the above equations, we obtain by a formal computation in analogy to \eqref{totalFPOVMn} that
\begin{equation}\label{totalFPOVMnt}
  \PPP(\# F \geq n, Z_1 \in \D z_1, \ldots, Z_n \in \D z_n) =
  \scp{\psi}{L_n^* L_n \,\psi} \, \D x_1 \cdots \D x_n
\end{equation}
with 
\be\label{totalLndeft}
  L_n = L_n(z_1,\ldots, z_n)  =  
  \Lambda(z_n)^{1/2} \, W^{t_n}_{t_{n-1}} \, 
  \Lambda(z_{n-1})^{1/2} \, W^{t_{n-1}}_{t_{n-2}}
  \cdots \Lambda(z_1)^{1/2} \, W^{t_1}_{t_0} \,,
\ee
where $W_s^{t}$ is defined by
\begin{equation}\label{Wdeft}
  W_t^t = I\,, \quad \frac{\D W_s^{t}}{\D t} = 
  \Bigl(-\tfrac{1}{2} \Lambda(\RRR^3,t) - \tfrac{\I}{\hbar} H(t)\Bigr) W_s^t
\end{equation}
for $t\geq s$ and $W_s^t = 0$ for $t<s$. As a formal consequence of \eqref{Wdeft},
\be\label{W*Wt}
\frac{\D }{\D t} W_s^t{}^* W_s^t = - W_s^t{}^* \, \Lambda(\RRR^3,t) \, W_s^t
\ee
for $t\geq s$.

\subsubsection{``Gauge'' Freedom}
\label{sec:gauge1}

There remains a certain freedom in the choice of the operators $H(t)$ and $\Lambda(z)$ used to define the theory. A different choice $\tilde H(t)$ and $\tilde\Lambda(z)$ of Hamiltonian and flash rate operators can lead to the same history POVM $\povm(\cdot)$ as $H(t)$ and $\Lambda(z)$, and thus to the same probability distribution of the flashes. In this case, we regard $\tilde H(t)$ and $\tilde\Lambda(z)$ as \emph{physically equivalent} to $H(t)$ and $\Lambda(z)$, that is, as a different representation of the same physical theory. See \cite{AGTZ06} for a discussion of the concept of physical equivalence. 

For this conclusion it plays a role that we regard the flashes as the primitive ontology, and the theory as defined by defining the distribution of the flashes. Had we regarded the wave function as the primitive ontology, then a change in $H(t)$ would not have been admissible, as it leads to a different function for $\psi_t$. One can say that in GRWf theories we care about wave functions only insofar as we care about flashes, and that is why two wave functions, $\psi_t$ and $\tilde\psi_t$, arising from different choices of $H(t)$ and $\Lambda(z)$, can be regarded as just two representations of the same physical evolution, mathematically represented by the probability distribution $\PPP(\cdot) = \scp{\psi_0}{\povm(\cdot) \, \psi_0}$ of the flashes.

Note the similarity between the freedom about $H(t)$ and $\Lambda(z)$ and the gauge invariance of (classical) electrodynamics: Different choices of vector potentials $A_\mu$ are physically equivalent, i.e., different representations of the same reality, because we regard not $A_\mu$ but the field strength $F_{\mu\nu}$ as the primitive ontology. (Alternatively, one may regard only the particle trajectories as the primitive ontology, and these depend only on the field strength $F_{\mu\nu}$.) 

Returning to $H(t)$ and $\Lambda(z)$, here is a way of constructing $\tilde H(t)$ and $\tilde\Lambda(z)$ that lead to the same history POVM $\povm(\cdot)$. Let $U_s^t$ for $s,t\in \RRR$ be a family of unitary operators such that
\begin{equation}\label{Ucomposition}
U_t^t=I\,,\quad U_s^t\, U_r^s = U_r^t
\end{equation}
for all $r,s,t\in \RRR$. Let us assume $t_0=0$ for ease of notation. Note that all $U_s^t$ are determined by the subfamily $(U_0^t)_{t\in \RRR}$ because, by \eqref{Ucomposition}, $U_s^t = U_0^t (U_0^s)^{-1}$. Now set
\begin{equation}\label{tildeH}
\tilde H(t) = U_t^{0} \, H(t)\, U_{0}^t + \I\hbar \frac{\D U_t^0}{\D t} U_0^t
\end{equation}
and
\begin{equation}\label{tildeLambda}
\tilde\Lambda_i(q,t) = U_t^0 \, \Lambda_i(q,t) \, U_0^t\,.
\end{equation}
Then
\begin{equation}
\tilde\Lambda(\RRR^3,t) = U_t^0 \, \Lambda(\RRR^3,t)\, U_0^t\,,
\end{equation}
\begin{equation}\label{tildeW}
\tilde W_s^t = U_t^0\, W_s^t \, U_0^s\,,
\end{equation}
\begin{equation}\label{tildeL}
\tilde L_n = U_{t_n}^0 L_n\,,
\end{equation}
and thus
\begin{equation}
\tilde\povm_n(\cdot) = \povm_n(\cdot)\,,
\end{equation}
as we have claimed.

This can be understood in the following way. Imagine there is a separate Hilbert space $\Hilbert_t$ for every time $t$. Then there are many ways of identifying $\Hilbert_s$ with $\Hilbert_t$ for all $s,t\in\RRR$, each corresponding to a family of unitary isomorphisms $V_s^t: \Hilbert_s \to \Hilbert_t$ with $V_t^t= I$ and $V_s^t \, V_r^s=V_r^t$. Two such families $V,\hat V$ differ by a family of unitary operators $U_s^t = V_t^0 \,\hat V_s^t \, V_0^s$ on $\Hilbert_0$ satisfying \eqref{Ucomposition}, and thus, if we started with a tacit identification of all $\Hilbert_t$'s, every other way of identifying them is represented by a family $U_s^t$. Also in ordinary quantum mechanics different ways of identifying the $\Hilbert_t$'s are known: the Schr\"odinger picture and the Heisenberg picture. In the Heisenberg picture, $\Hilbert_s$ and $\Hilbert_t$ are identified along the unitary time evolution, so that $\psi_s$ and $\psi_t$ are identified as the same vector; in the Schr\"odinger picture, $\Hilbert_s$ and $\Hilbert_t$ are so identified that the position operators (represented by a projection-valued measure on $\RRR^3$) are time-independent. Also in GRWf, we can speak of a Schr\"odinger picture and a Heisenberg picture. In GRWf, a role similar to that of the position operators in ordinary quantum mechanics is played by the flash rate operators $\Lambda(q)$. If we assume them to be time-independent, as we did in Section~\ref{sec:simpleGRWf}, then this entails a particular way of identifying the $\Hilbert_t$'s; if we drop this assumption, as we do in this section, then other identifications are possible. 

\bigskip

\underline{\textit{Heisenberg picture:}}
The analog of the Heisenberg picture in GRWf
is characterized by the condition
\begin{equation}
\tilde H(t) = 0\,,
\end{equation}
so that the Hamiltonian contribution to the evolution of the state vector $\psi_t$ disappears. It can be obtained through the choice
\begin{equation}\label{UHeisenberg}
\frac{\D U_0^t}{\D t} =-\tfrac{\I}{\hbar} H(t) \, U_0^t\,.
\end{equation}
(Note, however, a fine conceptual difference from the Heisenberg picture in ordinary quantum mechanics: In ordinary quantum mechanics, it is the observables that evolve, while in GRWf, which is not defined in terms of observables, it is the flash rate operators.) In the Heisenberg picture, we have that
\begin{equation}\label{WHeisenberg}
\frac{\D W_s^t}{\D t} =  -\tfrac{1}{2} \Lambda(\RRR^3,t) \, W_s^t
\end{equation}
for $t\geq s$, with $W_t^t=I$ and $W_s^t=0$ for $t<s$. (One might be tempted to think that \eqref{WHeisenberg}, as it does not contain the skew-adjoint factor $\I H(t)$, implies that all $W_s^t$ are self-adjoint, but this is generically not the case; it is the case when all $\Lambda(\RRR^3,t)$ commute with each other.)

\bigskip

\underline{\textit{Square-root picture:}}
This ``gauge'' is characterized by the condition
\begin{equation}
\tilde W_0^t \geq 0\,.
\end{equation}
In fact, since in every ``gauge'' $\tilde W_0^{t*} \tilde W_0^t = W_0^{t*} W_0^t$ by \eqref{tildeW}, $\tilde W$ can be expressed through $W$ according to
\be\label{squareroot}
\tilde W_0^t = (W_0^{t*} W_0^t)^{1/2}\,.
\ee
This relation gives the ``square-root picture'' its name. 

This picture can be obtained through the choice\footnote{The expression \eqref{Utsqrt} is indeed rigorously defined and unitary if $W_0^t$ is bijective. To see that it is well-defined, note that if $W_0^t$ is bijective then so are $W_0^{t*}$ and $W_0^{t*} W_0^t$, and thus also $T:=(W_0^{t*} W_0^t)^{1/2}$ (as the bijectivity of $T^2$ implies that of $T$). In particular, $U_0^t$ is bijective as the product of the bijective operators $W_0^t$ and $T^{-1}$. To see that $U_0^t$ is unitary, note that its adjoint is its inverse, $U_0^{t*} U_0^t = T^{-1} T^2 T^{-1} = I$ (as $T$ and thus $T^{-1}$ are self-adjoint). Finally note that $U_0^0 = I$ by definition, and that \eqref{tildeW} yields \eqref{squareroot}.}
\begin{equation}\label{Utsqrt}
U_0^t = W_0^t \, (W_0^{t*} W_0^t)^{-1/2}\,.
\end{equation}

The advantage of the square-root picture is that $\tilde W_0^t$ can be easily computed by \eqref{squareroot} if $W_0^{t*} W_0^t$ is given. Since $\tilde W_s^t$ need not be positive (nor self-adjoint) for $s\neq 0$, the square-root picture only simplifies the rightmost term in \eqref{totalLndeft}. But this will be different in Section~\ref{sec:generalGRWf} when we allow flash rate operators to depend on previous flashes.

\subsection{General Scheme of GRWf Theories}
\label{sec:generalGRWf}

The scheme of GRWf we have developed so far is this: Given the operators $H(t)$ and $\Lambda_i(q,t)$, the corresponding GRWf theory is defined by \eqref{timeflashrate}--\eqref{timeLambdaint}. This scheme can be naturally generalized in two ways. 

\subsubsection{Nonpositive Collapse Operators}
First, instead of the positive operators $\Lambda_i(q,t)^{1/2}$ in the collapse law \eqref{timecollapse} we can put a collapse operator $C_i(q,t)$ which satisfies
\begin{equation}
C_i(q,t)^* \, C_i(q,t) = \Lambda_i(q,t)
\end{equation}
but is not necessarily positive, and not necessarily self-adjoint. That is, we replace \eqref{timecollapse} with
\begin{equation}\label{Ccollapse}
  \psi_{T+} = \frac{C_I(Q,T) \psi_{T-}}
  {\|C_I(Q,T) \psi_{T-}\|}\,. 
\end{equation}
We also write $C(z)$ instead of $C_i(q,t)$, where $z=(q,t,i)$ is a labeled flash.

\subsubsection{Past-Dependent Operators}
Second, we can allow that both the Hamiltonian $H$ and the collapse operator $C$ depend on the previous flashes,
\be
H=H(z_1,\ldots,z_n,t)=H(f_n,t) \,,
\ee
\be
C= C_i(z_1,\ldots,z_n,q,t) = C(f_n,z)\,.
\ee
We write $f_n := (z_1,\ldots,z_n)$ and $f_{n-1}=(z_1,\ldots,z_{n-1})$. Indeed, this situation occurs in the relativistic GRWf model, see Section~\ref{sec:defrGRWf}. 
The GRWf process can then be defined, instead of by \eqref{totalFPOVMnt}--\eqref{Wdeft}, by
\begin{equation}\label{totalFPOVMnp}
  \PPP(\# F \geq n, F_n \in \D f_n) =
  \scp{\psi}{L_n^* L_n \,\psi} \, \D x_1 \cdots \D x_n
\end{equation}
with 
\begin{equation}\label{totalLndefp}
  L_n = L_n(f_n)  =  
  C(f_{n}) \, W^{t_n}(f_{n-1}) \,
  L_{n-1}(f_{n-1}) \,,\quad
  L_0(\emptyset) = I\,,
\end{equation}
where $W^{t}(f_n)$ 
is defined by
\begin{equation}\label{Wdefp}
  W^{t_n}(f_n) = I\,, \quad \frac{\D W^{t}(f_n)}{\D t} = 
  \Bigl(-\tfrac{1}{2} \Lambda(f_n,\RRR^3,t) - \tfrac{\I}{\hbar} H(f_n,t)\Bigr)
  W^t(f_n)
\end{equation}
for $t\geq t_n$ and $W^t(f_n) = 0$ for $t<t_n$; here,
\begin{equation}\label{CLambdaint}
\Lambda(f_n,\RRR^3,t) = \sum_{i\in \Lab} \int_{\RRR^3} 
C_i(f_n,q,t)^* \, C_i(f_n,q,t) \, \D^3 q\,. 
\end{equation}
(It is unnecessary now to specify \emph{two} times for the $W$ operator, as in the notation $W_s^t$, because now $s=t_n$, where $t_n$ is the time of the last flash in $f_n$.)

\subsubsection{``Gauge'' Freedom Once More}
\label{sec:gauge2}

In addition to the gauge freedom described in Section~\ref{sec:gauge1}, there is another gauge freedom when the operators $H$ and $C$ can depend on the past flashes $f=(z_1,\ldots,z_n)$, and exploiting this freedom one can ensure that all $C$'s are positive operators.

Here is a way of constructing different operators $\tilde H(f,t)$ and $\tilde C(f,z)$ that lead to the same history POVM $\povm(\cdot)$. The construction is the same as in Section~\ref{sec:gauge1}, except that the unitaries $U_s^t$ are now allowed to depend on the past flashes $f$. That is, let $U_s^t(f)$ for $s,t\in \RRR, f\in \cup_{n=0}^\infty (\RRR^4\times\Lab)^n$ 
be an arbitrary family of unitary operators such that
\begin{equation}\label{Ucompositionp}
U_t^t(f) =I\,,\quad U_s^t(f)\, U_r^s(f) = U_r^t(f)
\end{equation}
for all $r,s,t\in \RRR$, and set (assuming $t_0=0$ for ease of notation)
\begin{equation}\label{tildeHp}
\tilde H(f,t) = U_t^{0}(f) \, H(f,t)\, U_{0}^t(f) + \I\hbar \frac{\D U_t^0(f)}{\D t} U_0^t(f)
\end{equation}
and
\begin{equation}\label{tildeCp}
\tilde C(f,z) = U_t^0(f,z) \, C(f,z) \, U_0^t(f)\,.
\end{equation}
It follows that 
\begin{equation}
\tilde\Lambda(f,z) = U_t^0(f) \, \Lambda(f,z) \, U_0^t(f) \,, 
\end{equation}
\begin{equation}
\tilde\Lambda(f,\RRR^3,t) = U_t^0(f) \, \Lambda(f,\RRR^3,t)\, U_0^t(f)\,,
\end{equation}
\begin{equation}\label{tildeWp}
\tilde W^t(f) = U_t^0(f)\, W^t(f) \, U_0^{t_n}(f)\,,
\end{equation}
\begin{equation}\label{tildeLp}
\tilde L_n(f) = U_{t_n}^0(f) \, L_n(f)\,,
\end{equation}
for $f=(z_1,\ldots,z_n)$, and thus
\begin{equation}
\tilde\povm_n(\cdot) = \povm_n(\cdot)\,,
\end{equation}
as we have claimed.

\bigskip

\underline{\textit{Heisenberg-plus picture:}} This is the special case 
characterized by the conditions
\begin{equation}\label{posHeisenberg}
\tilde H(f,t) = 0 \,, \quad \tilde C(f,z) \geq 0 \,,
\end{equation}
so that
\be
\tilde C(f,z) = \tilde \Lambda(f,z)^{1/2}\,.
\ee
(The tag ``plus'' indicates that the $\tilde C$ are positive). It can be obtained through the particular choice of $U_0^t(f)$ defined by
\begin{equation}\label{UpODE}
\frac{\D U_0^t(f)}{\D t} =-\tfrac{\I}{\hbar} H(t,f) \, U_0^t(f)
\end{equation}
for $t>t_n > \ldots> t_1>0$ and $f=(z_1,\ldots,z_n)$, $z_k=(q_k,t_k,i_k)$, with the initial condition $U_0^{t_n}(f)$ chosen so that
\begin{equation}\label{Upjump}
U_0^0(\emptyset) = I\,,\quad U^{t_n}_0(f)^* \, C(f) \, U_0^{t_n}(f_{n-1}) \geq 0 
\end{equation}
with $f_{n-1} = (z_1,\ldots,z_{n-1})$. Indeed, $U^{t_n}_0(f)$ is determined by \eqref{Upjump} from $C(f)$ and $U_0^{t_n}(f_{n-1})$, provided that $C(f):\Hilbert \to\Hilbert$ is bijective.\footnote{This follows from the fact that the operator $T:=C(f) \, U_0^{t_n}(f_{n-1})$, as it is bounded and bijective, possesses a unique polar decomposition \cite[Thm~12.35]{Rud}
 $T=UP$ 
as a product of a unitary $U$ and a bounded positive $P$. Now if $V$ is a unitary so that $VT$ is positive, then $VT = (T^*V^*VT)^{1/2}= (T^*T)^{1/2}=P$ and $V=U^*$. That is, $U^{t_n}_0(f) = U$.}

\bigskip

\underline{\textit{Square-root-plus picture:}}
This case is characterized by the conditions
\be
\tilde W^t(f) \geq 0 \,, \quad \tilde C(f,z) \geq 0
\ee
and can be obtained through the particular choice of $U_0^t(f)$ defined by two equations, \eqref{Upjump} and
\be\label{Upsqrt}
U_{t_n}^t(f) = W \, (W^*W)^{1/2}
\ee
with $W= W^t(f)$ and $n=\# f$.
From \eqref{Upjump} it follows by \eqref{tildeCp} that all $\tilde C$ are positive. From \eqref{Upsqrt}, as in the context of \eqref{Utsqrt}, it follows by \eqref{tildeWp} that
\be\label{squarerootp}
\tilde W^t(f) = U_0^{t_n}(f)^* \, (W^*W)^{1/2} \, U_0^{t_n}(f)\,,
\ee
the analog of \eqref{squareroot}. As a consequence, 
$\tilde W^t(f)$ is positive.

\subsubsection{Ways of Specifying the Theory}
\label{sec:specify}

The theory can be specified by specifying all the operators $H(f,t)$ and $C(f)$. If all the $C(f)$ are positive, one can instead specify the $\Lambda(f)$ (as then $C(f) = \Lambda(f)^{1/2}$). In the Heisenberg-plus picture, one has to specify \emph{only} the operators $\Lambda(f)$ since all $H(f,t)=0$. 

Let us return to the case of nonzero $H(f,t)$. One can specify, instead of $H(f,t)$, directly the 
$W^t(f)$ (in addition to the $C(f)$), provided they satisfy the following condition of consistency with the specified $C(f)$:
\be\label{Wcondition}
\frac{\D}{\D t} W^t(f)^* W^t(f) = - W^t(f)^* \Lambda(f,\RRR^3,t) \, W^t(f)\,.
\ee
To specify the $W$ instead of the $H$ operators is analogous to specifying, in ordinary quantum mechanics, the unitary time-evolution operators $U_s^t$ instead of the Hamiltonians $H(t)$.

The wave function $\psi_t$ at time $t$ can be expressed through the $W$ operators. It depends, of course, on the flashes $f$ between $t_0$ and $t$:
\be
\psi_t = \frac{W^t(f) \, L_{\# f}(f) \, \psi}{\bigl\|W^t(f) \, L_{\# f}(f) \, \psi\bigr\|}
\ee
with $L_n$ defined by \eqref{totalLndefp}.

\subsection{Flashes + POVM = GRWf}

We have described how theories of the GRWf type, specified in terms of the flash rate (and other) operators, give rise to a distribution of flashes given by a POVM. We now argue that essentially every theory with flash ontology in which the distribution of the flashes is given by a POVM, arises from the GRWf scheme for a suitable choice of flash rate operators, and is thus a collapse theory. (A rigorous discussion is provided in Section~\ref{sec:reconLambdar}.) In particular, this suggest that rGRWf can be expressed, in any coordinate system, in the GRWf scheme.

We assume that the history POVM $\povm(\cdot)$ is such that for every $n\in \NNN$ its marginal $\povm_n(\cdot)$ for the first $n$ flashes has a positive-operator-valued density function $E_n:(\RRR^4\times\Lab)^n \to \Bdd(\Hilbert)$ relative to the Lebesgue measure, i.e.,
\be
\povm_n(A) = \int_A E_n(f_n) \,\D f_n\,, 
\ee
where $f_n=(x_1,i_1,\ldots,x_n,i_n)$, $x_k \in \RRR^4$ and the notation $\D f_n$ means
\be
\int_A g(f_n) \,\D f_n = \sum_{i_1\ldots i_n \in \Lab} \int_{\RRR^{4N}} 1_{f_n\in A} \, g(f_n) \, \D x_1 \cdots \D x_n\,.
\ee
This assumption is fulfilled for the history POVM $\povm(\cdot)$ of GRWf with
\be
E_n(f_n) = L_n(f_n)^* \, L_n(f_n)\,.
\ee

\subsubsection{Reconstructing $\Lambda$}
\label{sec:reconLambda}

We now explain how to reconstruct the flash rate operators from the history POVM, i.e., how to extract $\Lambda(f)$ from the operator-valued functions $E_n$, first in the square-root-plus picture, and afterwards in the Heisenberg-plus picture. 

\bigskip

\underline{\textit{Square-root-plus picture:}} Set $L_0(\emptyset) = I$,
\be\label{Wreconempty}
W^t(f=\emptyset) = \povm(T_1>t)^{1/2}\,,
\ee
where
\be
\povm(T_1>t) =\sum_{i\in\Lab} \int_t^\infty \D s\int_{\RRR^3} \D q \, E_1(q,s,i)\,.
\ee
Now set
\be\label{Lambdareconempty}
\Lambda(f=\emptyset,q,t,i) = W^t(\emptyset)^{-1} \, E_1(q,t,i) \, W^t(\emptyset)^{-1}\,.
\ee
Continue inductively along the number $n$ of flashes, setting
\be\label{Lnrecon}
L_n(f_n) = \Lambda(f_n)^{1/2} \, W^{t_n}(f_{n-1}) \, L_{n-1}(f_{n-1})
\ee
and
\be\label{Wrecon}
W^t(f_n) = \Bigl(L_n^*(f_n)^{-1}\sum_{i\in\Lab} \int_{t}^\infty \D s\int_{\RRR^3} \D q \, E_{n+1}(f_n,q,s,i) \, L_n(f_n)^{-1}\Bigr)^{1/2}\,.
\ee
Now set
\be\label{Lambdarecon}
\Lambda(f_n,z) = W^t(f_n)^{-1} \, L_n^*(f_n)^{-1} \, E_{n+1}(f_n,z)\, 
L_n(f_n)^{-1} \, W^t(f_n)^{-1} \,.
\ee
It then follows formally that 
\be
 L_n^*(f_n) \, L_n(f_n) = E_n(f_n)\,.
\ee
Theorem~\ref{thm:recon} in Section~\ref{sec:reconLambdar} provides conditions under which this computation actually works.

\bigskip

\underline{\textit{Heisenberg-plus picture:}}  Set $L_0(\emptyset) = I$. 
We determine $\Lambda(f=\emptyset,q,t,i)$ for $t>0$ (and all $q\in\RRR^3,i\in\Lab$) by simultaneously solving
\begin{equation}
\frac{\D W^t(\emptyset)}{\D t} = -\tfrac{1}{2} \Lambda(\emptyset,\RRR^3,t) \, W^t(\emptyset) 
\end{equation}
with initial datum
\begin{equation}
W^{t_0}(\emptyset) = I\,,
\end{equation}
and
\begin{equation}
\Lambda_i(f=\emptyset,q,t) = W^t(\emptyset)^*{}^{-1}\,E_1(q,t,i) \, W^t(\emptyset)^{-1}\,.
\end{equation}
That is,
\begin{equation}
\frac{\D W^t(\emptyset)}{\D t} = -\tfrac{1}{2} W^t(\emptyset)^*{}^{-1}\int_{\RRR^3} \D q\, \sum_{i\in\Lab} E_1(q,t,i) \,.
\end{equation}

Now we proceed by induction along the number of flashes. Suppose that for all sequences of up to $n-1$ flashes, $f_{n-1}=(z_1,\ldots,z_{n-1})$, the operators $\Lambda_i(q,t,f_{n-1})$, 
$W^t(f_{n-1})$, and $L_{n-1}(f_{n-1})$ are known for all $i\in\Lab$, $q\in\RRR^3$, and $t\geq t_n$. For arbitrary $f_n=(z_1,\ldots,z_n)$, set
\begin{equation}
L_n(f_n) = \Lambda(f_n)^{1/2} \, W^{t_n}(f_{n-1}) \, L_{n-1}(f_{n-1})\,.
\end{equation}
Solve simultaneously
\begin{equation}
\frac{\D W^t(f_n)}{\D t} = -\tfrac{1}{2} \Lambda(f_n,\RRR^3,t) \, W^t(f_n) 
\end{equation}
with initial datum
\begin{equation}
W^{t_{n}}(f_n) = I\,,
\end{equation} 
and
\begin{equation}\label{LambdareconH}
\Lambda(f_n,z) = \Bigl(L_n^*(f_n) \, W^t(f_n)^*\Bigr)^{-1}\, E_{n+1}(f_n,z) \, \Bigl( W^t(f_n) \, L_n(f_n)\Bigr)^{-1}\,.
\end{equation}

Then \eqref{totalLndefp} and \eqref{Wdefp} are satisfied by construction, and $L_n^*(f_n) \, L_n(f_n) = E_n(f_n)$ by \eqref{LambdareconH}.

\section{Rigorous Treatment of the GRWf Scheme}
\label{sec:schemer}

In this chapter we repeat the considerations of Chapter 2 in a rigorous treatment; here we provide the exact conditions under which our constructions work and the point processes exist.

\subsection{Weak Integrals}

Let $\Bdd(\Hilbert)$ denote the space of bounded operators on the Hilbert space $\Hilbert$. We say that an operator-valued function $\Lambda:(M,\salg) \to \Bdd(\Hilbert)$ is \emph{weakly measurable} if for every $\psi\in\Hilbert$ the function $f_\psi: M \to \CCC$, defined by $f_\psi(q)=\scp{\psi}{\Lambda(q)\, \psi}$, is Borel measurable. In that case, also $q \mapsto \scp{\phi}{\Lambda(q) \, \psi}$ is Borel measurable because, using polarization,
\be\label{polarweakmeas}
\scp{\phi}{\Lambda(q) \, \psi} = \tfrac{1}{4} \Bigl( 
f_{\phi+\psi}(q) - f_{\phi-\psi}(q) -\I f_{\phi+\I\psi}(q) +\I f_{\phi-\I\psi}(q) \Bigr)\,.
\ee
Moreover, also the adjoint $q\mapsto \Lambda(q)^*$ is weakly measurable.

Let $\mu$ be a $\sigma$-finite measure on $(M,\salg)$. We understand the expression $\int \Lambda(q) \, \mu(\D q)$ as a \emph{weak integral} defined by
\begin{equation}\label{weakintegraldef}
  T=\int \Lambda(q)\, \mu(\D q) \quad :\Leftrightarrow \quad \forall \psi\in\Hilbert: \:
  \scp{\psi}{T \psi} = \int \scp{\psi}{\Lambda(q)\,\psi} \, \mu(\D q)\,.
\end{equation}
Throughout this paper, all integrals over operators are weak integrals.

(Another concept of integration of Banach-space-valued functions is the \emph{Bochner integral} \cite{Yosida}, which is not suitable for our purposes since relevant examples of flash rate operators $\Lambda(q)$, such as \eqref{simpleGRWLambda}, are weakly integrable but not Bochner integrable: Bochner integrability requires $\int \|\Lambda(q)\| \, \mu(\D q) < \infty$, while, for example, for the $\Lambda(q)$ of the original GRW model, given by \eqref{simpleGRWLambda}, $\|\Lambda(q)\|=\lambda/(2\pi\sigma^2)^{3/2}=\mathrm{const.}$ for all $q\in M=\RRR^3$, and $\mu$ is the Lebesgue measure, so that in fact $\int \|\Lambda(q) \| \, \mu(\D q) = \infty$.)

Note that $T$ need not exist (for example, $\Lambda(q) = I$ for all $q\in\RRR^3$), but if it exists then it is unique, as $T$ is determined by the values $\scp{\psi}{T\psi}$. Moreover, if $T$ exists then
\be\label{phiTpsi}
\scp{\phi}{T \, \psi} = \int \scp{\phi}{\Lambda (q) \, \psi} \, \mu(\D q)
\ee
by \eqref{polarweakmeas}; in particular, $q\mapsto \scp{\phi}{\Lambda(q)\, \psi}$ is (absolutely) integrable.
We can guarantee the existence of $T$ in a special case:

\begin{lemma}\label{lemma:weakint}
If $\Lambda:M \to \Bdd(\Hilbert)$ is weakly measurable and $\Lambda(q)$ is positive for every $q\in M$ then 
\be
S=\{\psi\in\Hilbert: \int \scp{\psi}{\Lambda(q)\,\psi} \, \mu(\D q) < \infty\}
\ee
is a subspace, and
\be\label{sesquiweakint}
B(\phi,\psi) = \int \scp{\phi}{\Lambda(q) \, \psi} \, \mu(\D q) \quad \forall \phi,\psi\in S
\ee
defines a positive Hermitian sesquilinear form on $S$. 
Moreover, if $B$ is bounded and $S=\Hilbert$ then, by the Riesz lemma, there is a positive operator $T \in\Bdd(\Hilbert)$ such that $B(\phi,\psi) = \scp{\phi}{T \psi}$.
\end{lemma}

\bigskip

\begin{proof}
If $P:\Hilbert \to\Hilbert$ is a positive operator then it is self-adjoint.\footnote{Because $\scp{\phi+\psi}{P(\phi+\psi)} = \scp{P(\phi+\psi)}{\phi+\psi}$ implies $\scp{\phi}{P\,\psi}+ \scp{\psi}{P \, \phi} = \scp{P \,\phi}{\psi} +\scp{P\, \psi}{\phi}$; call this equation 1; consider the same equation with $\I\psi$ instead of $\psi$, and call it equation 2; equation 1 minus $\I$ times equation 2 yields $\scp{\phi}{P\,\psi} = \scp{P \, \phi}{\psi}$.} Therefore, $P^{1/2}$ exists, and $\scp{\psi}{P \, \psi} = \|P^{1/2}\psi\|^2$. As a consequence, setting $P=\Lambda(q)$, if $\psi\in S$ then $q\mapsto \|\Lambda(q)^{1/2}\psi\|$ is a square-integrable function, and thus, if $\phi,\psi\in S$,
\be
\int \bigl|\scp{\phi}{\Lambda(q) \, \psi}\bigr| \, \mu(\D q) \leq
\int \|\Lambda(q)^{1/2}\phi\| \, \|\Lambda(q)^{1/2}\psi\| \, \mu(\D q) < \infty
\ee
by the Cauchy--Schwarz inequality in $\Hilbert$ and that in $L^2(M,\mu)$. This shows that $S$ is a subspace. 

For $\phi,\psi\in S$, set $Q(\psi)= \int \scp{\psi}{\Lambda(q)\,\psi} \, \mu(\D q)$ and 
\be
B(\phi,\psi) = \frac{1}{4} \Bigl(
 Q(\phi + \psi)-
 Q(\phi - \psi)+i
 Q(\phi -i \psi)-i
 Q(\phi +i \psi) \Bigr)
\ee
Then \eqref{sesquiweakint} follows; sesquilinearity and positivity (and, in particular, Hermitian symmetry) follow from \eqref{sesquiweakint}. 
\end{proof}

\begin{lemma}\label{lemma:denseint}
Let $S$ be a dense subspace of $\Hilbert$, $T\in\Bdd(\Hilbert)$, $\Lambda:M \to \Bdd(\Hilbert)$ weakly measurable and $\Lambda(q)$ positive for every $q\in M$. If the equation
\be\label{TintLambda}
\scp{\psi}{T\psi} = \int_M \scp{\psi}{\Lambda(q)\,\psi} \,\mu(\D q)
\ee 
is true for all $\psi\in S$ then it is true for all $\psi\in\Hilbert$. In other words, if \eqref{TintLambda} holds on $S$ then $T=\int \Lambda(q)\,\mu(\D q)$.
\end{lemma}

\begin{proof}
For arbitrary $\psi\in\Hilbert$, there is a sequence $(\psi_n)_{n\in\NNN}$ in $S$ with $\psi_n \to\psi$. Since $T$ is bounded, $\scp{\psi_n}{T\psi_n} \to \scp{\psi}{T\psi}$. What we have to show is
\be\label{zuzeigen}
\int_M \scp{\psi_n}{\Lambda(q)\,\psi_n} \,\mu(\D q) \to 
\int_M \scp{\psi}{\Lambda(q)\,\psi} \,\mu(\D q)\,.
\ee
For every $n\in\NNN$, define the function $f_n: M \to [0,\infty)$ by
\be
f_n(q) = \scp{\psi_n}{\Lambda(q) \, \psi_n}\,.
\ee
Let $[f_n]$ denote its equivalence class modulo changes on a $\mu$-null set. Since $\int f_n(q) \, \mu(\D q) = \scp{\psi_n}{T\psi_n}<\infty$, $[f_n] \in L^1(M,\mu)$. The sequence $([f_n])$ is a Cauchy sequence in $L^1$:
\be
\|f_n-f_m\|_1 = \int \Bigl|\scp{\psi_n-\psi_m+\psi_m}{\Lambda(q) (\psi_n-\psi_m+\psi_m)} - \scp{\psi_m}{\Lambda(q) \, \psi_m}  \Bigr|\,\mu(\D q)
\ee
\be
\leq \int \scp{\psi_n-\psi_m}{\Lambda(q) (\psi_n-\psi_m)}\, \mu(\D q) + \int 2\Bigl| \scp{\psi_n-\psi_m}{\Lambda(q) \, \psi_m}\Bigr|\, \mu(\D q)\leq
\ee
[using the Cauchy--Schwarz inequality for $\Hilbert$]
\be
\leq \scp{\psi_n-\psi_m}{T (\psi_n-\psi_m)}
 + \int 2\|\Lambda(q)^{1/2}(\psi_n-\psi_m)\|\,\|\Lambda(q)^{1/2} \psi_m\|\, \mu(\D q)
\ee
[using the Cauchy--Schwarz inequality for $L^2(M,\mu)$]
\be
\leq \|T\|\,\|\psi_n-\psi_m\|^2 + 2\Bigl( \int \|\Lambda(q)^{1/2}(\psi_n-\psi_m)\|^2 \mu(\D q)\Bigr)^{1/2} \Bigl(\int \|\Lambda(q)^{1/2} \psi_m\|^2\, \mu(\D q) \Bigr)^{1/2}
\ee
\be
= \|T\|\,\|\psi_n-\psi_m\|^2 + 2 \scp{\psi_n-\psi_m}{T(\psi_n-\psi_m)}^{1/2} 
\scp{\psi_m}{T\psi_m}^{1/2} 
\ee
\be
\leq \|T\|\,\|\psi_n-\psi_m\|^2 + 2 \|T\|^{1/2}\|\psi_n-\psi_m\|  \|T\|^{1/2}\|\psi_m\| 
\ee
\be
=\|T\| \Bigl(\|\psi_n-\psi_m\| + 2 \|\psi_m\| \Bigr) \|\psi_n-\psi_m\| \to 0 
\ee
as $n,m\to \infty$. Since $([f_n])$ is a Cauchy sequence in the Banach space $L^1(M,\mu)$, it converges, say $[f_n] \to [f]$, and $\scp{\psi_n}{T\psi_n} = \int f_n(q)\, \mu(\D q)\to \int f(q)\, \mu(\D q)$. 
On the other hand, since the $\Lambda(q)$ are bounded, the $f_n$ converge pointwise to $q\mapsto\scp{\psi}{\Lambda(q)\,\psi}$, and $f$ (the $L^1$ limit) must agree with the pointwise limit $\mu$-almost everywhere. Thus, $q\mapsto \scp{\psi}{\Lambda(q)\,\psi}$ is an $L^1$ function, and \eqref{zuzeigen} holds.
\end{proof}

\medskip

Below I collect some lemmas about weak measurability. Most of the following proofs I have learned from Reiner Sch\"atzle (T\"ubingen).

\begin{lemma}\label{lemma:matrixweakmeas}
Let $\{\phi_n: n\in\NNN\}$ be an orthonormal basis of $\Hilbert$. $q\mapsto \Lambda(q)$ is weakly measurable if and only if for all $n,m\in\NNN$, $q \mapsto \Lambda_{nm}(q) := \scp{\phi_n}{\Lambda(q) \, \phi_m}$ is measurable. 
\end{lemma}

\begin{proof}
The ``only if'' part is clear, and the ``if'' part follows from
\be
\scp{\phi}{\Lambda(q)\, \psi} = \sum_{n=1}^\infty \sum_{m=1}^\infty \scp{\phi}{\phi_n} \, \Lambda_{nm}(q) \, \scp{\phi_m}{\psi}\,,
\ee
where the series converges for every $q$, and the fact that the pointwise limit of measurable functions is measurable. 
\end{proof}

\begin{lemma}\label{lemma:Rint}
If $q\mapsto\Lambda(q)$ is weakly measurable and $R,S$, and $T=\int \Lambda(q) \, \mu(\D q)$ are bounded operators then $q\mapsto R\,\Lambda(q) \, S$ is weakly measurable, and
\be\label{Rint}
RTS = \int R\, \Lambda(q) \, S\, \mu(\D q)\,.
\ee
\end{lemma}

\begin{proof}
$q\mapsto R\, \Lambda(q)\, S$ is weakly measurable because, if $\{\phi_n:n\in\NNN\}$ is an orthonormal basis, $\scp{\phi_n}{R\, \Lambda(q) \,S\, \phi_m} = \sum_{k=1}^\infty \sum_{\ell=1}^\infty \scp{\phi_n}{R\,\phi_k} \scp{\phi_k}{\Lambda(q) \, \phi_\ell}\scp{\phi_\ell}{S\, \phi_m}$, as $R$ and $\Lambda(q)$ are bounded.

To check \eqref{Rint}, note that since $R$ is bounded, its adjoint $R^*$ is defined on all of $\Hilbert$ and is bounded too, so that $\scp{R^*\phi}{\Lambda(q) \,S\, \psi}$ exists for all $\phi,\psi$ and $q$, and is integrable by \eqref{phiTpsi} with $\phi$ replaced by $R^*\phi$ and $\psi$ by $S\,\psi$:
\be
\scp{R^*\phi}{T\,S\,\psi} =\int \scp{R^*\phi}{\Lambda(q) \, S\,\psi} \, \mu(\D q) = \int \scp{\phi}{R\,\Lambda(q)\,S \,\psi} \,\mu(\D q)\,.
\ee
The left hand side equals $\scp{\phi}{RTS\,\psi}$, and the right hand side $\scp{\phi}{\int R\Lambda(q)S\, \mu(\D q) \, \psi}$.
\end{proof}

\begin{lemma}\label{lemma:productweakmeas}
If $\Lambda,\Lambda': M \to \Bdd(\Hilbert)$ are both weakly measurable then so is their product, $q\mapsto \Lambda(q) \, \Lambda'(q)$.
\end{lemma}

\begin{proof}
For every $q$,
\be
\scp{\phi_n}{\Lambda(q)\, \Lambda'(q) \, \phi_m} = \sum_{\ell=1}^\infty \Lambda_{n\ell}(q) \, \Lambda_{\ell m}(q)
\ee
because $\Lambda(q)$ is bounded. The right hand side is a measurable function of $q$ because products, sums, and limits of measurable functions are measurable.
\end{proof}

\begin{lemma}\label{lemma:normweakmeas}
If $\Lambda: M \to \Bdd(\Hilbert)$ is weakly measurable and every $\Lambda(q)$ is self-adjoint then $q\mapsto \|\Lambda(q)\|$ is measurable.
\end{lemma}

\begin{proof}
For bounded, self-adjoint $T$, it is known \cite[Thm.~12.25]{Rud} that
\be
\|T\| = \sup_{\|\psi\|=1} \bigl|\scp{\psi}{T\psi}\bigr|\,.
\ee
Let $S$ be any countable dense subset of the unit sphere of $\Hilbert$. Then
\be
\sup_{\|\psi\|=1} \bigl|\scp{\psi}{T\psi}\bigr| = \sup_{\psi \in S} \bigl|\scp{\psi}{T\psi}\bigr|\,.
\ee
The $\geq$ relation is clear, and for the $\leq$ relation consider any $\psi\in\Hilbert$ with $\|\psi\|=1$ and note that there is a sequence $(\psi_m)\subseteq S$ with $\psi_m\to\psi$ and therefore, by the boundedness of $T$, $\scp{\psi_m}{T\psi_m} \to\scp{\psi}{T\psi}$; as a consequence, for every $\varepsilon>0$,
\be
\bigl|\scp{\psi}{T\psi}\bigr|-\varepsilon \leq \bigl|\scp{\psi_m}{T\psi_m}\bigr|
\ee
for sufficiently large $m$. Thus,
\be
\|\Lambda(q)\| = \sup_{\psi\in S} \bigl|\scp{\psi}{\Lambda(q)\,\psi}\bigr|\,,
\ee
and the supremum of countably many measurable functions is measurable.
\end{proof}

\begin{lemma}\label{lemma:invweakmeas}
If $\Lambda: M \to \Bdd(\Hilbert)$ is weakly measurable and $\Lambda(q)$ is positive and bijective for every $q\in M$ then $q\mapsto \Lambda(q)^{-1}$ is weakly measurable.
\end{lemma}

\begin{proof}
A positive operator $\Lambda(q)$ that is defined on all of $\Hilbert$ is self-adjoint, and if it is bounded and bijective then its spectrum must be contained in some interval $[a,b]$ with $0<a<b<\infty$. For every $n\in\NNN$ let $A_n\subseteq M$ be the set of those $q$ for which the spectrum of $\Lambda(q)$ is contained in $[1/n,n]$. To see that this set is measurable, choose any countable dense subset $S$ of the unit sphere of $\Hilbert$ and define
\be
A_n' := \bigl\{ q\in M: \scp{\psi}{\Lambda(q)\,\psi}\in [\tfrac1n,n] \forall \psi\in S \bigr\}\,.
\ee
This set is measurable because it is the countable intersection of the measurable sets $A'_n(\psi)=\{q\in M: \scp{\psi}{\Lambda(q)\,\psi}\in [\tfrac1n,n]\}$. But in fact, $A_n= A'_n$: $A_n\subseteq A'_n$ is clear, and if $q\in A'_n$ and $\psi\in\Hilbert$ has norm 1 then there is a sequence $(\psi_m)$ in $S$ with $\psi_m \to \psi$, and by the boundedness of $\Lambda(q)$ also $\scp{\psi_m}{\Lambda(q)\, \psi_m}\to\scp{\psi}{\Lambda(q)\,\psi}$, and therefore $\scp{\psi}{\Lambda(q)\,\psi}\in[\tfrac1n,n]$. Since $\cup_n A_n =M$, it suffices to show on $A_n$ that $q\mapsto \Lambda(q)^{-1}$ is weakly measurable.

For $q\in A_n$, consider $1/n$ times the Neumann series applied to $I-\tfrac{1}{n} \Lambda(q)$,
\be\label{Neumann}
\frac{1}{n}\sum_{k=0}^\infty \bigl(I-\tfrac{1}{n} \Lambda(q)\bigr)^k\,.
\ee
The series converges in norm because $\|I-\tfrac{1}{n} \Lambda(q)\|\leq 1-1/n^2$, and since, in case of convergence, $\sum T^k = (I-T)^{-1}$, \eqref{Neumann} is the inverse of $\Lambda(q)$. As a consequence, \eqref{Neumann} also converges weakly, and
\be
\scp{\psi}{\Lambda(q)^{-1}\,\psi} = \frac1n\sum_{k=0}^\infty \scp{\psi}{\bigl(I-\tfrac{1}{n} \Lambda(q)\bigr)^k\,\psi}\,.
\ee
Each term on the right hand side is a measurable function of $q\in A_n$ by Lemma~\ref{lemma:productweakmeas}, and thus so is the series. 
\end{proof}

\begin{lemma}\label{lemma:sqrtweakmeas}
If $\Lambda: M \to \Bdd(\Hilbert)$ is weakly measurable and $\Lambda(q) \geq 0$ for every $q\in M$ then $q\mapsto \Lambda(q)^{1/2}$ is weakly measurable.
\end{lemma}

\begin{proof}
For $n\in\NNN$ set $A_n = \{q\in M: \|\Lambda(q)\|\leq n\}$. By Lemma~\ref{lemma:normweakmeas} this is a measurable set. Since $\cup_n A_n = M$, it suffices to show on $A_n$ that $\Lambda(q)^{1/2}$ is weakly measurable.  We use the Taylor expansion of the square root function $x\mapsto x^{1/2}$ around $x=1$,
\be
(1+t)^{1/2} = \sum_{k=0}^\infty \binom{1/2}{k} \, t^k\,,
\ee
where
\be
\binom{\alpha}{k} = \frac{\alpha(\alpha-1)\cdots(\alpha-k+1)}{k!}\,.
\ee
The series converges absolutely for $|t|<1$, and thus the corresponding operator series
\be\label{sqrtseries}
\sum_{k=0}^\infty \binom{1/2}{k} \, T^k
\ee
converges in norm for self-adjoint $T$ with $\|T\|<1$. In this case (in which $I+T\geq0$), we obtain from the functional calculus for self-adjoint operators that \eqref{sqrtseries} equals indeed $(I+T)^{1/2}$.

Now let $0<\varepsilon<1/2$ and $q\in A_n$, and set
\be
T= \tfrac{1}{n}\Lambda(q) - (1-\varepsilon)I\,,
\ee
so that $I+T = \varepsilon I + \tfrac{1}{n} \Lambda(q)$. Then $-(1-\varepsilon)I \leq T \leq \tfrac{1}{n}\Lambda(q) -\tfrac{1}{2} I \leq I-\tfrac{1}{2} I=\tfrac{1}{2}I$ and thus $\|T\| \leq 1-\varepsilon$. Thus,
\be
\Bigl(\varepsilon I + \tfrac{1}{n}\Lambda(q)\Bigr)^{1/2} = 
\sum_{k=0}^\infty \binom{1/2}{k} T^k = 
\sum_{k=0}^\infty \binom{1/2}{k} \Bigl( \tfrac{1}{n}\Lambda(q) - (1-\varepsilon)I \Bigr)^k\,.
\ee
From this we can conclude with Lemma~\ref{lemma:productweakmeas} that $q\mapsto\bigl(\varepsilon I + \tfrac{1}{n}\Lambda(q)\bigr)^{1/2}$ is weakly measurable. Since limits of measurable functions are measurable, it only remains to show that 
\be
\Bscp{\psi}{\bigl(\varepsilon I + \tfrac{1}{n}\Lambda(q)\bigr)^{1/2}\,\psi} \to 
\Bscp{\psi}{\tfrac{1}{\sqrt{n}}\Lambda(q)^{1/2} \,\psi}
\text{ as }\varepsilon\to 0\,. 
\ee
Indeed, for any positive bounded operator $S$, this convergence statement holds even in norm:
\be\label{sqrtcont}
\Bigl\|(\varepsilon I + S)^{1/2} - S^{1/2}\Bigr\| \to 0\text{ as }\varepsilon \to 0\,.
\ee
To see this, set $R_\pm = (\varepsilon I +S)^{1/2} \pm S^{1/2}$; note $R_+ \geq \varepsilon^{1/2} I$, so that $R_+$ is bijective and $\|R_+^{-1}\|\leq \varepsilon^{-1/2}$; note $R_+ R_- = \varepsilon I +S-S = \varepsilon I$ (since $(\varepsilon I+S)^{1/2}$ and $S^{1/2}$ commute because $\varepsilon I +S$ and $S$ commute); thus $R_- = \varepsilon R_+^{-1}$. As a consequence, $\|R_-\| = \varepsilon \|R_+^{-1}\| \leq \varepsilon^{1/2} \to 0$ as $\varepsilon \to 0$, which is \eqref{sqrtcont}.
\end{proof}

\subsection{POVMs}
\label{sec:POVM}

A relevant mathematical concept for GRW theories is that of POVM (positive-operator-valued measure). In this section, we recall the definition of POVM and a theorem about POVMs that we need, an analog of the Kolmogorov extension theorem \cite{Tum07c}. 
%

\begin{defn}
A \textbf{POVM} (positive operator valued measure) on the measurable space $(\Omega,\salg)$ acting on $\Hilbert$ is a mapping $\povm:\salg \to \Bdd(\Hilbert)$ from a $\sigma$-algebra $\salg$ on the set $\Omega$ such that
\begin{itemize}
\item[(i)] $\povm(\Omega) = I$, 
\item[(ii)] $\povm(A) \geq 0$ 
for every $A \in \salg$, and
\item[(iii)] (weak $\sigma$-additivity) for any sequence of pairwise disjoint sets $A_1, A_2, \ldots \in \salg$
\begin{equation}\label{sigmaadditive}
    \povm\Bigl( \bigcup_{i=1}^\infty A_i \Bigr) = \sum_{i=1}^\infty \povm(A_i)\,,
\end{equation}
where the sum on the right hand side converges weakly, i.e., $\sum_i \scp{\psi}{\povm(A_i)\, \psi}$ converges, for every $\psi\in\Hilbert$, to $\scp{\psi}{\povm(\cup_i A_i)\, \psi}$.
\end{itemize}
\end{defn}

If $\povm$ is a POVM on $(\Omega,\salg)$ and $\psi\in\Hilbert$ with $\|\psi\|=1$, then $A \mapsto \scp{\psi}{\povm(A)\, \psi}$ is a probability measure on $(\Omega,\salg)$. 





\bigskip

We quote a theorem that we need from \cite{Tum07c} (see there for the proof), an analog of the Kolmogorov extension theorem for POVMs. Recall that a \emph{Borel space} is a measurable space isomorphic to a Borel subset of $[0,1]$; in particular, any Polish space with its Borel $\sigma$-algebra is a Borel space \cite{Kal97}.

\begin{thm}\label{thm:Kol}
Let $(M,\salg)$ be a Borel space and $\povm_n(\cdot)$, for every $n\in\NNN$, a POVM on $(M^n,\salg^{\otimes n})$. If the family $\povm_n(\cdot)$ satisfies the consistency property
\begin{equation}\label{FPOVMconsistent}
  \povm_{n+1}(A \times M) = \povm_n(A) \quad \forall A\in\salg^{\otimes n}
\end{equation}
then there exists a unique POVM $\povm(\cdot)$ on $(M^\NNN, \salg^{\otimes \NNN})$ (where $\salg^{\otimes\NNN}$ is the $\sigma$-algebra generated by the cylinder sets) such that for all $n\in\NNN$ and all sets $A\in\salg^{\otimes n}$,
\begin{equation}\label{FPOVMcyl}
  \povm_n(A) = \povm(A \times M^\NNN) \,.
\end{equation}
Moreover, for every $\psi\in \Hilbert$ with $\|\psi\|=1$ there exists a unique probability measure $\mu^\psi$ on $(M^\NNN,\salg^{\otimes\NNN})$ such that for all $n\in\NNN$ and all sets $A\in\salg^{\otimes n}$, $\mu^\psi(A\times M^\NNN) = \scp{\psi}{\povm_n(A)\,\psi}$, and in fact $\mu^\psi(\cdot) = \scp{\psi}{\povm(\cdot)\,\psi}$.
\end{thm}

\subsection{The Simplest Case of GRWf}
\label{sec:simplestr}

Let $H$ be a (possibly unbounded) self-adjoint operator on the separable Hilbert space $\Hilbert$. Let $(\Q,\salg_\Q)$ be a Borel space and $\mu$ a $\sigma$-finite measure on $(\Q,\salg_\Q)$; $\Q$ plays the role of physical space, which in Section~\ref{sec:simpleGRWf} we took to be $\Q=\RRR^3$ with $\salg_\Q$ the Borel $\sigma$-algebra and $\mu$ the Lebesgue measure.

\begin{ass}\label{ass:totalrate}
For every $q\in \Q$, $\Lambda(q)$ is a bounded positive operator,  
$\Lambda : \Q \to \Bdd(\Hilbert)$ is weakly measurable, and
$$\int_{\Q} \Lambda(q) \, \mu(\D q) = \lambda \, I$$
for a constant $\lambda>0$.
\end{ass}

Let $\mu_\mathrm{Leb}$ denote the Lebesgue measure on $(\RRR,\Borel(\RRR))$.

\begin{defn}\label{defn:GRWfprocess}
Under Assumption~\ref{ass:totalrate}, a random variable
$$
  F=(X_1,X_2,\ldots) = \bigl((Q_1,T_1),(Q_2,T_2), \ldots\bigr)
$$
with values in $(\Omega,\salg) = \bigl( (\Q\times\RRR)^\NNN,(\salg_\Q \otimes \Borel(\RRR))^{\otimes \NNN} \bigr)$ is a \textbf{GRWf process} with Hamiltonian $H$, flash rate operators $\Lambda(q)$, initial time $t_0$ and initial state vector $\psi$ if for every $n\in\NNN$ the joint distribution of $X_1,\ldots, X_n$ is absolutely continuous relative to $(\mu \otimes \mu_\mathrm{Leb})^{\otimes n}$ on $(\Q\times \RRR)^n$ with density $\scp{\psi}{L_n^* L_n \,\psi}$, where $L_n(x_1,\ldots, x_n)$ is given by \eqref{Lndef}.
\end{defn}

\begin{thm}\label{thm:simplest}
Under Assumption~\ref{ass:totalrate}, there exists a GRWf process for every initial time $t_0$ and every initial state vector $\psi \in \Hilbert$ with $\|\psi\|=1$. Its distribution is unique and of the form $\scp{\psi}{\povm(\cdot)\, \psi}$ for a  suitable history POVM $\povm(\cdot)$ on $\bigl( (\Q\times \RRR)^\NNN,(\salg_\Q \otimes \Borel(\RRR))^{\otimes \NNN} \bigr)$.
\end{thm}

A crucial step towards proving Theorem~\ref{thm:simplest} is the following lemma.

\begin{lemma}\label{lemma:Lnconsistent}
Set $L_0 = I$. Under Assumption~\ref{ass:totalrate}, for all $n\in\NNN$, $L_n\in \Bdd(\Hilbert)$ is well defined, $(x_1,\ldots,x_n) \mapsto L_n^* L_n$ is weakly measurable, and
\begin{equation}\label{Lnconsistent}
  \int_{\RRR} \D t_n \int_\Q \mu(\D q_n) \, L_n^* L_n = L_{n-1}^* L_{n-1} \,.
\end{equation}
\end{lemma}

\begin{proof}
Since $H$ is self-adjoint, the expression $\E^{-\I Ht/\hbar}$ defines a unitary operator. Since $\Lambda(q)$ is positive and defined on all of $\Hilbert$, it is self-adjoint, and $\Lambda(q)^{1/2}$ exists and is a bounded operator. Thus, $L_n$ is well defined on all of $\Hilbert$ and a bounded operator. 

Moreover, $L_n^* \,L_n$ as a function
\be
(\Q\times\RRR)^{n} \ni (x_1,\ldots, x_n) \mapsto L_n^*(x_1,\ldots,x_n)\, L_n(x_1,\ldots,x_n) \in \Bdd(\Hilbert)
\ee
is weakly measurable: Every $\Lambda(q_k)$ is weakly measurable by definition, also as a function on $(\Q\times\RRR)^n$ that does not depend on $t_k$ and $x_\ell$ for $\ell\neq k$. By Lemma~\ref{lemma:sqrtweakmeas}, also $(x_1,\ldots,x_n) \mapsto \Lambda(q_k)^{1/2}$ is weakly measurable. The operator-valued function $t \mapsto \E^{-\I H t}$ is weakly measurable because $t \mapsto \scp{\phi}{\E^{-\I H t} \, \psi}$ is even continuous, as even $t \mapsto \E^{-\I H t} \, \psi$ is continuous for self-adjoint $H$ \cite{RS72}. Thus, also $(x_1,\ldots,x_n) \mapsto \E^{-\I H (t_{k+1}-t_k)/\hbar}$ is weakly measurable. The number-valued function $1_{t_0<t_1< \ldots < t_n} \, \E^{-\lambda(t_n-t_0)/2}$ is known to be measurable. By Lemma~\ref{lemma:productweakmeas}, the product \eqref{Lndef} is weakly measurable, and so is $(x_1,\ldots,x_n) \mapsto L_n^* L_n$.

Now note that the definition of $L_n$ can be written as
\begin{equation}
  L_n = 1_{t_{n-1}<t_n} \,\E^{-\lambda (t_n-t_{n-1})/2} \,\Lambda(q_n)^{1/2} 
  \, \E^{-\I H(t_n-t_{n-1})/\hbar} \, L_{n-1}\,,
\end{equation}
and thus, for any $\psi\in\Hilbert$,
$$
   \Bscp{\psi}{\int_{\RRR}\D t_n \int_\Q \mu(\D q_n) \, L_n^* L_n \,\psi} =
   \int_{\RRR}\D t_n \int_\Q \mu(\D q_n) \, \scp{\psi}{L_n^*\, L_n \, \psi} = 
$$
$$
   =\int_{\RRR} \D t_n \int_{\Q} \mu(\D q_n) \, 1_{t_{n-1}<t_n} \, \E^{-\lambda (t_n-t_{n-1})}
   \, \Bscp{\psi}{L_{n-1}^* \, \E^{\I H(t_n-t_{n-1})/\hbar} \, \Lambda(q_n) \,
   \E^{-\I H(t_n-t_{n-1})/\hbar} \, L_{n-1} \, \psi} =
$$
$$
   =\int \D t_n  \, 1_{t_{n-1}<t_n} \, \E^{-\lambda (t_n-t_{n-1})} \int \mu(\D q_n)
   \, \Bscp{\E^{-\I H(t_n-t_{n-1})/\hbar} \, L_{n-1} \, \psi}{\Lambda(q_n) \,
   \E^{-\I H(t_n-t_{n-1})/\hbar} \, L_{n-1} \, \psi} =
$$
$$
   =\int \D t_n  \, 1_{t_{n-1}<t_n} \, \E^{-\lambda (t_n-t_{n-1})} 
   \, \Bscp{\E^{-\I H(t_n-t_{n-1})/\hbar} \, L_{n-1} \, \psi}{\Bigl(\int \Lambda(q_n) \,\mu(\D q_n)\Bigr)
   \E^{-\I H(t_n-t_{n-1})/\hbar} \, L_{n-1} \, \psi} =
$$
[Assumption~\ref{ass:totalrate}]
$$
   =\int \D t_n  \, 1_{t_{n-1}<t_n} \, \E^{-\lambda (t_n-t_{n-1})}\, \lambda 
   \, \Bigl\|\E^{-\I H(t_n-t_{n-1})/\hbar} \, L_{n-1} \, \psi \Bigr\|^2 =
$$
$$
   =\int \D t_n  \, 1_{t_{n-1}<t_n} \, \E^{-\lambda (t_n-t_{n-1})}\, \lambda 
   \, \bigl\| L_{n-1} \, \psi \bigr\|^2 =
$$
$$
   = \bigl\| L_{n-1} \, \psi \bigr\|^2 \int_{t_{n-1}}^\infty \D t_n  \, 
   \E^{-\lambda (t_n-t_{n-1})}\, \lambda  =
$$
$$
   = \bigl\| L_{n-1} \, \psi \bigr\|^2  = \scp{\psi}{L_{n-1}^* \, L_{n-1}\, \psi}\,.
$$
This implies \eqref{Lnconsistent}.
\end{proof}

\bigskip

\begin{proofthm}{thm:simplest} 
We use Theorem~\ref{thm:Kol} for $(M,\salg)=(\Q\times\RRR,\salg\otimes\Borel(\RRR))$. We have to check that $L_n^* L_n$ is the density of a POVM $\povm_n(\cdot)$ satisfying the consistency property \eqref{FPOVMconsistent}.
Set, for all $A \in \salg_n :=(\salg_\Q \otimes \Borel(\RRR))^{\otimes n}$,
\begin{equation}
  \povm_n(A) = \int_A \tilde{\mu}^{\otimes n}(\D x_1\cdots \D x_n) \, L_n^* L_n
\end{equation}
with $\tilde{\mu}=\mu \otimes \mu_\mathrm{Leb}$.
This defines a POVM: For $A=(\Q\times\RRR)^n$, the right hand side is, by repeated application of \eqref{Lnconsistent}, the identity. To see that $\povm_n(A)$ is a well-defined bounded operator for all $A\in\salg_n$, we apply Lemma~\ref{lemma:weakint} to $M=(\Q\times\RRR)^n$, $\salg=\salg_n$, $\Lambda=L_n^* L_n \geq 0$: in our case $S=\Hilbert$ because for all $\psi\in\Hilbert$,
\be\label{simplepovmbdd1}
  \int_A \tilde{\mu}^{\otimes n}(\D x_1 \cdots \D x_n) \, \scp{\psi}{L_n^*\, L_n\,\psi} =
  \int_A \tilde{\mu}^{\otimes n}(\D x_1 \cdots \D x_n) \, \bigl\|L_n \psi\bigr\|^2 \leq
\ee
\be\label{simplepovmbdd2}
  \leq \int_{(\Q\times\RRR)^n} \tilde{\mu}^{\otimes n}(\D x_1 \cdots \D x_n)
   \, \bigl\|L_n \psi\bigr\|^2 = 
  \scp{\psi}{\povm_n\bigl((\Q\times\RRR)^n\bigr)\, \psi} = \|\psi\|^2\,.
\ee
According to Lemma~\ref{lemma:weakint}, the sesquilinear form \eqref{sesquiweakint} is defined on $\Hilbert\times\Hilbert$, and by \eqref{simplepovmbdd1}-\eqref{simplepovmbdd2} is bounded, thus defining a bounded operator $\povm_n(A)$. 
To see that $\povm_n(\cdot)$ is weakly $\sigma$-additive, just note that $\int_A \tilde{\mu}^{\otimes n}(\D x_1 \cdots \D x_n) \, \scp{\psi}{L_n^* \, L_n \, \psi}$ is $\sigma$-additive in $A$.

The consistency condition \eqref{FPOVMconsistent} follows from \eqref{Lnconsistent}. By Theorem~\ref{thm:Kol}, there is a unique POVM $\povm$ on $(\Q\times\RRR)^\NNN$ whose marginals are the $\povm_n$. Moreover, for every $\psi\in\Hilbert$ with $\|\psi\|=1$ there is a unique probability measure $\PPP^\psi$ on $(\Q\times\RRR)^\NNN$ that extends the distributions $\scp{\psi}{\povm_n(\cdot)\,\psi}$, and which is thus the distribution of the GRWf process with initial time $t_0$ and initial state vector $\psi$. Finally, $\PPP^\psi(\cdot) = \scp{\psi}{\povm(\cdot)\, \psi}$.
\end{proofthm}

\bigskip

The \emph{labeled GRWf processes} we considered in Section~\ref{sec:labeled} are included in Definition~\ref{defn:GRWfprocess}, and their existence is covered by Theorem~\ref{thm:simplest} by setting $\Q= \RRR^3 \times \Lab$, where $\Lab$ is a finite or countable set of labels, $\salg_\Q=\Borel(\RRR^3) \otimes \Power(\Lab)$, where $\Power(\Lab)$ is the power set of $\Lab$, and $\mu = \mu_\mathrm{Leb} \otimes \nu$, where $\mu_\mathrm{Leb}$ is the Lebesgue measure on $\RRR^3$ and $\nu$ the counting measure on $\Lab$. Thus, the labeled GRWf process is a point process in $\RRR^4 \times \Lab$, and the distribution of the first $n$ labeled flashes is given by \eqref{labelFPOVMn}. Assumption~\ref{ass:totalrate} requires, in the notation of Section~\ref{sec:labeled}, that every $\Lambda_i(q)$ is bounded, that $q\mapsto \Lambda_i(q)$ is weakly measurable for every $i\in\Lab$, and that \eqref{labelLambdaR3I} holds (where the series converges weakly if $\Lab$ is infinite).

\subsection{Time-Dependent Operators}

We now show the existence of a GRWf process for time-dependent $H(t)$ and $\Lambda(q,t)$ operators with a variable total flash rate (not requiring the normalization $\int \Lambda(q) \D q = \lambda \, I$), but only under rather restrictive assumptions, particularly that these operators are \emph{bounded}. In many of the physical applications it would be desirable to permit unbounded (self-adjoint, in particular densely defined) operators: first, the physical Hamiltonians $H(t)$ are (more often than not) unbounded, and second, the flash rate operators in quantum field theory, as described in Example~\ref{ex:qft}, are naturally unbounded.

As mentioned already in Section~\ref{sec:specify}, one can either ask about the construction of the process from given $H(t)$ and $\Lambda(z)$, or from given $W_s^t$ and $\Lambda(z)$. For our purposes it is useful to assume the second point of view first, and to turn to the construction of $W_s^t$ from $H(t)$ afterwards.

\subsubsection{Given $W$ and $\Lambda$}

Fix the initial time $t_0\in \RRR$. Suppose that we are given operators $W_s^t$ for every $s,t\geq t_0$ and $\Lambda(q,t)$ for every $t\geq t_0$ and $q\in\Q$, where $(\Q,\salg_\Q,)$ is again a Borel space and $\mu$ a $\sigma$-finite measure on $(\Q,\salg_\Q)$.

\begin{defn}\label{defn:GRWfprocesst}
Let $M= \Q\times \RRR \cup\{\ceme\}$ and $\salg_M = \salg_\Q \otimes \Borel(\RRR) \times \salg_\ceme$, where $\salg_\ceme = \bigl\{\emptyset,\{\ceme\}\bigr\}$. 
A random variable
$$
  F=(Z_1,Z_2,\ldots)
$$
with values in $\bigl( M^\NNN,\salg_M^{\otimes \NNN} \bigr)$ is a \textbf{GRWf process} with time-dependent flash rate operators $\Lambda(q,t)$, evolution operators $W_s^t$, initial time $t_0$, and initial state vector $\psi$ if $\ceme$ is absorbing and for every $n\in\NNN$ the joint distribution of $Z_1,\ldots, Z_n$ satisfies [the analogs of \eqref{totalFPOVMnt} and \eqref{totalLndeft}] 
\begin{equation}\label{Pdefr}
  \PPP\Bigl(\# F \geq n, (Z_1,\ldots, Z_n) \in A\Bigr) =
  \int_A \tilde{\mu}^{\otimes n} (\D z_1 \cdots \D z_n) \,\scp{\psi}{L_n^* L_n \,\psi} 
\end{equation}
for $A \in (\salg_\Q \otimes \Borel(\RRR))^{\otimes n}$, where $L_0=I$ and 
\be\label{Lndefr}
  L_n = L_n(z_1,\ldots, z_n)  =  
  \Lambda(z_n)^{1/2} \, W_{t_{n-1}}^{t_n} \, L_{n-1}(z_1,\ldots,z_{n-1})\,.
\ee
\end{defn}

\begin{ass}\label{ass:Lambda}
For every $q\in \Q$ and $t\geq t_0$, $\Lambda(q,t)$ is a bounded operator; $(q,t)\mapsto \Lambda(q,t)$ is weakly measurable; for every $t \geq t_0$,
\be
\Lambda(\Q,t) := \int_{\Q} \Lambda(q,t)\,\mu(\D q)
\ee
exists as a bounded operator.
\end{ass}

\begin{ass}\label{ass:W}
For every $s,t\geq t_0$, $W_s^t$ is a bounded operator; for $t<s$, $W_s^t =0$; the function $(s,t)\mapsto W_s^t$ is weakly measurable and satisfies the following weak version of \eqref{W*Wt}:
\be\label{weakW*W}
W_s^{t*} W_s^t - I = -\int_s^t \D t' \, W_s^{t'}{}^* \, \Lambda(\Q,t') \, W_s^{t'} \,.
\ee
\end{ass}

We remark that, as a consequence of the weak measurability of $(q,t)\mapsto \Lambda(q,t)$ and the existence of $\Lambda(\Q,t)$ as a bounded operator, $t\mapsto \Lambda(\Q,t)$ is weakly measurable.

\begin{thm}\label{thm:exist2}
Under Assumptions~\ref{ass:Lambda} and \ref{ass:W}, there exists a GRWf process for every initial time $t_0$ and every initial state vector $\psi \in \Hilbert$ with $\|\psi\|=1$ with flash rate operators $\Lambda(q,t)$ and evolution operators $W_s^t$. The distribution of the GRWf process is unique and of the form $\scp{\psi}{\povm(\cdot)\, \psi}$ for a  suitable history POVM $\povm(\cdot)$ on $\bigl( M^\NNN,\salg_M^{\otimes \NNN} \bigr)$.
\end{thm}

\begin{lemma}\label{lemma:limitW*W}
Under Assumptions~\ref{ass:Lambda} and \ref{ass:W}, there exists a unique positive operator $T_s\in\Bdd(\Hilbert)$, denoted $\lim\limits_{t\to\infty} W_s^{t*} W_s^t$ in the following, such that
\be\label{limitW*W}
\scp{\psi}{T_s \psi} = \lim_{t\to\infty}\scp{\psi}{W_s^{t*} W_s^t \, \psi}\,.
\ee
Indeed,
\be\label{limitW*Wint}
\lim\limits_{t\to\infty} W_s^{t*} W_s^t=T_s= I- \int_s^\infty \D t' \, W_s^{t'}{}^* \, \Lambda(\Q,t') \, W_s^{t'} \,.
\ee
Moreover, $s\mapsto T_s$ is weakly measurable.
\end{lemma}

\begin{proof}
Keep $s\in \RRR$ fixed. Since $W_s^{t'}{}^* \, \Lambda(\Q,t) \, W_s^{t'}$ is a positive operator, so is its integral over $t'$, so that \eqref{weakW*W} implies $W_s^{t*} W_s^t \leq I$ and $W_s^{t_2*} W_s^{t_2} \leq W_s^{t_1}{}^* W_s^{t_1}$ for $t_1\leq t_2$. Therefore, $t\mapsto \scp{\psi}{W_s^{t*} W_s^t \, \psi}$ is a decreasing nonnegative function for every $\psi\in\Hilbert$ and thus possesses a limit $\alpha_\psi$ as $t\to\infty$. Define (polarization)
\be\label{alphapolar}
\alpha(\phi,\psi) = \frac14 \Bigl( \alpha_{\phi+\psi} - \alpha_{\phi-\psi} -\I \alpha_{\phi+\I\psi} + \I\alpha_{\phi-\I\psi}\Bigr)\,.
\ee
Then
\be\label{W*Wlimit}
\alpha(\phi,\psi) = \lim_{t\to\infty} \scp{\phi}{W_s^{t*} W_s^t \, \psi}
\ee
for all $\phi,\psi\in\Hilbert$ by \eqref{alphapolar} and the linearity of limits; in particular, the limit on the right hand side exists. It follows that $\alpha$ is a sesquilinear form $\Hilbert\times\Hilbert\to\CCC$, Hermitian, positive, and bounded (with $\|\alpha\| \leq 1$). By the Riesz lemma, there is a bounded positive operator $T_s$ such that $\alpha(\phi,\psi)=\scp{\phi}{T\psi}$. Now \eqref{W*Wlimit} implies \eqref{limitW*W} and \eqref{limitW*Wint}. From \eqref{limitW*Wint} we see that $T_s$ is weakly measurable, as integrals (such as $\int_s^\infty \D t' \, \scp{\psi}{W_s^{t'}{}^* \, \Lambda(\Q,t') \, W_s^{t'} \,\psi}$) are measurable functions of their boundaries.
\end{proof}

\begin{lemma}\label{lemma:Lnconsistentt}
Under Assumptions~\ref{ass:Lambda} and \ref{ass:W}, for all $n\in\NNN$, $L_n\in \Bdd(\Hilbert)$ is well defined, $(x_1,\ldots,x_n) \mapsto L_n^* L_n$ is weakly measurable, and
\begin{equation}\label{Lnconsistentt}
  \int_{\RRR}\D t_n \int_\Q  \mu(\D q_n) \, L_n^* L_n = 
  L_{n-1}^* \Bigl( I-\lim\limits_{t\to\infty} W_{t_{n-1}}^{t*} W_{t_{n-1}}^t \Bigr) \,L_{n-1}\,.
\end{equation}
\end{lemma}

\begin{proof}
$L_n$ is well defined on all of $\Hilbert$ and a bounded operator because the same is true of $W_s^t$ and $\Lambda(q)^{1/2}$. Moreover, $L_n^* \,L_n$ as a function
\be
M^n \ni (z_1,\ldots, z_n) \mapsto L_n^*(z_1,\ldots,z_n)\, L_n(z_1,\ldots,z_n) \in \Bdd(\Hilbert)
\ee
is weakly measurable: $z\mapsto \Lambda(z)$ is weakly measurable by Assumption~\ref{ass:Lambda}, $z\mapsto \Lambda(z)^{1/2}$ by Lemma~\ref{lemma:sqrtweakmeas}, $(s,t) \mapsto W_s^t$ by Assumption~\ref{ass:W}; now by Lemma~\ref{lemma:productweakmeas}, $L_n$ and $L_n^* L_n$ are weakly measurable.

By definition \eqref{Lndefr}, for any $\psi\in\Hilbert$,
$$
   \Bscp{\psi}{\Bigl(\int_{\RRR}\D t_n \int_\Q \mu(\D q_n) \, L_n^* L_n \Bigr)\psi} =
   \int_{\RRR} \D t_n \int_\Q \mu(\D q_n) \, \scp{\psi}{L_n^*\, L_n \, \psi} = 
$$
$$
   =\int_{\RRR} \D t_n \int_{\Q} \mu(\D q_n) \, 
   \Bscp{\psi}{L_{n-1}^* \, W_{t_{n-1}}^{t_n*} \, \Lambda(q_n,t_n) \,
   W_{t_{n-1}}^{t_n} \, L_{n-1} \, \psi} =
$$
$$
   =\int \D t_n \int \mu(\D q_n)
   \, \Bscp{W_{t_{n-1}}^{t_n} \, L_{n-1} \, \psi}{\Lambda(q_n,t_n) \,
   W_{t_{n-1}}^{t_n} \, L_{n-1} \, \psi} =
$$
$$
   =\int \D t_n   \, \Bscp{W_{t_{n-1}}^{t_n} \, L_{n-1} \, \psi}
   {\Bigl(\int \Lambda(q_n,t_n) \,\mu(\D q_n)\Bigr)
   W_{t_{n-1}}^{t_n} \, L_{n-1} \, \psi} =
$$
$$
   =\int \D t_n   
   \, \Bscp{W_{t_{n-1}}^{t_n} \, L_{n-1} \, \psi}{\Lambda(\Q,t_n)\,
   W_{t_{n-1}}^{t_n} \, L_{n-1} \, \psi} =
$$
$$
   =  \Bscp{ L_{n-1} \, \psi}
   {\Bigl(\int \D t_n \, W_{t_{n-1}}^{t_n*} \,\Lambda(\Q,t_n)\,
   W_{t_{n-1}}^{t_n}\Bigr) \, L_{n-1} \, \psi} =
$$
[by \eqref{limitW*Wint} and $W_{t_{n-1}}^{t_n} = 0$ for $t_n < t_{n-1}$]
$$
   = \Bscp{ L_{n-1} \, \psi}{\Bigl(I-\lim_{t\to\infty} W_{t_{n-1}}^{t*} W_{t_{n-1}}^t\Bigr) \, L_{n-1} \, \psi} =
$$
$$
  = \Bscp{\psi}{L^*_{n-1} \Bigl(I-\lim_{t\to\infty} W_{t_{n-1}}^{t*} W_{t_{n-1}}^t \Bigr) \,  L_{n-1} \,\psi}\,.
$$
This implies \eqref{Lnconsistentt}.
\end{proof}

\bigskip

\begin{proofthm}{thm:exist2}  
We proceed very much as in the proof of Theorem~\ref{thm:simplest}, and begin with defining the POVM $\povm_n(\cdot)$ on $M^n$ that will be the marginal of the history POVM $\povm(\cdot)$. Recall that $M=\Q\times\RRR\cup \{\ceme\}$. For $k=0,1,2,\ldots,n$, let
\be\label{Omegakndef}
\Omega_{kn} = \Bigl\{ (z_1,\ldots,z_n)\in M^n: z_1,\ldots,z_k \in\Q\times\RRR, z_{k+1} = \ldots = z_n = \ceme \Bigr\}\,.
\ee
(In particular, $\Omega_{nn} = (\Q\times\RRR)^n$.) We want that $\povm_n(\cdot)$ is concentrated on $\cup_{k=0}^n \Omega_{kn}$ (so that sequences in which $\ceme$ is followed by a flash do not occur). Consider an arbitrary $A \in \salg_M^{\otimes n}$; then $A\cap \Omega_{kn}$ is of the form ${A}_k\times \{\ceme\}^{n-k}$ for suitable ${A}_k \subseteq (\Q\times\RRR)^k$, indeed with ${A}_k \in (\salg_\Q\otimes \Borel(\RRR))^{\otimes k}$. Note $A_n = A \cap \Omega_{nn}$. Set
\begin{multline}\label{povmndeft}
  \povm_n(A) = \sum_{k=0}^{n-1} \int_{{A}_k} 
  \tilde{\mu}^{\otimes k}(\D x_1\cdots \D x_k) \, 
  L_k^*(x_1,\ldots,x_k) \,\bigl(\lim_{t\to\infty}W_{t_k}^{t*} W_{t_k}^t  \bigr) 
  \, L_k(x_1,\ldots,x_k) +\\
  + \int_{A_{n}} \tilde{\mu}^{\otimes n} (\D x_1 \cdots \D x_n) 
  \, L_n^*(x_1,\ldots,x_n) \, L_n(x_1,\ldots,x_n)\,.
\end{multline}

We show that this defines a POVM. We begin with the case $A= M^n$: Then $A_k = (\Q\times\RRR)^k$, and
\be
\int\limits_{A_n} \tilde{\mu}^{\otimes n} (\D x_1 \cdots \D x_n)\, L_n^* \, L_n =
\int\limits_{(\Q\times\RRR)^{n-1}} \tilde{\mu}^{\otimes (n-1)} (\D x_1 \cdots \D x_{n-1}) \int \D t_n \int\limits_\Q \mu(\D q_n) \, L_n^* \, L_n=
\ee
[by Lemma~\ref{lemma:Lnconsistentt}]
\be
= \int\limits_{(\Q\times\RRR)^{n-1}} \tilde{\mu}^{\otimes (n-1)} (\D x_1 \cdots \D x_{n-1}) \, L_{n-1}^* \,(I - \lim_{t\to\infty} W_{t_{n-1}}^{t*} W_{t_{n-1}}^t) \, L_{n-1}=
\ee
\begin{multline}
= \int\limits_{A_{n-1}} \tilde{\mu}^{\otimes (n-1)} (\D x_1 \cdots \D x_{n-1}) 
\, L_{n-1}^* \, L_{n-1} -\\
-\int\limits_{A_{n-1}} \tilde{\mu}^{\otimes (n-1)} (\D x_1 \cdots \D x_{n-1}) 
\, L_{n-1}^* \,\bigl(\lim_{t\to\infty} W_{t_{n-1}}^{t*} W_{t_{n-1}}^t\bigr) \, L_{n-1}\,.
\end{multline}
Iterating this calculation $n$ times, we obtain
\be
\int\limits_{A_n} \tilde{\mu}^{\otimes n} (\D x_1 \cdots \D x_n)\, L_n^* \, L_n =
I - \sum_{k=0}^{n-1} \int\limits_{A_k} \tilde{\mu}^{\otimes k}(\D x_1 \cdots \D x_k) \, L_k^* \, (\lim_{t\to\infty} W_{t_k}^{t*} W_{t_k}^t) \, L_k\,.
\ee
Together with \eqref{povmndeft}, it follows that $\povm_n(A)=\povm_n(M^n)=I$.

To see that $\povm_n(A)$ is a well-defined bounded operator for all $A\in\salg_M^{\otimes n}$, it suffices, by Lemma~\ref{lemma:weakint}, that $\scp{\psi}{\povm_n(A) \psi}\leq \|\psi\|^2$ when $\povm_n(A)$ is replaced with its definition, i.e., with the right hand side of \eqref{povmndeft}. And indeed,
\begin{multline}
  \scp{\psi}{\povm_n(A) \psi} = \sum_{k=0}^{n-1} \int\limits_{{A}_k} 
  \tilde{\mu}^{\otimes k}(\D x_1\cdots \D x_k) \, 
  \scp{\psi}{L_k^* \,\bigl(\lim_{t\to\infty}W_{t_k}^{t*} W_{t_k}^t  \bigr) 
  \, L_k \,\psi} +\\
  + \int\limits_{A_{n}} \tilde{\mu}^{\otimes n} (\D x_1 \cdots \D x_n) 
  \, \scp{\psi}{L_n^* \, L_n \,\psi}\leq
\end{multline}
[because the integrands are nonnegative by Lemma~\ref{lemma:limitW*W}]
\begin{multline}
  \leq \sum_{k=0}^{n-1} \int\limits_{(\Q\times\RRR)^k} 
  \tilde{\mu}^{\otimes k}(\D x_1\cdots \D x_k) \, 
  \scp{\psi}{L_k^* \,\bigl(\lim_{t\to\infty}W_{t_k}^{t*} W_{t_k}^t  \bigr) 
  \, L_k \,\psi} +\\
  + \int\limits_{(\Q\times\RRR)^n} \tilde{\mu}^{\otimes n} (\D x_1 \cdots \D x_n) 
  \, \scp{\psi}{L_n^* \, L_n \,\psi} = \scp{\psi}{\povm_n(M^n)\,\psi} = \|\psi\|^2\,.
\end{multline}
To see that $\povm_n(\cdot)$ is $\sigma$-additive in the weak sense, just note that $\int_A$ is $\sigma$-additive in $A$.

We check the consistency condition \eqref{FPOVMconsistent}:
\be
\povm_{n+1}(A\times M) = \povm_{n+1}(A\times (\Q\times\RRR)) + \povm_{n+1}(A \times \{\ceme\}) =
\ee
[by definition \eqref{povmndeft}]
\begin{multline}
  = \int\limits_{A_n\times (\Q\times\RRR)} \tilde{\mu}^{\otimes (n+1)}
  (\D x_1 \cdots \D x_{n+1}) \, L_{n+1}^* \, L_{n+1} +\\
  + \sum_{k=0}^{n} \int\limits_{{A}_k} 
  \tilde{\mu}^{\otimes k}(\D x_1\cdots \D x_k) \, 
  L_k^* \,\bigl(\lim_{t\to\infty}W_{t_k}^{t*} W_{t_k}^t  \bigr) \, L_k =
\end{multline}
[by \eqref{Lnconsistentt}]
\begin{multline}
  = \int\limits_{A_n} \tilde{\mu}^{\otimes n} (\D x_1 \cdots \D x_n) 
  \, L_{n}^* \,\bigl(I- \lim_{t\to\infty} W_{t_n}^{t*}  W_{t_n}^t \bigr)\, L_{n} +\\
  + \sum_{k=0}^{n} \int\limits_{{A}_k} 
  \tilde{\mu}^{\otimes k}(\D x_1\cdots \D x_k) \, 
  L_k^* \,\bigl(\lim_{t\to\infty}W_{t_k}^{t*} W_{t_k}^t  \bigr) \, L_k =
\end{multline}
\be
  = \int\limits_{A_n} \tilde{\mu}^{\otimes n} (\D x_1 \cdots \D x_n) 
  \, L_{n}^* \, L_{n} 
  + \sum_{k=0}^{n-1} \int\limits_{{A}_k} 
  \tilde{\mu}^{\otimes k}(\D x_1\cdots \D x_k) \, 
  L_k^* \,\bigl(\lim_{t\to\infty}W_{t_k}^{t*} W_{t_k}^t  \bigr) \, L_k =
\ee
\be
= \povm_n(A)\,.
\ee

By Theorem~\ref{thm:Kol}, there is a unique POVM $\povm(\cdot)$ on $M^\NNN$ whose marginals are the $\povm_n(\cdot)$. It is concentrated on the set $\Omega$ given by \eqref{cemeOmega} of sequences for which $\ceme$ is absorbing, because any other sequence, one with a space-time point after $\ceme$, would already for sufficiently large $n$ fail to be contained in any of $\Omega_{kn}$ as defined in \eqref{Omegakndef}. The history POVM $\povm(\cdot)$ is also concentrated on those sequences that are ordered by the time coordinates of the flashes, $T_1 < T_2 < \ldots$. 

Moreover, for every $\psi\in\Hilbert$ with $\|\psi\|=1$ there is a unique probability measure $\PPP^\psi$ on $M^\NNN$ that extends the distributions $\scp{\psi}{\povm_n(\cdot)\,\psi}$. Indeed, $\PPP^\psi(\cdot) = \scp{\psi}{\povm(\cdot)\, \psi}$. To see that it satisfies \eqref{Pdefr}, note that for an event $A$ concerning $Z_1,\ldots,Z_n$ and entailing that $\# F\geq n$, in other words for $A\subseteq (\Q\times\RRR)^n$ with $A\in (\salg_\Q \otimes \Borel(\RRR))^{\otimes n}$, we have that $A_n = A$ and thus
\be
\povm_n(A) =  \int_{A} \tilde{\mu}^{\otimes n} (\D x_1 \cdots \D x_n) 
  \, L_n^* \, L_n
\ee
by \eqref{povmndeft}. As a consequence, $\PPP^\psi$ defines a GRWf process.

To show that $\PPP^\psi$ is uniquely determined by \eqref{Pdefr}, we show that the joint distribution of the first $n$ components of $F$, $Z_k\in M$, must be $\scp{\psi}{\povm_n(\cdot) \, \psi}$. Indeed, from \eqref{Pdefr} it follows that, for $A\subseteq (\Q\times\RRR)^n$ with $A\in(\salg_\Q\otimes \Borel(\RRR))^{\otimes n}$,
\begin{multline}
\PPP\Bigl(\# F = n, (Z_1,\ldots, Z_n) \in A\Bigr) = 
\PPP\Bigl(\# F \geq n, (Z_1, \ldots, Z_n) \in A\Bigr) - \\ -
\PPP\Bigl(\# F \geq n+1, (Z_1, \ldots, Z_{n+1}) \in A\times (\Q\times\RRR) \Bigr)= 
\end{multline}
\be
=\int\limits_A \tilde{\mu}^{\otimes n}(\D z_1\cdots \D z_n)\, \scp{\psi}{L_n^*L_n\psi} - \int\limits_{A\times(\Q\times\RRR)} \tilde{\mu}^{\otimes (n+1)}(\D z_1 \cdots \D z_{n+1}) \, \scp{\psi}{L_{n+1}^* L_{n+1} \psi} =
\ee
[by \eqref{Lnconsistentt}]
\be
=\int\limits_A \tilde{\mu}^{\otimes n}(\D z_1\cdots \D z_n)\, \scp{\psi}{L_n^*L_n\psi} 
- \int\limits_{A} \tilde{\mu}^{\otimes n}(\D z_1 \cdots \D z_{n}) \, 
\scp{\psi}{L_{n}^* (I - \lim_{t\to\infty} W_{t_n}^{t*} W_{t_n}^t) L_{n} \psi} =
\ee
\be
= \int_{A} \tilde{\mu}^{\otimes n}(\D z_1 \cdots \D z_{n}) \, \scp{\psi}{L_{n}^* 
 (\lim_{t\to\infty} W_{t_n}^{t*} W_{t_n}^t) L_{n} \psi} \,.
\ee
This implies that for $A\subseteq M^n$ with $A\in \salg_M^{\otimes n}$ (using that $\ceme$ is absorbing) that
\be
\PPP(A) = \sum_{k=0}^{n-1}\PPP\Bigl(\# F =k, (Z_1,\ldots,Z_k) \in A_k\Bigr) + \PPP\Bigl(\# F\geq n, (Z_1,\ldots,Z_n) \in A_n\Bigr) =
\ee
\begin{multline}
=\sum_{k=0}^{n-1}
 \int_{A_k} \tilde{\mu}^{\otimes k}(\D z_1 \cdots \D z_k) \, \scp{\psi}{L_k^* 
 (\lim_{t\to\infty} W_{t_k}^{t*} W_{t_k}^t) L_k \psi} +\\ +
 \int_{A_n} \tilde{\mu}^{\otimes n} (\D z_1\cdots \D z_n) \, \scp{\psi}{L_n^* L_n \psi}
 = \scp{\psi}{\povm_n(A) \, \psi}\,.
\end{multline}
\end{proofthm}

\subsubsection{Given $H$ and $\Lambda$}
\label{sec:HLambda1}

Now suppose that we are given operators $H(t)$ for every $t\geq 0$ and $\Lambda(q,t)$ for every $t\geq t_0$ and $q\in\Q$, where $(\Q,\salg_\Q)$ is again a Borel space and $\mu$ a $\sigma$-finite measure on $(\Q,\salg_\Q)$. Our aim now is to construct the evolution operators $W_s^t$.

\begin{ass}\label{ass:dyson}
For every $t\geq t_0$, $H(t)$ is a bounded self-adjoint operator; $t\mapsto H(t)$ is weakly measurable. 
Moreover, for every $t\geq t_0$
\begin{equation}\label{Htbdd}
\int_{t_0}^t \|H(s)\| \, \D s< \infty\,, \quad \int_{t_0}^t \|\Lambda(\Q,s)\|\,\D s< \infty\,.
\end{equation}
\end{ass}

The functions $t\mapsto \|H(t)\|$ and $t\mapsto \|\Lambda(\Q,t)\|$ are measurable by Lemma~\ref{lemma:normweakmeas}.

As an abbreviation, set
\be
R_t = -\tfrac{1}{2}\Lambda(\Q,t) -\tfrac{\I}{\hbar} H(t)\,.
\ee
Note that $t\mapsto R_t$ is weakly measurable and $R_t$ is bounded with $\|R_t\| \leq \frac{1}{2} \|\Lambda(\Q,t)\| + \frac{1}{\hbar} \|H(t)\|$, so that $\int_{t_0}^t \|R_s\| \, \D s<\infty$. Now define $W_s^t$ by the \emph{Dyson series}
\begin{equation}\label{dyson}
W_s^t = I+\sum_{n=1}^\infty \int_s^t \D t_1 \int_{t_1}^t \D t_2 \cdots \int_{t_{n-1}}^t \D t_n \,  R_{t_n} \cdots R_{t_1}\,.
\end{equation}

\begin{lemma}\label{lemma:dyson}
Under Assumptions~\ref{ass:Lambda} and \ref{ass:dyson}, the Dyson series \eqref{dyson} is weakly convergent and defines a bounded operator $W_s^t$ on $\Hilbert$. The function $(s,t)\mapsto W_s^t$ is weakly measurable and satisfies the following weak version of \eqref{Wdeft}:
\be\label{weakWdeft}
  W_s^{t} - I =  \int_s^t \D t' \,
  \Bigl(-\tfrac{1}{2} \Lambda(\Q,t') - \tfrac{\I}{\hbar} H(t')\Bigr) W_s^{t'} \,,
\ee
as well as \eqref{weakW*W}. Thus, Assumption~\ref{ass:W} is fulfilled.
\end{lemma}

\begin{proof}
To see that \eqref{dyson} is weakly convergent, note that
\begin{equation}
\sum_{n=1}^\infty \int_s^t \D t_1 \int_{t_1}^t \D t_2 \cdots \int_{t_{n-1}}^t \D t_n \, 
\Bigl| \scp{\psi}{R_{t_n} \cdots R_{t_1} \, \psi} \Bigr| \leq
\end{equation}
\begin{equation}
\leq \|\psi\|^2 \sum_{n=1}^\infty \int_s^t \D t_1 \int_{t_1}^t \D t_2 \cdots \int_{t_{n-1}}^t \D t_n \,  \|R_{t_n}\| \cdots \|R_{t_1}\| =
\end{equation}
\begin{equation}
= \|\psi\|^2 \sum_{n=1}^\infty \frac{1}{n!} \Bigl(\int_s^t \D t_1 \|R_{t_1}\| \Bigr)^n \leq \|\psi\|^2 \E^{\int_s^t \D t_1 \|R_{t_1}\|}< \infty\,.
\end{equation}
As a consequence, $\scp{\psi}{W_s^t \,\psi}$ is well defined and defines a bounded quadratic form and thus a bounded operator $W_s^t:\Hilbert\to\Hilbert$.

To see that $(s,t) \mapsto W_s^t$ is weakly measurable, note that (i)~$t'\mapsto R_{t'}$ is; (ii)~by Lemma~\ref{lemma:productweakmeas}, $(t_1,\ldots,t_n) \mapsto R_{t_n} \cdots R_{t_1}$ is; (iii)~integrals are measurable functions of their boundaries; and (iv)~limits of measurable function are measurable.

To check \eqref{weakWdeft}, note first that the domain of integration in $\RRR^n$ for the $n$-th term of \eqref{dyson} is characterized by $s\leq t_1 \leq \ldots \leq t_n \leq t$, and changing the order of integration (because of absolute weak convergence), \eqref{dyson} can be rewritten as
\be
W_s^t = I+\sum_{n=1}^\infty \int_s^t \D t_n \int_{s}^{t_n} \D t_{n-1} \cdots \int_{s}^{t_2} \D t_1 \,  R_{t_n} \cdots R_{t_1}\,.
\ee
As a consequence, the right hand side of \eqref{weakWdeft} is
\be
\int_s^t \D t' \, R_{t'} W_s^{t'} = \int_s^t \D t' \, R_{t'} + \int_s^t \D t' \, R_{t'} \sum_{n=1}^\infty  \int_s^{t'} \D t_n \int_{s}^{t_n} \D t_{n-1} \cdots \int_{s}^{t_2} \D t_1 \, R_{t_n} \cdots R_{t_1} =
\ee
[using \eqref{Rint}]
\be
= \int_s^t \D t' \, R_{t'} + \int_s^t \D t' \sum_{n=1}^\infty  \int_s^{t'} \D t_n \int_{s}^{t_n} \D t_{n-1} \cdots \int_{s}^{t_2} \D t_1 \, R_{t'} \, R_{t_n} \cdots R_{t_1} =
\ee
[$\int\D t'$ and $\sum_n$ can be exchanged because of absolute (weak) convergence]
\be
= \int_s^t \D t' \, R_{t'} + \sum_{n=1}^\infty   \int_s^t \D t' \int_s^{t'} \D t_n \int_{s}^{t_n} \D t_{n-1} \cdots \int_{s}^{t_2} \D t_1 \, R_{t'} \, R_{t_n} \cdots R_{t_1} =
\ee
[rename $t'\to t_{n+1}$]
\be
= \int_s^t \D t_1 \, R_{t_1} + \sum_{n=1}^\infty \int_s^t \D t_{n+1}  \int_s^{t_{n+1}} \D t_n \int_{s}^{t_n} \D t_{n-1} \cdots \int_{s}^{t_2} \D t_1 \,  R_{t_{n+1}} \, R_{t_n} \cdots R_{t_1}=
\ee
[$m:=n+1$]
\be
= \sum_{m=1}^\infty \int_s^t \D t_{m}  \int_s^{t_{m}} \D t_{m-1} \cdots \int_{s}^{t_2} \D t_1 \,  R_{t_{m}} \, R_{t_{m-1}} \cdots R_{t_1}= W_s^t - I\,.
\ee

To check \eqref{weakW*W}, we proceed in a similar way. To simplify notation, set $\tau=(t_1,\ldots,t_n)$, $R_\tau= R_{t_n} \cdots R_{t_1}$, and
\be
S_n(s,t)= \bigl\{(t_1,\ldots,t_n) \in \RRR^n: s\leq t_1 \leq \ldots \leq t_n \leq t\bigr\}\,.
\ee
For $n=0$, set
\be
R_\emptyset = I\quad \text{and} \quad \int_{S_0(s,t)} \D \tau \, f(\tau) = f(\emptyset)\,.
\ee
Then the Dyson series \eqref{dyson} can be written as
\be
W_s^t = \sum_{n=0}^\infty \int\limits_{S_n(s,t)} \D \tau \, R_\tau\,.
\ee

Now observe that the right hand side of \eqref{weakW*W} is
\be
-\int_s^t \D t'\, W_s^{t'}{}^* \Lambda(\Q,t') \, W_s^{t'} = \int_s^t \D t' \,
 W_s^{t'}{}^* (R_{t'}^* + R_{t'}) W_s^{t'} =
\ee
[using \eqref{Rint}; the ordering of summation and integration can be changed because of absolute (weak) convergence]
\be
= \int_s^t\D t' \sum_{n,n^*=0}^\infty \int\limits_{S_n(s,t')} \D \tau \int\limits_{S_{n^*}(s,t')} \D \tau^* \, R^*_{\tau^*}  (R^*_{t'} + R_{t'}) R_{\tau} =
\ee
[separating $R^*_{t'}$ and $R_{t'}$]
\begin{multline}
= \int_s^t\D t' \sum_{n,n^*=0}^\infty \int\limits_{S_n(s,t')} \D \tau \int\limits_{S_{n^*}(s,t')} \D \tau^* \, R^*_{\tau^*}  \, R^*_{t'} \, R_{\tau} +\\
+ \int_s^t\D t' \sum_{n,n^*=0}^\infty \int\limits_{S_n(s,t')} \D \tau \int\limits_{S_{n^*}(s,t')} \D \tau^* \, R^*_{\tau^*}  \, R_{t'} \, R_{\tau} =
\end{multline}
[changing the ordering of integration and summation, and setting $t_0=t^*_0=s$]
\begin{multline}
= \sum_{n,n^*=0}^\infty \int\limits_{S_n(s,t)} \D \tau \int\limits_{S_{n^*}(s,t)} \D \tau^* \int_s^t\D t' \, 1_{t_n\leq t'} \, 1_{t^*_{n^*} \leq t'}\, R^*_{\tau^*}  \, R^*_{t'} \, R_{\tau} +\\
+ \sum_{n,n^*=0}^\infty \int\limits_{S_n(s,t)} \D \tau \int\limits_{S_{n^*}(s,t)} \D \tau^*\int_s^t\D t' \, 1_{t_n\leq t'} \, 1_{t^*_{n^*}\leq t'} \, R^*_{\tau^*}  \, R_{t'} \, R_{\tau}  =
\end{multline}
[renaming either $t' \to t_{n+1}$ or $t' \to t^*_{n^*+1}$]
\begin{multline}
= \sum_{n,n^*=0}^\infty \int\limits_{S_n(s,t)} \D \tau \int\limits_{S_{n^*+1}(s,t)} \D \tau^* \, 1_{t_n\leq t^*_{n^*+1}} \, R^*_{\tau^*}  \, R_{\tau} +\\
+ \sum_{n,n^*=0}^\infty \int\limits_{S_{n+1}(s,t)} \D \tau \int\limits_{S_{n^*}(s,t)} \D \tau^* \, 1_{t^*_{n^*}\leq t_{n+1}} \, R^*_{\tau^*}  \, R_{\tau} =
\end{multline}
[renaming either $m^*=n^*+1$ and $m=n$, or $m^*=n^*$ and $m=n+1$]
\begin{multline}
= \sum_{m=0}^\infty \sum_{m^*=1}^\infty \int\limits_{S_m(s,t)} \D \tau \int\limits_{S_{m^*}(s,t)} \D \tau^* \, 1_{t_m\leq t^*_{m^*}} \, R^*_{\tau^*}  \, R_{\tau} +\\
+ \sum_{m=1}^\infty \sum_{m^*=0}^\infty \int\limits_{S_{m}(s,t)} \D \tau \int\limits_{S_{m^*}(s,t)} \D \tau^* \, 1_{t^*_{m^*}\leq t_{m}} \, R^*_{\tau^*}  \, R_{\tau} =
\end{multline}
[separating the terms with $m=0$ or $m^*=0$]
\begin{multline}
= \sum_{m,m^*=1}^\infty \int\limits_{S_m(s,t)} \D \tau \int\limits_{S_{m^*}(s,t)} \D \tau^* \, 1_{t_m\leq t^*_{m^*}} \, R^*_{\tau^*}  \, R_{\tau} 
+ \sum_{m^*=1}^\infty \int\limits_{S_{m^*}(s,t)} \D \tau^*  \, R^*_{\tau^*} +\\
+ \sum_{m=1}^\infty \sum_{m^*=1}^\infty \int\limits_{S_{m}(s,t)} \D \tau \int\limits_{S_{m^*}(s,t)} \D \tau^* \, 1_{t^*_{m^*}\leq t_{m}} \, R^*_{\tau^*}  \, R_{\tau}
+ \sum_{m=1}^\infty  \int\limits_{S_{m}(s,t)} \D \tau \, R_{\tau} =
\end{multline}
[combining the first and third term]
\be
= \sum_{m,m^*=1}^\infty \int\limits_{S_m(s,t)} \D \tau \int\limits_{S_{m^*}(s,t)} \D \tau^* \, R^*_{\tau^*}  \, R_{\tau}
+ \sum_{m^*=1}^\infty \int\limits_{S_{m^*}(s,t)} \D \tau^* \, R^*_{\tau^*} 
+ \sum_{m=1}^\infty \int\limits_{S_{m}(s,t)} \D \tau \, R_{\tau} =
\ee
\be
= -I + \sum_{m=0}^\infty \sum_{m^*=0}^\infty \int\limits_{S_m(s,t)} \D \tau \int\limits_{S_{m^*}(s,t)} \D \tau^* \, R^*_{\tau^*}  \, R_{\tau} = -I +W_s^{t*} W_s^t\,.
\ee
This shows \eqref{weakW*W}.
\end{proof}

\begin{cor}
Under Assumptions~\ref{ass:Lambda} and \ref{ass:dyson}, there exists, for every initial time $t_0$ and every initial state vector $\psi \in \Hilbert$ with $\|\psi\|=1$, a GRWf process with Hamiltonians $H(t)$ and flash rate operators $\Lambda(q,t)$, where $W_s^t$ is given by the Dyson series \eqref{dyson}. The distribution of the process is unique and of the form $\scp{\psi}{\povm(\cdot)\, \psi}$ for a  suitable history POVM $\povm(\cdot)$ on $\bigl( M^\NNN,\salg_M^{\otimes \NNN} \bigr)$.
\end{cor}

\begin{proof}
By Lemma~\ref{lemma:dyson}, Assumption~\ref{ass:W} is fulfilled, and the statement follows from Theorem~\ref{thm:exist2}.
\end{proof}

\subsection{The General GRWf Scheme}

The methods developed in the previous section for time-dependent $H$ and $\Lambda$ operators cover also the general scheme, in which the operators may depend on previous flashes and the collapse operator $C$ is not necessarily the positive square root of $\Lambda$. Since the proofs are essentially the same, we formulate only the results.

\subsubsection{Given $W$ and $\Lambda$}

Fix the initial time $t_0\in \RRR$, let $(\Q,\salg_\Q)$ be a Borel space and $\mu$ a $\sigma$-finite measure on $(\Q,\salg_\Q)$. Let 
\be\label{Omegandefp}
\Omega := \bigcup_{n=0}^\infty \Omega^{(n)} := \bigcup_{n=0}^\infty
\Bigl\{ (z_1,\ldots,z_n) : z_k= (q_k,t_k) \in \Q\times\RRR\,,\: t_0 \leq t_1 \leq \ldots \leq t_n \Bigr\}\,.
\ee
For $f\in\Omega^{(n)}$ set $\# f := n$.
Suppose that for every sequence $f\in\Omega$
we are given operators $W^t(f)$ for every $t\geq t_{\# f}$ and $C(f,q,t)$ for every $t\geq t_{\# f}$ and $q\in\Q$.

\begin{defn}\label{defn:GRWfprocessp}
Let $M= \Q\times \RRR \cup\{\ceme\}$ and $\salg_M = \salg_\Q \otimes \Borel(\RRR) \times \salg_\ceme$, where $\salg_\ceme = \bigl\{\emptyset,\{\ceme\}\bigr\}$. 
A random variable
$$
  F=(Z_1,Z_2,\ldots)
$$
with values in $\bigl( M^\NNN,\salg_M^{\otimes \NNN} \bigr)$ is a \textbf{GRWf process} with past-dependent collapse operators $C(f,q,t)$, evolution operators 
$W^t(f)$, initial time $t_0$, and initial state vector $\psi$ if $\ceme$ is absorbing and for every $n\in\NNN$ the joint distribution of $Z_1,\ldots, Z_n$ satisfies
\begin{equation}\label{Pdefpr}
  \PPP\Bigl(\# F \geq n, (Z_1,\ldots, Z_n) \in A\Bigr) =
  \int_A \tilde{\mu}^{\otimes n} (\D z_1 \cdots \D z_n) \,\scp{\psi}{L_n^* L_n \,\psi} 
\end{equation}
for $A \in (\salg_\Q \otimes \Borel(\RRR))^{\otimes n}$, where $L_0=I$ and 
\be\label{Lndefpr}
  L_n = L_n(z_1,\ldots, z_n)  =  
  C(z_1,\ldots,z_n)\, W^{t_n}(z_1,\ldots,z_{n-1}) \, 
  L_{n-1}(z_1,\ldots,z_{n-1})\,.
\ee
\end{defn}

\begin{ass}\label{ass:C}
For every $f\in\Omega$, $q\in \Q$ and $t\geq t_{\# f}$, $C(f,q,t)$ is a bounded operator; $(f,q,t)\mapsto C(f,q,t)$ is weakly measurable; for every $t \geq t_{\# f}$,
\be
\Lambda(f,\Q,t) := \int_{\Q} C(f,q,t)^*\, C(f,q,t)\,\mu(\D q)
\ee
is a bounded operator.
\end{ass}

\begin{ass}\label{ass:Wp}
For every $f\in\Omega$ and $t\geq t_{\# f}$, 
$W^t(f)$ is a bounded operator; for $t<t_{\# f}$, 
$W^t(f) =0$; the function 
$(f,t)\mapsto W^t(f)$ is weakly measurable and satisfies
\be\label{weakW*Wp}
W^t(f)^* W^t(f) - I = -\int_{t_{\# f}}^t \D t' \, W^{t'}(f)^* \, \Lambda(f,\Q,t') \, W^{t'}(f) \,.
\ee
\end{ass}

We remark that, as a consequence of Lemma~\ref{lemma:productweakmeas}, of the weak measurability of $(f,q,t)\mapsto C(f,q,t)$ and of the existence of $\Lambda(f,\Q,t)$ as a bounded operator, $(f,t)\mapsto \Lambda(f,\Q,t)$ is weakly measurable.

\begin{thm}\label{thm:exist3}
Under Assumptions~\ref{ass:C} and \ref{ass:Wp}, there exists a GRWf process for every initial time $t_0$ and every initial state vector $\psi \in \Hilbert$ with $\|\psi\|=1$ with collapse operators $C(f,q,t)$ and evolution operators 
$W^t(f)$. The distribution of the GRWf process is unique and of the form $\scp{\psi}{\povm(\cdot)\, \psi}$ for a  suitable history POVM $\povm(\cdot)$ on $\bigl( M^\NNN,\salg_M^{\otimes \NNN} \bigr)$.
\end{thm}

\subsubsection{Given $H$ and $\Lambda$}

Now suppose that for every $f\in \Omega$, $t \geq t_{\# f}$ and $q\in\Q$ we are given operators $H(f,t)$ and $C(f,q,t)$.

\begin{ass}\label{ass:dysonp}
For every $f\in\Omega$ and $t\geq t_{\# f}$, $H(f,t)$ is a bounded self-adjoint operator; $(f,t)\mapsto H(f,t)$ is weakly measurable. 
Moreover, for every $t\geq t_{\# f}$
\begin{equation}\label{Htbddp}
\int_{t_{\# f}}^t \|H(f,s)\| \, \D s< \infty\,, \quad 
\int_{t_{\# f}}^t \|\Lambda(f,\Q,s)\|\,\D s< \infty\,.
\end{equation}
\end{ass}

The functions $(f,t)\mapsto \|H(f,t)\|$ and $(f,t)\mapsto \|\Lambda(f,\Q,t)\|$ are measurable, as pointed out in Section~\ref{sec:HLambda1}. Set
\be
R_t(f) = -\tfrac{1}{2}\Lambda(f,\Q,t) -\tfrac{\I}{\hbar} H(f,t)\,.
\ee
Then $(f,t)\mapsto R_t(f)$ is weakly measurable and $R_t(f)$ is bounded with $\|R_t(f)\| \leq \frac{1}{2} \|\Lambda(f,\Q,t)\| + \frac{1}{\hbar} \|H(f,t)\|$, so that $\int_{t_{\# f}}^t \|R_s(f)\| \, \D s<\infty$. Now define 
$W^t(f)$ by the appropriate Dyson series
\begin{equation}\label{dysonp}
W^t(f) = I+\sum_{n=1}^\infty \int_{t_{\# f}}^t \D s_1 \int_{s_1}^t \D s_2 \cdots \int_{s_{n-1}}^t \D s_n \,  R_{s_n}(f) \cdots R_{s_1}(f)\,.
\end{equation}

\begin{cor}
Under Assumptions~\ref{ass:C} and \ref{ass:dysonp}, there exists, for every initial time $t_0$ and every initial state vector $\psi \in \Hilbert$ with $\|\psi\|=1$, a GRWf process with past-dependent Hamiltonians $H(f,t)$ and collapse operators $C(f,q,t)$, where 
$W^t(f)$ is given by the Dyson series \eqref{dysonp}. The distribution of the process is unique and of the form $\scp{\psi}{\povm(\cdot)\, \psi}$ for a  suitable history POVM $\povm(\cdot)$ on $\bigl( M^\NNN,\salg_M^{\otimes \NNN} \bigr)$.
\end{cor}

\subsection{Reconstructing $W$ and $\Lambda$}
\label{sec:reconLambdar}

We now make the considerations of Section~\ref{sec:gauge2} rigorous and show that the ``square-root-plus picture'' exists. For simplicity, we ignore the possibility that the sequence of flashes could stop, thus discarding the symbol $\ceme$ and assuming that $\povm(\cdot)$ is a POVM on $(M^\NNN,\salg_M^{\otimes \NNN})$ with $M= \Q\times \RRR$ and $\salg_M = \salg_\Q \otimes \Borel(\RRR)$. Define the marginal $\povm_n(\cdot)$ of $\povm(\cdot)$ by
\be
\povm_n(A) = \povm(A\times M^\NNN)
\ee
for all $A \in \salg_M^{\otimes n}$. Let
\be\label{timeorderedOmega}
\Omega = \Bigl\{ (z_1,z_2,\ldots) \in M^\NNN: z_k=(q_k,t_k) \in M, \: 
t_0\leq t_1\leq t_2 \leq \ldots\Bigr\}
\ee
be the set of time-ordered sequences of flashes, and
\be
\Omega^{(n)} = \Bigl\{ (z_1,\ldots,z_n) \in M^n: t_0\leq t_1\leq \ldots \leq t_n\Bigr\}
\ee
the set of length-$n$ time-ordered sequences as in \eqref{Omegandefp}.

\begin{ass}\label{ass:recon}
The POVM $\povm(\cdot)$ on $(M^\NNN,\salg_M^{\otimes \NNN})$ is such that
\begin{itemize}
\item each of its marginals $\povm_n(\cdot)$ possesses an operator-valued density function $E_n$, i.e., there is a weakly measurable $E_n: M^n \to \Bdd(\Hilbert)$ with
\be\label{povmdensityE}
\povm_n(A) = \int_A \tilde{\mu}^{\otimes n}(\D z_1\cdots \D z_n) \, E_n(z_1,\ldots,z_n)
\ee
for all $A\in\salg_M^{\otimes n}$; 
\item $\povm(\cdot)$ is concentrated on $\Omega$ as given by \eqref{timeorderedOmega}, i.e., $\povm(\Omega)=I$;
\item for all $f\in M^n$ and $t\geq t_{n}$,
\be\label{intEbdd}
\int_\Q \mu(\D q) \, E_{n+1}(f,q,t)
\ee
exists as a bounded operator; 
\item $E_n(f): \Hilbert\to\Hilbert$ is a bijective operator for all $f\in M^n$, and 
\be\label{inttinftyF}
\int_{t}^\infty \D s\int_{\Q} \mu(\D q) \, E_{n+1}(f,q,s)
\ee
is a bijective operator $\Hilbert\to\Hilbert$ for every $t\geq t_n$.
\end{itemize}
\end{ass}

\begin{thm}\label{thm:recon}
If a given POVM $\povm(\cdot)$ on $(M^\NNN,\salg_M^{\otimes \NNN})$ satisfies Assumption~\ref{ass:recon} then there exist positive operators $C(f)$ and 
$W^t(f)$ (square-root-plus picture), satisfying Assumptions~\ref{ass:C} and \ref{ass:Wp}, so that $\povm(\cdot)$ is the history POVM of the GRWf process associated with $C(f)$ and 
$W^t(f)$ by Theorem~\ref{thm:exist3}.
\end{thm}

I conjecture that the last item in Assumption~\ref{ass:recon} is stronger than necessary, in particular that $E_n(f)$ does not have to be a bijective operator. In particular, rGRWf possesses a positive-operator-valued density function $E_n(f)$ which is not bijective, and I conjecture that it fits the GRWf scheme nonetheless.

\begin{lemma}
If the POVM $\povm(\cdot)$ on $(M^\NNN,\salg_M^{\otimes \NNN})$ is such that each of its marginals $\povm_n(\cdot)$ possesses an operator-valued density function $E_n$ as in \eqref{povmdensityE} relative to $\tilde{\mu}^{\otimes n}$, then $E_n(f)\geq 0$ for $\tilde{\mu}^{\otimes n}$-almost all $f\in M^n$, and
\be\label{Enconsistent}
\int_M \tilde{\mu}(\D z) \, E_{n+1}(f,z) = E_n(f)
\ee
for $\tilde{\mu}^{\otimes n}$-almost all $f\in M^n$. If $\povm(\Omega)=I$ then $E_n(f)=0$ for $\tilde{\mu}^{\otimes n}$-almost all $f\in M^n\setminus \Omega^{(n)}$.
\end{lemma}

\begin{proof}
We begin with showing that $E_n(f)\geq 0$ for $\tilde{\mu}^{\otimes n}$-almost all $f$. Let $S$ be a countable dense subset of $\Hilbert$, and for $\psi\in S$ let $A_\psi$ be the set of $f$ for which $\scp{\psi}{E_n(f) \, \psi}<0$. Since $f \mapsto \scp{\psi}{E_n(f) \,\psi}$ is a Radon--Nikodym density function of the measure $\scp{\psi}{\povm_n(\cdot) \,\psi}$ relative to $\tilde{\mu}^{\otimes n}$, it is nonnegative almost everywhere, i.e., $A_\psi$ is a null set. As a consequence, $A_S := \cup_{\psi \in S} A_\psi$ is a null set. Now for arbitrary $\psi \in \Hilbert$, there is a sequence $(\psi_m)_{m\in\NNN}$ in $S$ with $\psi_m\to\psi$ as $m\to \infty$, and hence $\scp{\psi_m}{E_n(f)\,\psi_m} \to\scp{\psi}{E_n(f) \, \psi}$. Since the limit cannot be negative if none of the members of the sequence is, $\scp{\psi}{E_n(f) \, \psi}\geq 0$ on $M^n \setminus A_S$, which is what we claimed. 

We turn to \eqref{Enconsistent}. There is no loss of generality in assuming that $E_n(f)\geq 0$ for \emph{all} (instead of almost all) $f\in M^n$ (and all $n$): a change of $(f,z)\mapsto E_{n+1}(f,z)$ on a $\tilde{\mu}^{\otimes(n+1)}$-null set entails that for $\tilde{\mu}^{\otimes n}$-almost all $f$, $z\mapsto E_{n+1}(f,z)$ changes only on a $\tilde{\mu}$-null set of $z$'s, so that the integral in \eqref{Enconsistent} is not affected. Let $\psi\in\Hilbert$ and consider the two functions
\be
g^\psi(f)= \int_M \tilde{\mu}(\D z) \,\scp{\psi}{ E_{n+1}(f,z) \,\psi}\,,\quad
h^\psi(f) = \scp{\psi}{E_n(f) \, \psi}\,.
\ee
By the Fubini--Tonelli theorem,
\be
\int_A \tilde{\mu}^{\otimes n} (\D f) \, g^\psi(f) = 
\int_{A\times M} \tilde{\mu}^{\otimes(n+1)}(\D f_{n+1}) \,\scp{\psi}{ E_{n+1}(f_{n+1}) \,\psi}
=\povm_{n+1}(A\times M)\,.
\ee
Since $\povm_{n+1}(A\times M)=\povm_n(A)$, $g^\psi$ is a density function of the measure $\scp{\psi}{\povm_n(\cdot) \,\psi}$ relative to $\tilde{\mu}^{\otimes n}$. Of course, $h^\psi$ is another density function of the same measure. By the Radon--Nikodym theorem, the density is unique up to changes on null sets, and thus 
\be\label{gfEnf}
g^\psi(f) = h^\psi(f)
\ee
for almost all $f$. 

We still have to show that a null set containing all $f$ for which \eqref{gfEnf} fails to hold can be chosen independently of $\psi$. 
To this end, let $S$ be a countable dense subset of $\Hilbert$; without loss of generality we assume that $S$ is a vector space over the complex rationals $\QQQ + \I\QQQ$. For $\psi \in S$ let $A_\psi$ be the set of those $f$ for which \eqref{gfEnf} fails to hold. We know that $A_\psi$ is a null set, and thus that $A_S = \cup_{\psi \in S} A_\psi$ is a null set. 
Fix $f\in M^n \setminus A_S$. We have that for all $\psi\in S$, $g^\psi(f) = h^\psi(f)$. By the vector space structure of $S$, if $\phi$ and $\psi$ are contained in $S$ then so are $\phi\pm\psi$ and $\phi\pm \I\psi$; using the polarization identity
\be
\scp{\phi}{T\psi}= \tfrac{1}{4} \Bigl( 
Q(\phi+\psi) - Q(\phi-\psi) -\I Q(\phi+\I\psi) +\I Q(\phi-\I\psi) \Bigr)
\ee
with $Q(\chi) = \scp{\chi}{T\chi}$, we obtain that
\be
\int_M \tilde{\mu}(\D z) \,\scp{\phi}{ E_{n+1}(f,z) \,\psi} = \scp{\phi}{E_n(f) \, \psi}
\ee
for all $\phi,\psi\in S$. By linearity in $\phi,\psi$ of each side, this is also true for all $\phi,\psi$ in the $\CCC$ vector space spanned by $S$. That is, $g^\psi(f) = h^\psi(f)$ for all $\psi$ from a dense subspace of $\Hilbert$, and hence, by Lemma~\ref{lemma:denseint}, for all $\psi\in\Hilbert$. That is, \eqref{Enconsistent} holds for all $f\in M^n \setminus A_S$.

Now suppose $\povm(\Omega)=I$. Then $\povm_n(\Omega^{(n)})=I$, or $\povm_n(A_n)=0$ for $A_n :=M^n \setminus \Omega^{(n)}$. Let $S$ be a countable dense subset of $\Hilbert$, and for $\psi\in S$ let $A_\psi$ be the set of those $f\in A_n$ for which $\scp{\psi}{E_n(f)\,\psi} \neq 0$. Since the integral of the nonnegative function $f \mapsto \scp{\psi}{E_n(f)\, \psi}$ over $A_n$ equals $\scp{\psi}{\povm_n(A_n)\,\psi}=0$, the function must vanish $\tilde{\mu}^{\otimes n}$-almost everywhere in $A_n$, and thus $\tilde{\mu}^{\otimes n}(A_\psi) = 0$. As a consequence, $A_S := \cup_{\psi \in S} A_\psi$ is a null set. Now for arbitrary $\psi \in \Hilbert$, there is a sequence $(\psi_m)_{m\in\NNN}$ in $S$ with $\psi_m\to\psi$ as $m\to \infty$, and hence (since $E_n(f)$ is bounded) $0=\scp{\psi_m}{E_n(f)\,\psi_m} \to \scp{\psi}{E_n(f)\,\psi}$ for every $f \in A_n\setminus A_S$, which is what we claimed.
\end{proof}

\bigskip

\begin{proofthm}{thm:recon}
There is no loss of generality in assuming that $E_n(f)\geq 0$ for \emph{all} (instead of almost all) $f\in M^n$ (and all $n$), that \eqref{Enconsistent} holds for \emph{all} $f\in M^n$, and that $E_n(f) = 0$ for \emph{all} $f \in M^n \setminus \Omega^{(n)}$: Inductively along $n$, we change $E_n(f)$ to zero if the given $E_n(f)$ was not positive or nonzero for $f\notin \Omega^{(n)}$; then we change $E_{n+1}(f,z)$ on a null set of $f$'s (and thus for a null set of pairs $(f,z)\in M^{n+1}$) so as to make \eqref{Enconsistent} true for \emph{all} $f\in M^n$ (which is clearly possible in a weakly measurable way).

For the reconstruction of the $W$ and $C$ operators we proceed along the lines of \eqref{Wreconempty}--\eqref{Lambdarecon}. 
Set $L_0 :=I$ and, for $t\geq t_0$,
\be
W^t(\emptyset) := \Bigl(\int_t^\infty \D s\int_\Q \mu(\D q) \, E_1(q,s)  \Bigr)^{1/2}\,.
\ee
By \eqref{inttinftyF}, the bracket is a well-defined and bijective operator, and must be positive because $E_1(q,s)\geq 0$. Thus, the square root exists and is positive; it is bijective, too, since if $T^2$ is bijective then so is $T$. For $t<t_0$ set 
$W^t(\emptyset)=0$. The function 
$t\mapsto W^t(\emptyset)$ is weakly measurable because integrals (such as $\int_t^\infty \D s \int_\Q \mu(\D q) \, \scp{\psi}{E_1(q,s)\,\psi}$) are measurable functions of their boundaries, and by Lemma~\ref{lemma:sqrtweakmeas} the root is measurable, too. Now set, for all $q\in\Q$ and $t\geq t_0$
\be\label{Lambda1reconr}
\Lambda(q,t) := W^t(\emptyset)^{-1} \, E_1(q,t) \, W^t(\emptyset)^{-1}\,.
\ee
This is well-defined and bijective (since $E_1(q,t)$ was assumed bijective); it is positive because $E_1(q,t)$ is positive and 
$W^t(\emptyset)$ is self-adjoint (and thus so is its inverse). It is weakly measurable as a function of $(q,t)$ because 
$W^t(\emptyset)$ is, by Lemma~\ref{lemma:invweakmeas} 
$W^t(\emptyset)^{-1}$ is, $E_1(q,t)$ is by assumption, and the product is by Lemma~\ref{lemma:productweakmeas}. Now set
\be
C(q,t) := \Lambda(q,t)^{1/2}\,.
\ee
It is clearly well-defined, bijective, positive, and weakly measurable as a function of $(q,t)$.

Our induction hypothesis asserts that 
$$
L_{n-1}(z_1,\ldots,z_{n-1}), 
\:\: W^{t_n}(z_1,\ldots,z_{n-1}),
\:\:\Lambda(z_1,\ldots,z_n),
\:\:\text{and }C(z_1,\ldots,z_n)
$$
are all well defined, bijective, positive except $L_{n-1}(z_1,\ldots,z_{n-1})$, and weakly measurable as a function of $(z_1,\ldots,z_n)\in \Omega^{(n)}$; the set $\Omega^{(n)}$ was defined in \eqref{Omegandefp}; furthermore, it is part of the induction hypothesis that $L^{-1}_{n-1}(z_1,\ldots,z_{n-1})$ is weakly measurable.

Now set, for $f_n=(z_1,\ldots,z_n)\in\Omega^{(n)}$ and $f_{n-1}=(z_1,\ldots,z_{n-1})$, 
\be\label{Lnreconr}
L_n(f_n) := C(f_n)\, W^{t_n}(f_{n-1})\, L_{n-1}(f_{n-1})\,.
\ee
By induction hypothesis, all factors are well defined, bijective, and weakly measurable as a function of $f_n$, and using Lemma~\ref{lemma:productweakmeas}, so is $L_n$; $L_n^{-1}$ is weakly measurable too, since $C(f_n)^{-1}$ and $W^{t_n}(f_{n-1})^{-1}$ are by Lemma~\ref{lemma:invweakmeas}, and $L_{n-1}(f_{n-1})^{-1}$ is by induction hypothesis. Set, for $t\geq t_n$,
\be\label{Wreconr}
W^t(f_n) := \Bigl( L_n^*(f_n)^{-1} \int_t^\infty \D s \int_\Q \mu(\D q) \, E_{n+1}(f_n,q,s) \, L_n(f_n)^{-1} \Bigr)^{1/2}\,.
\ee
Note that the adjoint of a bijective operator is bijective (because if $S$ is a left (right) inverse of $T$ then $S^*$ is a right (left) inverse of $T^*$), and the inverse of the adjoint is the adjoint of the inverse. That is why the bracket is a positive operator, so that the square root can be taken. By assumption, \eqref{inttinftyF} is bijective for $t\geq t_n$, and thus so is 
$W^t(f_n)$. 
We already know that $f_n \mapsto L_n(f_n)^{-1}$ is weakly measurable; so is the adjoint, and the middle integral is because $(f_n,q,s) \mapsto E_{n+1}(f_n,q,s)$ is by assumption. Thus, 
$(f_n,t)\mapsto W^t(f_n)$ is weakly measurable.

By the same arguments, with $z=(q,t)\in\Q\times\RRR$ and $t\geq t_n$, 
\be\label{Lambdareconr}
\Lambda(f_n,z) := W^t(f_n)^{-1} \, L_n^*(f_n)^{-1} \, E_{n+1}(f_n,z) \, L_n(f_n)^{-1} \, W^t(f_n)^{-1}
\ee
and
\be
C(f_n,z) = \Lambda(f_n,z)^{1/2}
\ee
are well defined, bijective, positive and weakly measurable as functions of $(f_n,z)$. This proves the induction hypothesis for $n+1$.

It now follows directly from \eqref{Lambdareconr}, \eqref{Lnreconr}, and \eqref{Lambda1reconr} that
\be\label{EnLn}
E_n(f_n) = L_n^*(f_n) \, L_n(f_n)
\ee
for $f_n \in \Omega^{(n)}$, and $E_n(f_n) = 0 = L_n^*(f_n) \, L_n(f_n)$ for $f_n \in M^n \setminus \Omega^{(n)}$.

To show that Assumption~\ref{ass:C} is fulfilled, it remains to check that $\Lambda(f,\Q,t)$ exists as a bounded operator. Indeed,
\be
\int_\Q \mu(\D q) \, \scp{\psi}{C(f,q,t)^* \, C(f,q,t)\,\psi} = 
\int_\Q \mu(\D q) \, \scp{\psi}{\Lambda(f,q,t) \,\psi} = 
\ee
[by the definition of $\Lambda(f,q,t)$]
\be
= \Bscp{L_n(f_n)^{-1} \, W^t(f_n)^{-1} \, \psi}{\Bigl(\int_\Q \mu(\D q)\, E_{n+1}(f_n,z) \Bigr) \, L_n(f_n)^{-1} \, W^t(f_n)^{-1} \,\psi}\leq
\ee
\be
\leq \Bigl\| \int_\Q \mu(\D q)\, E_{n+1}(f_n,z) \Bigr\|\, 
\|L_n(f_n)^{-1}\|^2 \, \|W^t(f_n)^{-1}\|^2\, \|\psi\|^2\,.
\ee
The operators in the norms are bounded because $L_n(f_n)^{-1}$ and 
$W^t(f_n)^{-1}$ are bijective, and \eqref{intEbdd} was assumed to be bounded.

To show that Assumption~\ref{ass:Wp} is fulfilled, it remains to check \eqref{weakW*Wp}.
\be
\int_{t_n}^t \D t' \, W^{t'}(f_n)^* \, \Lambda(f,\Q,t')\, W^{t'}(f_n) =
\ee
\be
= \int_{t_n}^t \D t'\int_\Q \mu(\D q) \, W^{t'}(f_n)^* \, \Lambda(f,q,t')\, W^{t'}(f_n) =
\ee
[by \eqref{Lambdareconr}]
\be
= \int_{t_n}^t \D t'\int_\Q \mu(\D q) \, L_n^*(f_n)^{-1} \, E_{n+1}(f_n,q,t') \, L_n(f_n)^{-1}=
\ee
[by Lemma \ref{lemma:Rint}]
\be
= L_n^*(f_n)^{-1} \,\Bigl(\int_{t_n}^t \D t'\int_\Q \mu(\D q) \, E_{n+1}(f_n,q,t')\Bigr) \, L_n(f_n)^{-1}\,,
\ee
while by \eqref{Wreconr}
\be
W^t(f_n)^* W^t(f_n) = L_n^*(f_n)^{-1} \Bigl(\int_t^\infty \D t' \int_\Q \mu(\D q) \, E_{n+1}(f_n,q,t')\Bigr) \, L_n(f_n)^{-1}\,.
\ee
Thus, the sum of the two equations is
\begin{multline}
W^t(f_n)^* W^t(f_n) + 
\int_{t_n}^t \D t' \, W^{t'}(f_n)^* \, \Lambda(f,\Q,t')\, W^{t'}(f_n) =\\
=L_n^*(f_n)^{-1} \Bigl(\int_{t_n}^\infty \D t' \int_\Q \mu(\D q) \, E_{n+1}(f_n,q,t')
\Bigr) \, L_n(f_n)^{-1} =
\end{multline}
[by \eqref{Enconsistent}]
\be
= L_n^*(f_n)^{-1} \, E_n(f_n) \, L_n(f_n)^{-1} = L_n^*(f_n)^{-1} \, L_n^*(f_n)\, L_n(f_n)  \, L_n(f_n)^{-1} = I
\ee
by \eqref{EnLn}.
\end{proofthm}

\section{Relativistic GRW Theory}\label{sec:defrGRWf}

We begin by introducing some terminology and notation. We generally intend that all manifolds, surfaces, and curves are $C^\infty$. 
A \emph{space-time} $(M,g)$ is a time-oriented Lorentzian 4-manifold (see, e.g., \cite{ON}). The simplest example is Minkowski space-time $\bigl(M=\RRR^4,g=\mathrm{diag}(1,-1,-1,-1)\bigr)$. A \emph{3-surface} is a 3-dimensional embedded submanifold (without boundary) of $M$ that is closed in the topology of $M$. A 3-surface $\Sigma$ is \emph{spacelike} if every nonzero tangent vector to $\Sigma$ is spacelike. Note that a spacelike 3-surface is a Riemannian manifold. If $\Sigma$ is a spacelike 3-surface and $x,y \in \Sigma$, the \emph{spacelike distance from $x$ to $y$ along $\Sigma$}, $\sdist_\Sigma(x,y)$, is the infimum of the Riemannian lengths of all curves in $\Sigma$ connecting $x$ to $y$. A curve in $M$ is \emph{timelike} if every nonzero tangent vector to the curve is timelike; we will always regard timelike curves as directed towards the future, i.e., we assume that the derivative relative to the curve parameter is future-pointing. A timelike curve is \emph{inextendible in $M$} if it is not a proper subset of a timelike curve in $M$. A curve in $M$ is \emph{causal} if every nonzero tangent vector to the curve is either timelike or lightlike; we also regard causal curves as directed towards the future. For every subset $A\subseteq M$, the \emph{(causal) future} of $A$ is the set
\begin{equation}
\future(A)=\{y\in M: \exists x\in A \: \exists \text{ a causal curve from $x$ to }y\}\,,
\end{equation}
and the \emph{(causal) past} of $A$ is
\begin{equation}
\past(A)=\{y\in M: \exists x\in A \: \exists \text{ a causal curve from $y$ to }x\}\,.
\end{equation}
For example, in Minkowski space-time
\be
\future(x) = \bigl\{y\in \RRR^4: (y^\mu-x^\mu)(y_\mu -x_\mu) \geq 0 , \: y^0-x^0\geq 0 \bigr\}\,.
\ee
(As usual, $y_\mu = g_{\mu\nu} y^\nu$, and we adopt the sum convention implying summation over indices that appear both upstairs and downstairs.)

For $y \in \future(x)$, the \emph{timelike distance of $y$
from $x$}, $\tdist(y,x)$, is the supremum of the lengths of all causal curves
connecting $x$ to $y$. For Minkowski space-time,
\begin{equation}
  \tdist(y,x) = \bigl((y^\mu-x^\mu)(y_\mu -x_\mu) \bigr)^{1/2}.
\end{equation}

\begin{ass}\label{ass:tdist}
$(M,g)$ is such that $\tdist(\cdot,x):\future(x) \to [0,\infty)$ is $C^\infty$ on the interior of $\future(x)$, and its derivative $\nabla_\mu\tdist$ vanishes nowhere. Furthermore, $\tdist(y,x) = 0$ if and only if $y\in\partial \future(x)$.
\end{ass}

For example, this is the case in Minkowski space-time. It is not the case in space-time manifolds with closed timelike curves, in which $\tdist$ may have nondifferentiable points. 

\begin{figure}[ht]
\begin{center}
\includegraphics[width=.4 \textwidth]{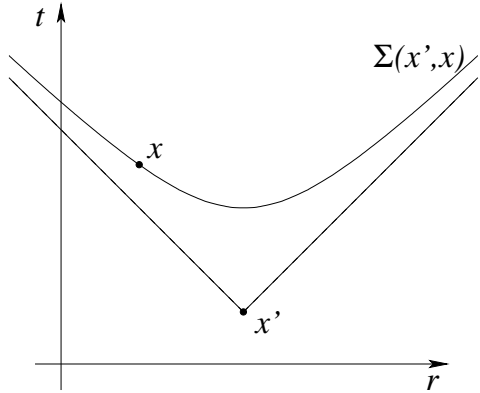}
\caption{The 3-surface $\Sigma(x',x) = \hyper(x',x)$ of constant timelike distance from $x'$ containing $x$, in Minkowski space-time.}
\end{center}
\label{figtwo}
\end{figure}

The \emph{future hyperboloid} based at a point $x$ and with distance parameter $s>0$ is the set
\begin{equation}
  \hyper_s(x) = \{y \in \future(x): \tdist(y,x) =s\}\,.
\end{equation}
If $x\in \future(x')$ then we write $\hyper(x,x') = \hyper_{\tdist(x,x')}(x')$ for the future hyperboloid based at $x'$ containing $x$. In Minkowski space-time, the future hyperboloids are 
\be\label{hyperMinkowski}
\hyper_s(x) = \Bigl\{(y^0,y^1,y^2,y^3)\in\RRR^4: 
y^0 = x^0 +\Bigl(s^2 + \sum_{i=1}^3 (y^i-x^i)^2\Bigr)^{1/2} \Bigr\}\,.
\ee
From Assumption~\ref{ass:tdist} it follows (by the implicit function theorem) that $\hyper_s(x)$ is an embedded submanifold, and thus a 3-surface; it is spacelike because $\nabla_\mu\tdist$ is timelike.

A \emph{Cauchy surface} in $M$ is a spacelike 3-surface that intersects every inextendible causal curve in $M$ exactly once.\footnote{O'Neill \cite{ON} defines a Cauchy surface as a \emph{subset} that intersects every inextendible \emph{timelike} curve in $M$ exactly once. That is different in two ways: it allows submanifolds that are not $C^\infty$, and it allows 3-surfaces possessing lightlike tangent vectors.} Let $\Cauchy$ be the set of all Cauchy surfaces in $M$, $\Hyper$ the set of all future hyperboloids in $M$. The future hyperboloids are not necessarily Cauchy surfaces. In Minkowski space-time, for example, they never are: Indeed, for given $x\in\RRR^4$, $t\mapsto y(t) = x+(t,\sqrt{1+t^2},0,0)$ is an inextendible timelike curve that does not intersect $\future(x)$, and in particular not the future hyperboloids. To see this, note first that its tangent vector $u^\mu = \D y^\mu/\D t = (1,t/\sqrt{1+t^2},0,0)$, is always timelike as $u^\mu \, u_\mu = 1-t^2/(1+t^2)>0$, and since $u^\mu$ is nonzero every other tangent vector is a multiple of $u^\mu$. It is inextendible because $y^0(t) \to \pm \infty$ as $t\to\pm \infty$, and it does not intersect $\future(x)$ because $(y^\mu(t)-x^\mu)(y_\mu(t)-x_\mu) = t^2-(1+t^2) = -1<0$.

As a consequence of its Lorentzian metric, $M$ is endowed with a natural $\sigma$-finite measure, which we denote $\D^4 x$. Similarly, every spacelike 3-surface $\Sigma$, being a Riemannian manifold, is endowed with a natural $\sigma$-finite measure, the Riemannian volume measure, which we denote $\D^3 x$ (it will always be clear which $\Sigma$ we refer to). For example, if the hyperboloid $\hyper_s(0)$ given by \eqref{hyperMinkowski} is coordinatized by $x^1,x^2,x^3$ then the measure $\D^3 x$ has density $1/\sqrt{1+r^2/s^2}$ in coordinates, i.e.,
\be\label{hyperdensity}
\int_{\hyper_s(0)} \D^3 x\, f(x) = \int_{\RRR^3} \D x^1\, \D x^2 \, \D x^3\, 
\frac{f\bigl(\sqrt{s^2+r^2}, x^1,x^2,x^3 \bigr)}{\sqrt{1+r^2/s^2}}\,,
\ee
where $r(x^1,x^2,x^3):=(\sum_{k=1}^3 (x^k)^2)^{1/2}$.

\begin{lemma}\label{lemma:D4D3}
(Coarea formula) Under Assumption~\ref{ass:tdist}, for any $x'\in M$ and any measurable $f:\future(x') \to [0,\infty)$,
\be
\int_{\future(x')} \D^4x \, f(x) = \int_0^\infty \D s \int_{\hyper_s(x')} \D^3 x\, f(x)\,.
\ee
\end{lemma}

\begin{proof}
The general coarea formula can be found as Theorem 3.2.12 in \cite{Fed}. (Actually, it is not necessary that $\tdist(\cdot,x')$ be $C^\infty$: we only need locally Lipschitz (which is true in any Lorentzian manifold), but would then have to say more about the definition of $\D^3x$.) The easiest way to see that the Jacobian factor is correct is by noting that an orthonormal basis of the tangent space in $x$ to $\hyper_s(x')$, together with the future-pointing unit normal in $x$ to $\hyper_s(x')$, forms an orthonormal basis of the tangent space in $x$ to $M$.
\end{proof}

\subsection{Abstract Definition of the Relativistic Flash Process}

We now present the abstract definition of the relativistic GRW flash process, or rGRWf process. It is abstract in the sense that it supposes certain operators as given, for which we provide concrete specification later in Section~\ref{sec:concreterGRWf}.

Suppose that for every $\Sigma \in \Cauchy \cup \Hyper$ we are given a Hilbert space $\Hilbert_\Sigma$, and we are given another Hilbert space $\Hilbert_0$. 
Suppose that we are given unitary time evolution operators, as in ordinary quantum mechanics, in the following sense: For every two $\Sigma,\Sigma' \in \Cauchy \cup \Hyper$ we are given a unitary isomorphism $U_\Sigma^{\Sigma'}: \Hilbert _\Sigma \to \Hilbert_{\Sigma'}$ such that 
\begin{equation}\label{Ucocycle}
  U_\Sigma^\Sigma = I_{\Sigma} \,, \quad
  U_{\Sigma'}^{\Sigma''} U_\Sigma^{\Sigma'} = U^{\Sigma''}_\Sigma \,,
\end{equation}
where $I_\Sigma$ denotes the identity operator on $\Hilbert_\Sigma$. (Our example will be the time evolution defined by the Dirac equation.)

The family $(U_\Sigma^{\Sigma'})$ can be represented in another way through a family of unitary isomorphisms  $U_\Sigma: \Hilbert_0 \to \Hilbert_\Sigma$. Indeed, if $(U_\Sigma^{\Sigma'})$ are given, choose an arbitrary $\Sigma_0\in\Cauchy \cup\Hyper$ and an arbitrary unitary isomorphism $U_{\Sigma_0}:\Hilbert_0 \to \Hilbert_{\Sigma_0}$, and set $U_\Sigma = U_{\Sigma_0}^\Sigma \, U_{\Sigma_0}$. Conversely, if a family $U_\Sigma: \Hilbert_0\to\Hilbert_\Sigma$ is given, define $U_\Sigma^{\Sigma'} = U_{\Sigma'}^{-1} \, U_\Sigma$, and \eqref{Ucocycle} is satisfied. (Note that identifying $\Hilbert_0$ with $\Hilbert_\Sigma$ by means of $U_\Sigma$ is nothing but the Heisenberg picture.)

Furthermore, for every $\Sigma \in \Hyper$ we are given an operator-valued function $\Lambda_\Sigma: \Sigma \to \Bdd(\Hilbert_\Sigma)$ such that every $\Lambda_\Sigma(x)$ is positive. Let $\lambda>0$ be a constant (the same as in Section~\ref{sec:simpleGRWf}).

\begin{ass}\label{ass:LambdaSigma}
$\Lambda_\Sigma: \Sigma \to \Bdd(\Hilbert_\Sigma)$ is weakly measurable, and
\begin{equation}\label{Lambdanorm}
  \int_\Sigma \D^3 x \, \Lambda_\Sigma(x) = \lambda\, I_\Sigma \,.
\end{equation}
In addition, on the set $\{(x,x')\in M^2: x\in \future(x')\}$ the function
\be\label{hypermeas}
(x,x') \mapsto U_{\hyper(x,x')}^{-1} \, \Lambda_{\hyper(x,x')}(x) \, U_{\hyper(x,x')} \in \Bdd(\Hilbert_0)
\ee
is weakly measurable.
\end{ass}

(For a concrete specification of $\Hilbert_\Sigma$, $U_\Sigma^{\Sigma'}$, and $\Lambda_\Sigma(x)$ see Section~\ref{sec:concreterGRWf}.)

Moreover, suppose we are given a finite label set $\Lab$; set $N:=\# \Lab$. For every $i\in\Lab$ we are given a point $X_{i,0}\in M$, called the \emph{seed flash} with label $i$. Finally, we are given a vector
\be
\psi \in \bigotimes_{i\in\Lab}\Hilbert_0
\ee
(i.e., the product of $N$ copies of $\Hilbert_0$) with $\|\psi\|=1$.

Let the history space be
\be
\Omega := M^{\Lab \times \NNN}
\ee
(corresponding to one sequence of flashes in $M$ for every label $i$) with $\sigma$-algebra
\be
\salg := \Borel(M)^{\otimes (\Lab\times\NNN)}\,.
\ee
For $x,x'\in M$ define the operator $K_{x'}(x) \in \Bdd(\Hilbert_0)$ by
\begin{equation}\label{coldef}
  K_{x'}(x) := 1_{x\in \future(x')} \:
  \E^{-\lambda \tdist(x,x')/2} \: 
  U^{-1}_{\Sigma} \:
  \Lambda_{\Sigma}(x)^{1/2} \: U_{\Sigma}\,,
\end{equation}
where $\Sigma=\hyper(x,x')$.
For any sequence $f=(x_0,x_1,x_2,\ldots,x_n)$ of space-time points, set
\begin{equation}
  K(f) := K_{x_{n-1}}(x_n) \cdots K_{x_1}(x_2) \, K_{x_0}(x_1)\,.
\end{equation}

\begin{defn}
Given the data just listed, 
an \textbf{rGRWf process} is a random variable 
$$
F = \bigl(X_{i,k}: i\in\Lab, k\in \NNN \bigr)
$$
with values in $(\Omega,\salg)$ 
such that for every choice of $\vec{n}=(n_i) \in \NNN^\Lab$ the joint distribution of the first $n_i$ flashes of type $i$ is
\be\label{rGRWfdistn}
\PPP\Bigl( X_{i,k} \in \D^4 x_{i,k}: i \in\Lab, k \leq n_i \Bigr) = 
  \Bigl\| \bigotimes_{i\in\Lab} K(f_i) \, \psi \Bigr\|^2 \, \D f
\end{equation}
with the notation $f_i=(x_{i,0},\ldots, x_{i,n_i})$ and
\begin{equation}\label{df}
  \D f = \prod_{i\in\Lab} \prod_{k=1}^{n_i} \D^4x_{i,k}\,.
\end{equation}
\end{defn}

\begin{thm}\label{thm:absexist}
Given the data listed above, and if Assumptions~\ref{ass:tdist} and \ref{ass:LambdaSigma} hold,
then there exists an rGRWf process and is unique in distribution. The distribution is $\scp{\psi}{\povm(\cdot)\,\psi}$ for a suitable history POVM $\povm(\cdot)$ on the history space $\Omega$.
\end{thm}

\begin{lemma}\label{lemma:KK1}
Under Assumptions~\ref{ass:tdist} and \ref{ass:LambdaSigma},
$(x,x') \mapsto K_{x'}^*(x) \, K_{x'}(x)$ is weakly measurable, and
\begin{equation}\label{POVM}
  \int_M \D^4 x \: K_{x'}^*(x) \, K_{x'}(x) = I\,.
\end{equation}
\end{lemma}

\begin{proof}
The function $M^2\ni(x,x')\mapsto K_{x'}(x) \in \Bdd(\Hilbert_0)$ is weakly measurable because it is, up to the measurable factor $1_{x\in \future(x')} \: \E^{-\lambda \tdist(x,x')/2}$, the square root of $(x,x') \mapsto U_{\hyper(x,x')}^{-1} \, \Lambda_{\hyper(x,x')}(x) \, U_{\hyper(x,x')}$, which is weakly measurable 
by Assumption~\ref{ass:LambdaSigma}. By the usual arguments, $M^2\ni(x,x')\mapsto K_{x'}^*(x) \, K_{x'}(x)$ is weakly measurable.

By definition \eqref{coldef}, with $\Sigma = \hyper(x,x')$ and $x\in\future(x')$,
\be
  e^{\lambda\tdist(x,x')}\,K_{x'}^*(x) \, K_{x'}(x) = (U^{-1}_{\Sigma} \:
  \Lambda_{\Sigma}(x)^{1/2} \: U_{\Sigma})^* 
  \:  U^{-1}_{\Sigma} \:
  \Lambda_{\Sigma}(x)^{1/2} \: U_{\Sigma} =
\ee
[because $U_{\Sigma}$ is unitary and $\Lambda_\Sigma(x)$ is self-adjoint] 
\be
  = U_{\Sigma}^{-1} \:
  \Lambda_{\Sigma}(x)^{1/2} \: U_{\Sigma} 
  \:  U^{-1}_{\Sigma} \:
  \Lambda_{\Sigma}(x)^{1/2} \: U_{\Sigma} = 
  U^{-1}_{\Sigma} \: \Lambda_{\Sigma}(x) \: U_{\Sigma} \,.
\ee
Thus,
\[
  \int_M \D^4 x \: K_{x'}^*(x) \, K_{x'}(x) =   \int_M \D^4 x \: 1_{x\in \future(x')} \:
  \E^{-\lambda \tdist(x,x')} \: 
  U^{-1}_{\hyper(x,x')} \:
  \Lambda_{\hyper(x,x')}(x) \: U_{\hyper(x,x')} = 
\]
[by Lemma~\ref{lemma:D4D3}]
\[
  =  \int_0^\infty \D s\: \E^{-\lambda s}  \int_{\hyper_s(x')} \D^3 x \:
  U^{-1}_{\hyper_s(x')} \:
  \Lambda_{\hyper_s(x')}(x) \: U_{\hyper_s(x')} = 
\]
[by Lemma~\ref{lemma:Rint}]
\[
  =  \int_0^\infty \D s\: \E^{-\lambda s} \:U^{-1}_{\hyper_s(x')} \: 
  \Bigl(\int_{\hyper_s(x')} \D^3 x \:
  \Lambda_{\hyper_s(x')}(x)\Bigr) \: U_{\hyper_s(x')} = 
\]
[by \eqref{Lambdanorm}]
\[
  =  \int_0^\infty \D s\: \E^{-\lambda s} \:U^{-1}_{\hyper_s(x')} \: \lambda \, 
  I_{\hyper_s(x')} \: U_{\hyper_s(x')} 
  =  \int_0^\infty \D s\: \E^{-\lambda s} \: \lambda \, I = I\,.
\]
\end{proof}

The following lemma is the analog of Lemma~\ref{lemma:Rint} for tensor products.

\begin{lemma}\label{lemma:oRint}
If $\Hilbert_1,\Hilbert_2,\Hilbert_3$ are separable Hilbert spaces, $q\mapsto\Lambda(q)\in\Bdd(\Hilbert_2)$ is weakly measurable, and $R \in \Bdd(\Hilbert_1)$, $S\in\Bdd(\Hilbert_3)$, and $T=\int \Lambda(q) \, \mu(\D q) \in\Bdd(\Hilbert_2)$ then $q\mapsto R\otimes \Lambda(q) \otimes S \in\Bdd(\Hilbert_1\otimes\Hilbert_2\otimes\Hilbert_3)$ is weakly measurable, and
\be\label{oRint}
R\otimes T \otimes S = \int R \otimes \Lambda(q) \otimes S \, \mu(\D q)\,.
\ee
\end{lemma}

\begin{proof}
This is an immediate consequence of Lemma~\ref{lemma:Rint}: replace $R \to R\otimes I \otimes I$, $T \to I\otimes T \otimes I$, and $S \to I\otimes I\otimes S$, and note that $(P\otimes I) (I\otimes Q) = P\otimes Q$.
\end{proof}

\bigskip

\begin{proofthm}{thm:absexist}
For every $n\in\NNN$, we define a POVM $\povm_n(\cdot)$ on $\bigl((M^\Lab)^n, \Borel(M)^{\otimes \Lab \times n} \bigr)$ as follows:
\be\label{rpovmndef}
\povm_n(A) = \int_A  
\bigotimes_{i\in\Lab} K(f_i)^* \, K(f_i)\,\prod_{i\in\Lab}\prod_{k=1}^n \D^4x_{i,k} \,.
\ee

First, for $A=(M^\Lab)^n$, we obtain from $nN$-fold application of Lemma~\ref{lemma:KK1} (and Lemma~\ref{lemma:oRint}) that $\povm_n(A) = I$. For arbitrary $A \in \Borel(M)^{\otimes \Lab \times n}$, the existence and boundedness of $\povm_n(A)$ follows again from Lemma~\ref{lemma:weakint} if only the right hand side of \eqref{rpovmndef}, when sandwiched between $\psi$'s, remains $\leq \|\psi\|^2$. Indeed,
\be
\int_A \prod_{i\in\Lab}\prod_{k=1}^n \D^4x_{i,k} \, 
\scp{\psi}{\otimes_i K(f_i)^* \, K(f_i) \, \psi} \leq
\ee
[because the integrand is nonnegative]
\be
\leq\int_{(M^\Lab)^n} \prod_{i\in\Lab}\prod_{k=1}^n \D^4x_{i,k} \, 
\scp{\psi}{\otimes_i K(f_i)^* \, K(f_i) \, \psi} = \|\psi\|^2\,.
\ee
To see that $\povm_n(\cdot)$ is $\sigma$-additive, just note that $\int_A$ is $\sigma$-additive in $A$. Thus, $\povm_n(\cdot)$ is a POVM.

The consistency condition \eqref{FPOVMconsistent} follows from $N$-fold application of Lemma~\ref{lemma:KK1} (and Lemma~\ref{lemma:oRint}), namely to $K_{x_{i,n}}(x_{i,n+1})^* \, K_{x_{i,n}}(x_{i,n+1})$ for all $i\in\Lab$. Now Theorem~\ref{thm:Kol} provides a POVM $\povm(\cdot)$ on $(M^\Lab)^\NNN=M^{\Lab\times\NNN}$ whose marginals are the $\povm_n(\cdot)$. To see that $\scp{\psi}{\povm(\cdot)\,\psi}$ is the distribution of an rGRWf process, note that in case $n_i=n$ for all $i\in\Lab$, \eqref{rGRWfdistn} means
\be
\PPP\Bigl( X_{i,k} \in \D^4 x_{i,k}: i \in\Lab, k \leq n \Bigr) = \scp{\psi}{\povm_n(\D f) \, \psi}\,,
\ee
while \eqref{rGRWfdistn} for unequal $n_i$ follows from the case before by choosing $n$ large enough ($n=\max\{n_i:i\in\Lab\}$) and applying Lemma~\ref{lemma:KK1} to integrate out some of the $x_{i,k}$.

The uniqueness of the distribution $\PPP_\psi$ of the rGRWf process follows from Theorem~\ref{thm:Kol} because \eqref{rGRWfdistn} implies for the case in which all $n_i=n$ that the joint distribution of the first $n$ flashes of all labels is given by the POVM $\povm_n(\cdot)$, and then the uniqueness statement of Theorem~\ref{thm:Kol} provides the uniqueness of $\PPP_\psi$.
\end{proofthm}

\subsection{Concrete Specification}
\label{sec:concreterGRWf}

We now present an outline for defining $\Hilbert_\Sigma$, $U_\Sigma$, and $\Lambda_\Sigma(x)$. A rigorous definition will be presented in Section~\ref{sec:existthm} for Minkowski space-time.

Concretely, we intend to take $\Hilbert_\Sigma$ to be $L^2(\Dirac|_\Sigma)$, the space of square-integrable measurable sections of the vector bundle $\Dirac|_\Sigma$ modulo changes on null sets. The vector bundle $\Dirac$ is the bundle of Dirac spin spaces \cite{Dim82,PR84}, a complex bundle of rank 4 over $M$, endowed with a connection (whose curvature arises from the curvature of $M$).
We obtain the operators $U_\Sigma^{\Sigma'}$ by solving the Dirac equation
\be\label{Dirac}
-\I\hbar \gamma^\mu \bigl(\nabla_\mu-\tfrac{\I e}{\hbar}A_\mu\bigr)\psi = m \psi\,, 
\ee
where $\gamma^\mu$ are the Dirac matrices, $\nabla$ is the covariant derivative operator, $e\in\RRR$ is a constant, the charge parameter, $A_\mu$ is a 1-form, the electromagnetic vector potential, and $m>0$ is a constant, the mass parameter.

We take $\Lambda_\Sigma(x)$ to be the multiplication operator on $L^2(\Dirac|_\Sigma)$ by a function of the spacelike distance from $x$ along $\Sigma$,
\be\label{LambdaSigmadef}
\Lambda_\Sigma(x) \, \psi(y) = \lambda \, \mathcal{N}(y) \,  \profile\bigl(\sdist_\Sigma(x,y)\bigr) \, \psi(y)\,,
\ee
for all $y \in \Sigma$, where $\profile: [0,\infty) \to [0,1]$ is a fixed function that we call the \emph{profile function}, and
\be\label{Ndefprofile}
\mathcal{N}(y) = \Bigl(\int_\Sigma \D^3 x \, \profile\bigl(\sdist_\Sigma(x,y)\bigr) \Bigr)^{-1}\,.
\ee
The normalizing factor $\mathcal{N}$ is chosen as to ensure \eqref{Lambdanorm}. As an example, $\profile$ could be a Gaussian,
\begin{equation}\label{profileGauss}
  \profile(u) = \exp \Bigl( - \frac{u^2}{2\sigma^2} \Bigr)
\end{equation}
with $\sigma$ the same constant as in Section~\ref{sec:simpleGRWf}. It is sometimes useful to assume that $\profile$ has compact support $[0,\sigma]$; then it is clear that $\mathcal{N}(y)$ is finite.

From the concrete specification just given, 
it becomes clear that rGRWf is defined in a covariant way. No coordinate system on $M$ was ever chosen; the unitary evolution from one 3-surface to another is given by the Dirac equation; in contrast to Bohmian mechanics for relativistic space-time (as developed in \cite{HBD}), no ``time foliation'' (preferred foliation of $M$ into spacelike 3-surfaces) is assumed or constructed; more generally, no concept of simultaneity-at-a-distance is involved.

The reader should note that this is more than that the Poincar\'e group (the isometry group of Minkowski space-time) acts on the theory's solutions. In detail, let us call a theory \emph{weakly covariant} if the set of possible probability measures on the history of the primitive ontology (PO) is closed under the action of the Poincar\'e group. Indeed, rGRWf is weakly covariant, but the concept of weak covariance is too weak to capture the idea of a relativistic theory. As a simple example, we can turn non-relativistic classical mechanics (with instantaneous interaction-at-a-distance) into a weakly covariant theory in Minkowski space-time in the following way: (i)~postulate the existence of an additional physical object mathematically represented by a timelike vector field $n^\mu$ subject to the field equation $\partial_\nu n^\mu =0$ (which ensures $n^\mu$ is constant); (ii)~the vector field selects a Lorentz frame (whose time axis lies in the direction of $n^\mu$); (iii)~in this frame apply the non-relativistic equations. (Since probability plays no role here, we can think of the probability measure as a Dirac measure concentrated on a single history.) This theory is weakly covariant, as the transformed history simply has a different $n^\mu$ vector, even though in a world governed by this theory the Michelson--Morley experiment has a nonzero result and superluminal communication is possible. (On top of that, the definition of weak covariance is limited to special relativity, and it is not clear how to adapt it to curved space-time.)

Fay Dowker (personal communication, January 28, 2004) has proposed the following definition for the concept of a \emph{covariant law}: Suppose a law $\law$ is such that for every Cauchy surface $\Sigma$ in $M$ there is a set $\ini_\Sigma$ of possible initial data on $\Sigma$, and $\law$ associates with every $D\in\ini_\Sigma$ a probability measure $\PPP_D$ on the space of possible histories of the PO in $\future(\Sigma)$. Now call the law $\law$ \emph{strongly covariant} if for every two Cauchy surfaces $\Sigma_1,\Sigma_2$ with $\Sigma_2 \subseteq \future(\Sigma_1)$ and every $D_1 \in \ini_{\Sigma_1}$ there is a random variable $D_2$ with values in $\ini_{\Sigma_2}$ so that the history of the PO in $\future(\Sigma_2)$ can equally be regarded as generated by the initial datum $D_1$ on $\Sigma_1$ or $D_2$ on $\Sigma_2$; that is, the distribution $\PPP_{D_2}$ averaged over the distribution of $D_2$ agrees with $\PPP_{D_1}$ restricted to $\future(\Sigma_2)$.

This definition is fulfilled by the law of rGRWf (when suitably formulated, see \cite{Tum07,Tum06}), where the initial data on $\Sigma$ are the wave function on $\Sigma$ and the last flash of each label before $\Sigma$. This definition is intended to exclude that a theory presupposes or generates a foliation, or any other notion of simultaneity-at-a-distance.

\subsection{Existence Theorem in Minkowski Space-Time}
\label{sec:existthm}

Let $(M,g)$ be Minkowski space-time, $M=\RRR^4$, $g=\mathrm{diag}(1,-1,-1,-1)$. Then the Dirac bundle is trivial, $\Dirac = M \times \CCC^4$, and its connection is flat, so that we can replace the covariant derivative $\nabla_\mu$ by the partial derivative $\partial_\mu$.
For $\Sigma\in\Cauchy\cup\Hyper$, $\Hilbert_\Sigma:=L^2(\Dirac|_\Sigma)=L^2(\Sigma,\CCC^4,h,\D^3 x)$, which means the space of measurable functions $\psi:\Sigma \to \CCC^4$ (modulo changes on null sets) that are square-integrable in the sense
\be
\int_\Sigma \D^3x \, \psi^*(x) \, h(x)\, \psi(x)< \infty\,,
\ee
where $h:\Sigma \to \CCC^{4\times 4}$ is a measurable function into the positive definite Hermitian $4\times 4$ matrices that we define below. The scalar product in $L^2(\Sigma,\CCC^4,h,\D^3 x)$ is
\be\label{scpdef}
\scp{\phi}{\psi}_\Sigma = \int_\Sigma \D^3x \, \phi^*(x) \, h(x)\, \psi(x)\,.
\ee
It is clear that $\Hilbert_\Sigma$ is a Hilbert space. Here,
\be\label{hdef}
h(x) = \gamma^0 \, \gamma^\mu \, n_\mu(x)\,,
\ee
where $n^\mu(x)$ is the future-pointing unit normal on $\Sigma$ at $x\in\Sigma$, so normalized that $n^\mu(x) \, n_\mu(x) = 1$. Since $\Sigma$ is $C^\infty$, so are $n_\mu$ and $x\mapsto h(x)$. It is a known fact that the matrix $\gamma^0 \, \gamma^\mu\, n_\mu$ is positive definite for every timelike vector $n_\mu$. The scalar product \eqref{scpdef} can also be written
\be
\scp{\phi}{\psi}_\Sigma = \int_\Sigma \D^3 x\, \overline{\phi}(x) \, \gamma^\mu\, n_\mu(x) \, \psi(x)\,, 
\ee
where $\overline{\phi}(x) = \phi^*(x) \, \gamma^0$ (while $\phi^*$ means component-wise conjugation).  It is a known fact that $\phi \to \overline{\phi}$ is a Lorentz-invariant operation, while $\phi \to \phi^*$ is not. As a consequence, $\Hilbert_\Sigma$ and $\scp{\cdot}{\cdot}_\Sigma$ are defined in a Lorentz-invariant way.

\begin{ass}\label{ass:profile}
The profile function $\profile:[0,\infty) \to [0,1]$ is (Borel) measurable, and 
\be\label{profilefinite}
0< \int_0^\infty  \profile(u) \, e^{\kappa u} \, \D u < \infty
\ee
for every $\kappa>0$.
\end{ass}

This is true, for example, of the Gaussian \eqref{profileGauss}, and when $\profile$ has compact support (and is not almost-everywhere zero). 
The operators $\Lambda_\Sigma(x)$ are defined by \eqref{LambdaSigmadef}.

\begin{ass}\label{ass:Amu}
The 1-form $A:\RRR^4\to \RRR^4$ is time-independent in a suitable Lorentz frame, $C^\infty$, and satisfies
\begin{equation}\label{Amu}
\exists M,\xi>0: \: \forall x \in \RRR^3, \forall \mu:\: |A_\mu(x)| < M (|x|+1)^{-4-\xi}.
\end{equation}
\end{ass}

\begin{thm}\label{thm:Minkexist}
Under Assumptions~\ref{ass:profile} and \ref{ass:Amu} and with $\Hilbert_\Sigma$, $U_\Sigma$, and $\Lambda_\Sigma(x)$ as specified above and $\Lab$ any finite label set, the hypotheses of Theorem~\ref{thm:absexist} are fulfilled. As a consequence, an rGRWf process exists, it is unique in distribution, and the distribution is $\scp{\psi}{\povm(\cdot)\,\psi}$ for a certain POVM $\povm(\cdot)$.
\end{thm}

I conjecture that Assumption~\ref{ass:Amu} is stronger than necessary, in particular that $A_\mu$ does not have to be time-independent.

According to a theorem of Dimock \cite{Dim82}, the Dirac equation defines a unitary isomorphism  $U_\Sigma^{\Sigma'}:\Hilbert_\Sigma \to \Hilbert_{\Sigma'}$ for all Cauchy surfaces $\Sigma,\Sigma'$. The evolution from a Cauchy surface to a hyperboloid is provided by the following lemma. 

\begin{lemma}\label{lemma:unitary}
Under Assumption~\ref{ass:Amu}, the Dirac equation \eqref{Dirac} defines a unitary isomorphism $U_\Sigma^{\Sigma'}:\Hilbert_\Sigma \to \Hilbert_{\Sigma'}$ for all $\Sigma,\Sigma'\in \Cauchy \cup \Hyper$. 
\end{lemma}

\begin{proof}
First note that we assume $m>0$ in the Dirac equation (with $m=0$ this proof would not work).
Choose a Lorentz frame in which Assumption~\ref{ass:Amu} holds (allowing us to identify $M$ with $\RRR^4$), and let $\Sigma_0$ be the 3-surface defined by $t=0$. Since $\D^3x$ on $\Sigma_0$ is just the Lebesgue measure and $h=I$, we write $L^2(\RRR^3,\CCC^4)$ instead of $L^2(\Sigma_0,\CCC^4,h,\D^3 x)$. It suffices to define $U_{\Sigma_0}^\Sigma$ for all $\Sigma \in \Hyper$. We define the $U$ operator first on a dense subspace $S$ of $L^2(\RRR^3,\CCC^4)$, then show that it is bounded and take its bounded extension on all of $L^2(\RRR^3,\CCC^4)$; we leave $S$ to be chosen later but assume $S\subseteq C^\infty(\RRR^3,\CCC^4)$. We define the $U$ operators by solving the Dirac equation for $\psi_0 \in S$ to obtain $\psi: \RRR^4 \to \CCC^4$ on space-time and then restricting $\psi$ to $\Sigma$. By a result of Chernoff \cite{Chernoff}, for $C^\infty$ time-independent  $A_\mu$, the Dirac Hamiltonian is essentially self-adjoint on $C_0^\infty(\RRR^3,\CCC^4)$ (i.e., compactly supported functions), so that there is no ambiguity about $H$, and $\psi\in C^\infty(\RRR^4,\CCC^4)$ if $\psi_0\in C^\infty(\RRR^3,\CCC^4) \cap L^2(\RRR^3,\CCC^4)$. As a consequence, no ambiguity about changing $\psi$ on null sets arises, and $\psi_\Sigma := \psi|_\Sigma$ is well defined. By the linearity of the Dirac equation, $\psi_0 \to \psi_\Sigma$ is linear. We write $\|\psi\|_\Sigma$ for $\scp{\psi_\Sigma}{\psi_\Sigma}_\Sigma^{1/2}$.

We now show that
\be\label{psinormineq}
\|\psi\|_\Sigma \leq \|\psi_0\|\text{ for }\Sigma\in\Hyper\,.
\ee 
Without loss of generality, we assume $\|\psi_0\|=1$. Define the \emph{probability current vector field}\footnote{This is standard terminology. In rGRWf, of course, it does not signify the flow of probability. In Bohmian mechanics it does.}  $j:\RRR^4 \to\RRR^4$ by
\be
j^\mu = \overline{\psi} \, \gamma^\mu\, \psi
\ee
and note that, for any spacelike 3-surface $\Sigma$,
\be
\|\psi\|_\Sigma^2 = \int_\Sigma \D^3x\, j^\mu(x) \, n_\mu(x)
\ee
is the flux of $j$ across $\Sigma$. Some well-known properties of $j$: (i)~Since $h$ in \eqref{scpdef} is positive definite, $j$ has positive Lorentzian scalar product $j^\mu n_\mu$ with every future-pointing timelike $n^\mu$, and thus $j$ is future-pointing causal. (ii)~$j$ is divergence free, i.e., the continuity equation
\be
\partial_\mu j^\mu = 0
\ee
holds as a consequence of the Dirac equation. (iii)~Since $j^\mu(x) \, n_\mu(x) \geq 0$, $\D^3x \, j^\mu(x) \, n_\mu(x)$ is a $\sigma$-finite measure $\nu$ on $\Sigma$.

We use the following fact \cite{Tum01}: \emph{Given a future-pointing causal $C^\infty$ divergence-free vector field $j$ on $\RRR^4$ whose flux across $\Sigma_0=\{0\}\times \RRR^3$ is 1 and a spacelike 3-surface $\Sigma$, then for every measurable $A\subseteq \Sigma$,}
\be\label{rwl}
\PPP(L \cap A \neq \emptyset) = \nu(A\setminus B_0) \,,
\ee
\textit{where $L$ is the random Bohmian trajectory, i.e., integral curve of $j$, whose initial point has distribution $|\psi_0|^2 \D^3x$ on $\Sigma_0$; $\PPP(L \cap A \neq \emptyset)$ is the probability of the Bohmian trajectory intersecting $A$; and $B_0$ is the set of points $x\in\Sigma$ with $\psi(x)\neq 0$ which do not lie on any Bohmian trajectory starting on $\Sigma_0$.} Note that the hypotheses on $j$ are fulfilled in our case; note also that the Bohmian trajectories are causal since $j$ is, and thus a spacelike 3-surface intersects each Bohmian trajectory at most once. We thus obtain a stochastic interpretation of the flux across $A$ as the probability of the random curve $L$ intersecting $A$, but we need to get control of the set $B_0$. 

To this end, we use the global existence theorem of Teufel and myself \cite{TT05} for Bohmian trajectories, which implies the following: \emph{Given an electromagnetic potential $A_\mu$ on $\RRR^4$ that is time-independent and $C^\infty$, and an initial wave function $\psi_0 \in L^2(\RRR^3,\CCC^4)\cap C^\infty(\RRR^3,\CCC^4)$, then almost all Bohmian trajectories exist for all times, where ``almost all'' refers to the $|\psi_0|^2$ distribution over the initial point on $\{0\}\times\RRR^3$.} From this 
we get control of $B_0\subseteq \Sigma \in \Cauchy\cup\Hyper$, namely
\be\label{B0control}
\nu(B_0)=0\,.
\ee
For $\Sigma$ of the form $\Sigma_t := \{t\}\times\RRR^3$, this would be immediate from the global existence theorem, applied to $\Sigma_t$ as the initial time, noting that the $|\psi_t|^2$ distribution coincides with $\nu$. For $\Sigma\in\Hyper$, we have to do some work: Let  
$B_t$ be the set of points $x\in\Sigma$ such that $\psi(x) \neq 0$ (so that there exists a trajectory through $x$) and the trajectory through $x$ does not exist at time $t$, i.e., does not intersect $\Sigma_t := \{t\}\times\RRR^3$. 
Now we ask for the probability that the random trajectory $L$ starting on $\Sigma_0$ intersects $\Sigma$ in a point $x\in B_t$. In that event, $L$ has to coincide with the unique trajectory through $x$, which does not intersect $\Sigma_t$ and thus does not exist globally. By the global existence theorem, this probability is zero:
\be
0=\PPP(L\cap B_t \neq \emptyset) = \nu(B_t\setminus B_0)\,,
\ee
where the second equality is \eqref{rwl}.  
Now choose an arbitrary measurable set $A\subseteq \Sigma$ with $\nu(A)<\infty$ and consider $A\cap B_{t_1} \cap \ldots \cap B_{t_m}$ instead of $B_t$ and observe that
\be
\nu\bigl((A\cap B_{t_1} \cap \ldots \cap B_{t_m})\setminus B_0\bigr) 
\leq \nu(B_{t_1}\setminus B_0) =0\,.
\ee
Put differently,
\be
\nu\bigl(A\cap B_{t_1} \cap \ldots \cap B_{t_m}\bigr) = 
\nu\bigl(A\cap B_{t_1} \cap \ldots \cap B_{t_m}\cap B_0\bigr)\,.
\ee
By the same argument for any time $t_{m+1}$ instead of 0,
\be
\nu\bigl(A\cap B_{t_1} \cap \ldots \cap B_{t_m}\bigr) = 
\nu\bigl(A\cap B_{t_1} \cap \ldots \cap B_{t_m}\cap B_{t_{m+1}}\bigr)\,.
\ee
Setting $t_1=0$ and by induction along $m\in \NNN$,
\be
\nu(A\cap B_0) = \nu\bigl(A\cap B_0 \cap B_{t_1} \cap \ldots \cap B_{t_m}\bigr)\,.
\ee
Now consider an infinite sequence $(t_m)$ that is dense in $\RRR$ (say, an enumeration of $\QQQ$); then
\be
\nu(A\cap B_0) = \lim_{m\to\infty} 
\nu\bigl(A\cap B_0 \cap B_{t_1} \cap \ldots \cap B_{t_m}\bigr) =
\nu\Bigl(A\cap B_0 \cap \bigcap_{m=1}^\infty B_{t_m} \Bigr) = 0
\ee
because $\bigcap_m B_{t_m}=\emptyset$, as every trajectory exists for some time interval of positive length (and thus, e.g., at some rational time). Since $A$ was arbitrary with finite measure, $\nu(B_0)=0$, which is what we claimed in \eqref{B0control}.

As a consequence of \eqref{B0control}, we have from \eqref{rwl} that
\be
1\geq\PPP(L \cap \Sigma \neq \emptyset) = \nu(\Sigma)=\|\psi\|_\Sigma^2 \,,
\ee
which shows \eqref{psinormineq}.

Now we show that $\|\psi\|_\Sigma = \|\psi_0\|$ for $\Sigma \in\Hyper$.
For this we use the flux-across-surfaces theorem of D\"urr and Pickl \cite{DP03}, which implies the following: \emph{Under Assumption~\ref{ass:Amu}, for all $\psi_0$ with $\|\psi_0\|=1$ from a suitable dense subspace $S$ of $L^2(\RRR^3,\CCC^4)$ with $S\subseteq C^\infty(\RRR^3,\CCC^4)$ it is true that}
\begin{equation}\label{fast}
\lim_{s\to\infty} \int_{\hyper_s(0)} \D^3x \, j^\mu\, n_\mu = 1\,.
\end{equation}
This fixes the subspace $S$ (and this is where the condition \eqref{Amu} enters). We want to show that $P_s :=\|\psi\|_{\hyper_s(0)}^2 = 1$ for every $s>0$ and $\psi_0\in S$ with $\|\psi_0\|=1$. This quantity is the probability that the random Bohmian trajectory $L$ intersects $\hyper_s(0)$. In particular, it is decreasing in $s$,
\begin{equation}
P_{s_1} \geq P_{s_2} \quad\text{if }s_1\leq s_2\,.
\end{equation}
Now, according to \eqref{fast}, $P_s \to 1$ as $s\to\infty$ while $P_s \leq 1$, and thus $P_s =1$ for all $s>0$.

What we have obtained is that, for any $\Sigma\in\Hyper$, $U:=U_{\Sigma_0}^\Sigma: S \to \Hilbert_\Sigma$ is norm-preserving. It is therefore bounded and possesses a unique bounded extension $\tilde{U}$ to all of $\Hilbert_{\Sigma_0} = L^2(\RRR^3,\CCC^4)$. To see that $\tilde{U}$ is norm-preserving, too, consider a convergent sequence $\psi_n \to \psi$ with $\psi_n\in S$ and note that
\be
\|\tilde{U}\psi\|_\Sigma =
\|\lim_{n\to\infty} U \psi_n \|_\Sigma = \lim_{n\to\infty} \|U\psi_n\|_\Sigma = 
\lim_{n\to\infty} \|\psi_n\| = \|\lim_{n\to\infty}\psi_n\|=\|\psi\|\,.
\ee
In the following we write $U_{\Sigma_0}^{\Sigma}$ for $\tilde{U}$.

Now we show that $U:=U_{\Sigma_0}^{\Sigma}$ is onto. We first observe that the range of $U$ is a closed subspace because if, in $\Hilbert_\Sigma$, $\phi_n\to\phi$ and $\phi_n=U\psi_n$ then $(\phi_n)$ is a Cauchy sequence, and thus so is $(\psi_n)$, and thus $(\psi_n)$ converges, and $U\lim_n \psi_n = \lim_n U\psi_n= \lim_n \phi_n = \phi$, so that $\phi$ lies in the range. It remains to show that the range of $U$ is dense in $\Hilbert_\Sigma$: The range of $U$ contains $C_0^\infty(\Sigma,\CCC^4)$ (i.e., compactly supported) because for such a $\psi_\Sigma$ there exists a Cauchy surface $\Sigma'$ that has the support of $\psi_\Sigma$ in common with $\Sigma$. By Dimock's existence theorem, there is a unique $\psi:\RRR^4\to\CCC^4$, solving the Dirac equation, whose restriction to $\Sigma'$, and thus to $\Sigma$, is $\psi_\Sigma$. Set $\psi_0$ to be the restriction of $\psi$ to $\Sigma_0$.
\end{proof}

\bigskip

\begin{proofthm}{thm:Minkexist}
To begin with, Assumption~\ref{ass:tdist} is satisfied in Minkowski space-time, and the $U_\Sigma$ operators are provided by Lemma~\ref{lemma:unitary}. 
Now we show that the quantity $\mathcal{N}(y)$ is always well defined by \eqref{Ndefprofile}, which means that $\int_\Sigma \D^3x \, \profile\bigl(\sdist_\Sigma(x,y)\bigr)$ is finite and nonzero: It could only be zero if $\profile$ were zero almost everywhere, which is excluded by the positivity in \eqref{profilefinite}. To check that it is finite, we only need check that it is finite for $x'=0$ and $y=(s,0,0,0)$ since there is an isometry of Minkowski space carrying $\hyper_s(x')$ is into $\hyper_s(0)$ and $y$ into $(s,0,0,0)$. In particular, $\mathcal{N}(y)$ is actually independent of $y$ (in Minkowski space-time!). Now we calculate 
\be
\int_{\Sigma=\hyper_s(0)} \D^3x \, \profile\circ\sdist_\Sigma\bigl(x,(s,0,0,0)\bigr) =
\ee
[by \eqref{hyperdensity}] 
\be
=\int_{\RRR^3} \D x^1\, \D x^2 \, \D x^3\,  \frac{\profile\circ\sdist_\Sigma
\bigl((\sqrt{s^2+r^2}, x^1,x^2,x^3),(s,0,0,0) \bigr)}{\sqrt{1+r^2/s^2}}=
\ee
\be
=\int_{\RRR^3} \D x^1\, \D x^2 \, \D x^3\,  \frac{\profile\bigl(s\sinh^{-1}(r/s)\bigr)}
{\sqrt{1+r^2/s^2}}=
\ee
[where $\sinh^{-1}$ means the inverse function of $\sinh$]
\be
=\int_0^\infty \D r\,  \profile\bigl(s\sinh^{-1}(r/s)\bigr) \, \frac{4\pi r^2}
{\sqrt{1+r^2/s^2}}=
\ee
[substituting $r=s\sinh (u/s)$ so that $\D u = \D r/\sqrt{1+r^2/s^2}$]
\be
=4\pi s^2\int_0^\infty \D u\,  \profile(u) \,  \sinh(u/s)^2 \,.
\ee
Since $\profile$ is a bounded function, what is relevant for finiteness of this integral is the asymptotics for $u\to \infty$, where $\sinh \sim \tfrac12 \exp$ and thus $\sinh(u/s)^2 \sim \tfrac14 \exp(2u/s)$. Thus, the finiteness in \eqref{profilefinite} is (necessary and) sufficient for the finiteness of this integral for every $s>0$.

The operators $\Lambda_\Sigma(x)$, defined by \eqref{LambdaSigmadef}, are weakly measurable as a function of $x\in\Sigma=\hyper_s(x')$ whenever $(x,y)\mapsto \profile\circ\sdist(x,y)$ is measurable. This is satisfied since for future hyperboloids in Minkowski space-time, $(x,y)\mapsto \sdist(x,y)$ is a measurable (even $C^\infty)$ function, $\profile:[0,\infty)\to[0,1]$ is measurable by Assumption~\ref{ass:profile}, and $\mathcal{N}(y)$ is actually independent of $y$.

To check \eqref{Lambdanorm} for $\Sigma = \hyper_s(x')$ and arbitrary $\psi\in\Hilbert_\Sigma$,
\be
\int_{\Sigma} \D^3 x\, \scp{\psi}{\Lambda_\Sigma(x)\,\psi} =
\ee
\be
=\int_{\Sigma} \D^3 x\, \int_{\Sigma} \D^3y\, \overline{\psi}(y) \, 
\gamma^\mu \, n_\mu(y) \, \lambda \, \mathcal{N}(y) \,  \profile\bigl(\sdist_\Sigma(x,y)\bigr) \, \psi(y) =
\ee
[we can reorder the integrals because the integrand is nonnegative]
\be
= \lambda\int_{\Sigma} \D^3y\, \overline{\psi}(y) \, 
\gamma^\mu \, n_\mu(y) \,   \psi(y)\, \mathcal{N}(y)\int_{\Sigma} \D^3 x\,\profile\bigl(\sdist_\Sigma(x,y)\bigr)  =
\ee
\be
= \lambda\int_{\Sigma} \D^3y\, \overline{\psi}(y) \, 
\gamma^\mu \, n_\mu(y) \, \psi(y)=\lambda \, \scp{\psi}{\psi}\,.
\ee

We now show the measurability of \eqref{hypermeas}. To this end, we define, for every hyperboloid $\hyper_s(x)$, a diffeomorphism $\varphi_{s,x}:\hyper_s(x) \to\RRR^3$ by $\varphi_{s,x}(y) = (y^1-x^1,y^2-x^2,y^3-x^3)$. This induces a linear mapping $M_{s,x}:L^2(\RRR^3,\CCC^4) \to \Hilbert_{\hyper_s(x)}$ 
defined by $M_{s,x}\psi(y) = \psi\bigl(\varphi_{s,x}(y)\bigr)$; $M_{s,x}\psi$ is square-integrable because
\be
\|M_{s,x}\psi\|^2_{\hyper_s(x)} = \int_{\hyper_s(x)} \D^3y \, (M_{s,x}\psi)^*(y) \, \gamma^0\gamma^\mu\, n_\mu(y) \, (M_{s,x}\psi)(y)= 
\ee
\be
=\int_{\RRR^3}\D^3v\, \psi^*(v) \, \gamma^0\gamma^\mu\, (1,v/\sqrt{s^2+v^2})_\mu \, \psi(v) \leq \int_{\RRR^3}\D^3v\, |\psi(v)|^2 \, \sum_{\mu=0}^3 \|\gamma^0\gamma^\mu\|_{\CCC^4} < \infty \,,
\ee
which indeed implies $\|M_{s,x}\|\leq (\sum_\mu \|\gamma^0\gamma^\mu\|)^{1/2}$. Similarly, $M_{s,x}^{-1}\psi(v) = \psi(\varphi_{s,x}^{-1}(v))$ is a bounded operator. We check that $(x,x')\mapsto M_{\tau(x,x'),x'}^{-1} \, \Lambda_{\hyper(x,x')}(x) \, M_{\tau(x,x'),x'}$ is weakly measurable:
\be
\scp{\psi}{M_{\tau(x,x'),x'}^{-1} \, \Lambda_{\hyper(x,x')}(x) \, M_{\tau(x,x'),x'}\, \psi} =
\ee 
\be
=\int_{\RRR^3} \D^3 v\, \psi^*(v) \psi(v)\, \lambda\mathcal{N} \, 
\profile\bigl(\sdist(x,\varphi_{\tau(x,x'),x'}(v)\bigr)
\ee
which is measurable since the integrand is measurable in $(x,x',v)$. It remains to show that $(x,x')\mapsto M_{\tau(x,x'),x'}^{-1} \, U_{\Sigma_0}^{\hyper(x,x')}$ is weakly measurable. By a translation $x'\to 0$, it suffices to show that $s\mapsto \scp{\psi_0}{M_{s,0}^{-1} \, U_{\Sigma_0}^{\hyper_s(0)} \, \psi_0}$ is measurable for all $\psi_0\in L^2(\RRR^3,\CCC^4)$, which  follows  (since the operators are bounded) from the fact that $s\mapsto M_{s,0}^{-1} \, U_{\Sigma_0}^{\hyper_s(0)} \, \psi_0(v) = \psi(\varphi_{s,0}^{-1}(v))$ is continuous for all $v\in\RRR^3$ and $\psi_0\in S$, as then $\psi:\RRR^4\to\CCC^4$ is $C^\infty$.

Thus, Assumption~\ref{ass:LambdaSigma} is fulfilled, too, and Theorem~\ref{thm:absexist} applies.
\end{proofthm}

\section{Outlook}

\subsection{Nonlocality}
\label{sec:nonlocality}

Locality means that if two space-time regions $A$ and $B$ are spacelike separated then events in $A$ cannot influence those in $B$ or vice versa. Let me point out why rGRWf is a \emph{nonlocal} theory.

rGRWf specifies the joint distribution of flashes, some of which may occur in $A$ and some in $B$. The distribution of those in $A$, i.e., of how many flashes occur in $A$ and at which space-time points, is in general not independent of the flashes in $B$ (except in case the initial state vector factorizes):
\be
\PPP\bigl(F\cap A\in \cdot\big| F\cap B\bigr) \neq \PPP(F\cap A\in \cdot)\,.
\ee
But this is not yet an influence between $B$ and $A$: correlation is not causation. After all, the flashes in $A$ and those in $B$ may be correlated because of a common cause in the past. Taking this into account, the criterion for the absence of an influence between $A$ and $B$ is that $F\cap A$ and $F\cap B$ are \emph{conditionally independent}, given the history of their common past. And also this can fail in rGRWf:
\be
\PPP\Bigl(F\cap A\in\cdot\Big|F\cap B, F\cap \past(A)\cap \past(B)\Bigr)\neq
\PPP\Bigl(F\cap A\in\cdot  \Big| F\cap\past(A)\cap\past(B) \Bigr)\,.
\ee
Thus, rGRWf is nonlocal.

The nonlocality of rGRWf should be seen in connection with Bell's famous nonlocality argument \cite{Bell64,Bell87b}, according to which the laws of our universe must be nonlocal. The argument shows that every local theory entails that the predicted probabilities for certain experiments satisfy Bell's inequality, which however is violated according to the quantum formalism and in experiment (and according to rGRWf).

Many authors, beginning with Einstein, Podolsky and Rosen \cite{EPR35}, have expressed the view that locality follows from relativistic covariance. This view seems dubious given Bell's result that locality is wrong while relativity has been extraordinarily successful. More detailed arguments to the effect that nonlocality does not contradict relativity (or, in other words, that the concept of locality is not equivalent to that of relativistic covariance) have been given in \cite{Timbook,GT03}. The strongest argument to this effect that I see is, however, the existence of rGRWf, a nonlocal theory that is convincingly covariant. 

Indeed, the biggest hurdle on the way to a relativistic quantum theory without observer was to find a theory that is nonlocal yet covariant. Thus, this is perhaps the most remarkable aspect of rGRWf. So how does rGRWf accomplish this feat? How does it reconcile relativity and nonlocality? I think that the following point, which I have first described in \cite{Tum07}, is crucial: If space-time regions $A$ and $B$ are spacelike separated, then nonlocality means that events in $A$ can influence those in $B$ \emph{or vice versa}. Of course, an influence from $A$ to $B$ would mean an influence to the past in some Lorentz frames. In rGRWf, however, the words ``or vice versa'' are important, as in rGRWf there is no objective fact about whether the influence took place from $A$ to $B$ or from $B$ to $A$. The rGRWf laws simply prescribe the joint distribution of flashes in $A$ and $B$, but do not say that nature made the first random decision in $A$, which then influenced the flashes in $B$. There is no need for rGRWf to specify in which order to make random decisions. One can say that the direction of the influence depends on the chosen Lorentz frame. In a frame in which $A$ is earlier than $B$ one would conclude that the flashes in $A$ have influenced those in $B$, while in a frame in which $B$ is earlier than $A$ one would conclude the opposite. The following simple illustration of how an influence can fail to have a direction is due to Conway and Kochen \cite{CK05}. 

\begin{ex}
Consider a discrete space-time $M$ as depicted in Fig.~\ref{fig:CKspacetime}, which can be thought of as a subset of $1+1$-dimensional Minkowski space. In terms of a suitable time coordinate function $T$, all space-time points have positive integer values of $T$, and at time $T$ there exist $T$ space points. The PO is a field $\phi:M \to \{0,1\}$ subject to two laws: (i)~If $x$ is any point in $M$ and $y,z$ its two neighbors in the future then $\phi(x) + \phi(y) + \phi(z) \in\{0,2\}$. (ii)~Given all values of $\phi$ up to time $T'$, the random event $\phi(x)=1$ has conditional probability $1/2$ for any point $x$ with $T(x)>T'$.

\begin{figure}[ht]
\begin{center}
\includegraphics[width=.6 \textwidth]{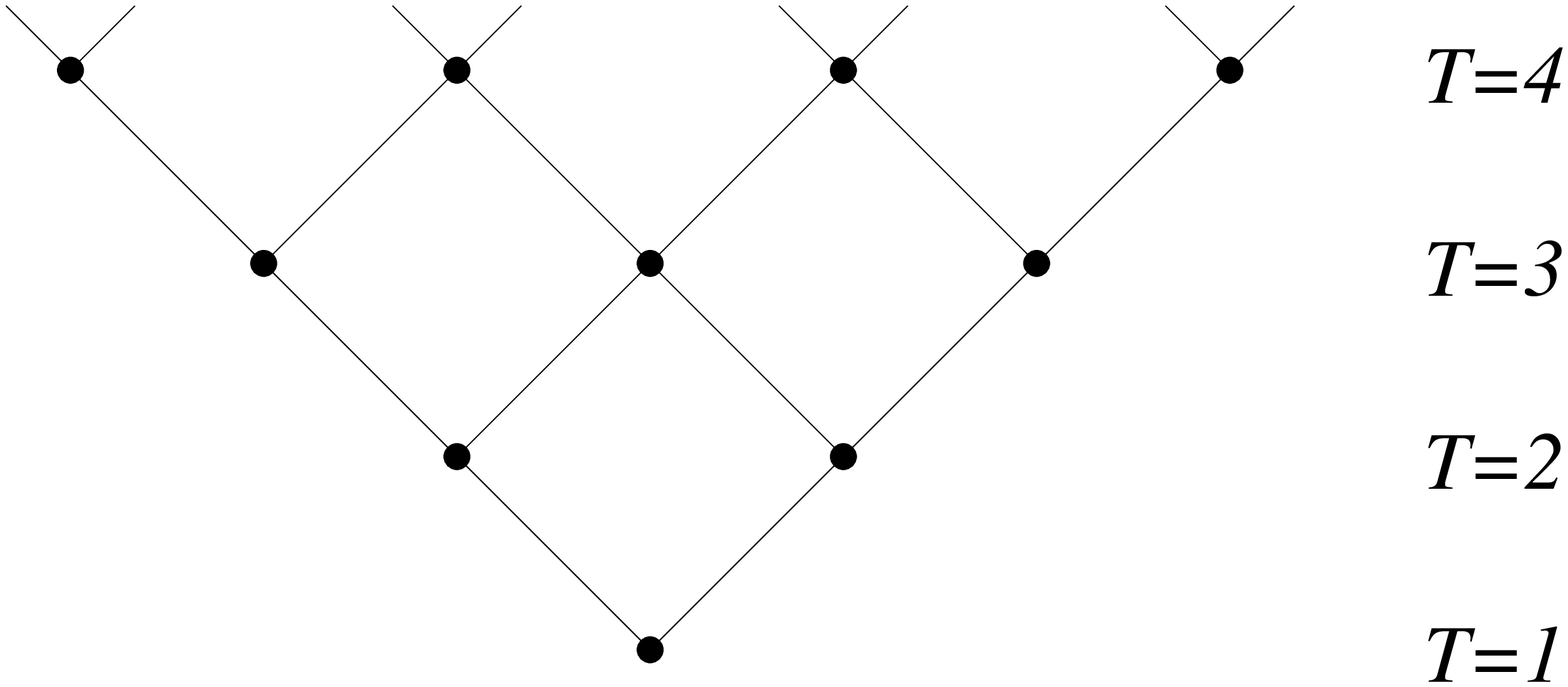}
\caption{The discrete space-time considered in the text, and the $T$ function on it. The bullets symbolize the space-time points, while the lines have no physical meaning and serve only for indicating how to continue the figure to infinity.}
\label{fig:CKspacetime}
\end{center}
\end{figure}

Let us generate a random space-time history according to these laws. On the one point $x$ with $T(x)=1$ we choose $\phi(x)$ at random according to (ii), with probability $1/2$ for $\phi(x)=1$. Then we can choose, for the left point $y$ with $T(y)=2$, the value $\phi(y)$, again with probability $1/2$ for $\phi(y)=1$. Then, by (i), for the right point $z$ with $T(z)=1$, the value $\phi(z)$ is determined by $\phi(x)$ and $\phi(y)$. Similarly, if we have chosen all $\phi$ values up to time $T'$ then any single $\phi$ value in the row $T'+1$ will determine all the other values in this row.

This model world is not meant to be relativistic, but it illustrates influences without direction: Suppose we simulate the model one time step after another, and suppose we have filled in the $\phi$ values up to time $T'$. Let $x$ be the leftmost point at time $T'+1$, and $y$ the rightmost one. Now we may throw a coin to choose $\phi(x)$, and then compute all the other $\phi$ values in that row. Or we may throw a coin for $\phi(y)$ and compute $\phi(x)$ from that. In one case there was an influence from $x$ to $y$, in the other from $y$ to $x$. But there is no objective direction of the influence in the model world. The theory specifies no such direction, and there is no need to specify it. For a physical theory it suffices to specify the joint probability distribution of the history of the PO. The direction of the influence lies only in the way we choose to look at, or simulate, the model world, like a choice of gauge or a choice of coordinates; it represents no objective fact in the world. The situation is the same as any other situation of simulating two dependent random variables $X,Y$ with known joint distribution: One could first simulate $X$ according to its known marginal distribution and then $Y$ according to its known conditional distribution given $X$, or vice versa, and none of these two orderings is more correct than the other.
\end{ex}

\subsection{Other Approaches to Relativistic Collapse Theories}\label{sec:other}

In this subsection, I mention the approaches to relativistic collapse theories other than rGRWf in the literature, and describe the differences.

A crucial part of the problem of specifying a relativistic collapse theory is to specify a law for the primitive ontology. The need for a clear specification of the primitive ontology has often not been sufficiently appreciated in the literature. Many authors have focused on the problem of specifying a Lorentz-invariant law that associates with every spacelike 3-surface $\Sigma$ in space-time a wave function $\psi_\Sigma$, in such a way that macroscopic superpositions collapse appropriately (e.g., \cite{Gi89,Pe90,GGP90, Pe99,NR}). But such a law is only half of what is needed for a relativistic collapse theory: the other half concerns the primitive ontology.

\medskip

\underline{\textit{Dowker and Henson}} \cite{Fay02} describe a collapse model on a lattice space-time $\ZZZ^2$ in $1+1$ dimension. This model has many traits in common with rGRWf (except that rGRWf lives on manifolds). In particular, it is relativistic in the appropriate lattice sense,  
and it defines a primitive ontology consisting of field values at the lattice sites (a primitive ontology not among the examples I listed in Section~\ref{sec:phil}). In contrast to rGRWf, this model incorporates interaction while rGRWf assumes non-interacting ``particles'' (of course, there are no particles in this theory, just flashes). An important future goal for rGRWf is the development of a version with interaction.

\medskip

\underline{\textit{Hellwig and Kraus}} \cite{HK70} worry about the relativistic invariance of wave function collapse in ordinary quantum mechanics and propose that wave functions collapse along the past light cone of the space-time point at which a measurement takes place. They assume as given the space-time points $X_1,\ldots,X_n$ at which measurements take place (some of which may be spacelike separated) and the observables $O_1,\ldots,O_n \in\Bdd(\Hilbert)$ measured there with results $R_1,\ldots,R_n \in \RRR$ and associate with every $x\in M$ a collapsed state vector $\psi_x\in \Hilbert$. In detail, they assume the Heisenberg picture in which the unitary evolution of the state vector disappears; let $P_k$, for $k=1,\ldots,n$, be the projection to the eigenspace of $O_k$ with eigenvalue $R_k$ and set
\be\label{HKrule}
\psi_x = \frac{\bigl(\prod_{k:X_k\in\past(x)} P_k\bigr) \psi}
{\Bigl\|\bigl(\prod_{k:X_k\in\past(x)} P_k\bigr) \psi\Bigr\|} \in \Hilbert\,,
\ee
where $\psi$ is the initial state vector, an empty product is understood as the identity operator, and the ordering in the product is such that whenever $X_k\in\future(X_\ell)$ then $P_k$ stands to the left of $P_\ell$. It is assumed that for spacelike separated $X_k$ and $X_\ell$, $O_k$ commutes with $O_\ell$, and thus $P_k$ with $P_\ell$. [I mention that in \cite{HK70}, the term $\mathrm{Tr}(QPW)$ in equations (3)--(5) should read $\mathrm{Tr}(QPWP)$.]

This rule involves a kind of retrocausation, as the decision, made by an observer at $X_k$, about which $O_k$ to measure influences the reality in the past, more precisely at those points $x$ that are spacelike separated from $X_k$ and that therefore are earlier than $X_k$ in some inertial frames. Even more problematic is that the use of the proposal of Hellwig and Kraus remains unclear, for two reasons. 

First, in ordinary quantum mechanics the formalism is usually supposed to specify the joint probability distribution of the results $R_k$, which follows from the conventional quantum formalism (with instantaneous collapse at every measurement)
\be\label{PRP}
\PPP(R_1=r_1,\ldots,R_n=r_n) = \Bigl\|\bigl(\prod_{k=1}^n P_k\bigr) \psi\Bigr\|^2
\ee
with $P_k$ the projection to the eigenspace with eigenvalue $r_k$, and the ordering of the factors in the product as before (whenever $X_k\in\future(X_\ell)$ then $P_k$ is left of $P_\ell$, while for spacelike separated $X_k$ and $X_\ell$, $P_k$ commutes with $P_\ell$). Formula \eqref{PRP} is manifestly Lorentz invariant, and since the measurement results constitute (in a vague and imprecise way) the primitive ontology of ordinary quantum mechanics it suffices that their distribution be specified by the laws of the theory in a Lorentz-invariant manner, making a rule like \eqref{HKrule} irrelevant. 

Second, instead of defining a state vector $\psi_x$ for every space-time point $x$ it seems more natural to define a state vector $\psi_\Sigma$ for every spacelike 3-surface $\Sigma$ (even for a single particle in the presence of collapses, be they due to flashes or to measurements). Indeed, such is the case in rGRWf (and in the model of Dowker and Henson \cite{Fay02}), so it certainly does not conflict with relativistic invariance (as Hellwig and Kraus seem to think). The notion of a state vector $\psi_\Sigma$ for every surface $\Sigma$ is, of course, much older; it is used by Tomonaga and Schwinger in the 1940's, and implicit in the derivation of \eqref{PRP}. If $\psi$ is admitted to depend on $\Sigma$ then the apparent conflict between instantaneous collapse and relativity evaporates: it is then completely consistent that $\psi$ collapses instantaneously (on all of 3-space) \emph{in every Lorentz frame} because the collapse is associated with some space-time point $X$, and $\psi_\Sigma$ is a \emph{collapsed} state vector on every spacelike 3-surface $\Sigma$ with $X\in\past(\Sigma)$ but \emph{uncollapsed} on every $\Sigma$ with $X \in\future(\Sigma)$. In contrast, for the primitive ontology at a space-time point $x$ it would not make sense to depend on a 3-surface $\Sigma$.

\medskip

\underline{\textit{Dove and Squires}} \cite{dovepaper,dovethesis} essentially reiterate the ideas of Hellwig and Kraus in the context of a GRW theory with flash ontology. They propose a Lorentz-invariant rule for collapsing the wave function \emph{given} the flashes, but no law for the flashes given the initial wave function. That is, what they provide is, at best, a part of a collapse theory. Furthermore, their proposal is based on the misconception that they have to define the value $\psi(x)$ of the wave function for every space-time point $x$ (if the system consists of a single particle, $N=1$). I have discussed this already above in the context of Hellwig and Kraus's proposal.

\medskip

\underline{\textit{Blanchard and Jadczyk}} \cite{BJ96} start from the consideration of a system of quantum particles continuously observed by detectors of limited efficiency, which manage only every now and then to detect a particle. This consideration is related to GRW theory as the detection events are points in space-time, and are reasonably modeled in a stochastic way by a point process in space-time whose distribution may coincide with that of a GRWf process. To obtain a relativistic version of this model, one might try to analyze the behavior of detectors consisting of relativistic particles, but Blanchard and Jadczyk instead try to guess relativistic equations. What they guess is not related to rGRWf, and in fact does not answer the question of the probability distribution of the detection events. They consider a wave function $\Psi_\tau$ on space-time that, instead of being a solution to the Dirac equation, \emph{evolves}. That is, the wave function is not a function on space-time but a one-parameter family of functions on space-time, where the parameter $\tau$ is a pseudo-time, anyway a fifth coordinate (in addition to the four space-time coordinates). I do not see why a theory based on such a wave function should lead to any predictions related to those of quantum mechanics. In Blanchard and Jadczyk's model of detection, they propose a stochastic rule for a random $\tau$ value associated with the detection event, but no rule for a random space-time point. Moreover, this rule is not Lorentz invariant but assumes a preferred frame, which they call the rest frame of the detector. That may seem natural when modeling a detector, but it would not be admissible for a relativistic theory of flashes.

\medskip

\underline{\textit{Ruschhaupt}} \cite{Ru02} continues where Blanchard and Jadczyk have stopped. His contribution is to associate a space-time point with the detection event as follows: he assumes that a world line $s\mapsto x(s)$ of the detector is given, parameterized with proper time, and when Blanchard and Jadczyk's rule generates a random value $\tau$ of the pseudo-time, Ruschhaupt inserts this value into $x(\cdot)$ to obtain a random space-time point $x(\tau)$. Since the world line $x(\cdot)$ is given, this model, unlike rGRWf, does not qualify as a fundamental theory. On top of that, I see no reason why the predictions of this model should be related to those of quantum mechanics.

\medskip

\underline{\textit{Conway and Kochen}} \cite{CK05} claim to have shown that relativistic GRW theories are impossible. rGRWf is a counterexample to their claim; the model of Dowker and Henson \cite{Fay02} is another counterexample. I have given a detailed evaluation of their arguments in \cite{Tum07}; see \cite{BG06} for a further critique, and \cite{CK07} for Conway and Kochen's reply to \cite{BG06} and \cite{Tum07}. Here is a summary of \cite{Tum07}: Conway and Kochen claim that the impossibility of relativistic GRW theories is a corollary of a physical statement they derive in \cite{CK05} and call the ``free will theorem''; it is intended to exclude deterministic theories of quantum mechanics. The proof of the free will theorem contains a logical gap in the sense that it uses a hypothesis that is stronger than formulated in the statement of the ``theorem.'' The weaker version of the hypothesis (``FIN'' or ``effective locality'') is, in fact, fulfilled by rGRWf, while the stronger one is violated. The stronger version is equivalent to locality (in the sense of Einstein, Podolsky, Rosen, and Bell \cite{Bell87b}, and in the sense of Section~\ref{sec:nonlocality} above), which was shown by Bell in 1964 \cite{Bell64} to conflict with certain probability distributions predicted by quantum mechanics and afterwards confirmed in experiment. Thus, EPRB locality is wrong in our world, making a theorem assuming it useless. (However, the Conway--Kochen proof could be turned around into a disproof of EPRB locality, assuming determinism \cite{BG06}.) Moreover, Conway and Kochen's argument from the free will theorem to the impossibility of relativistic GRW theories supposes that every stochastic theory is equivalent to a deterministic one (by making all random decisions at the initial time), which in this case is incorrect in a relevant way because the probability distribution in rGRWf depends on the external field $A_\mu$, which observers are free to influence at later times.

\subsection{The Value of a Precise Definition}
\label{sec:gooddef}

In the introduction I said that the GRW theory provides a precise definition of quantum mechanics. As always with precise definitions, it is easy to find many physicists who will honestly declare that they don't need such a definition for their work. So I should give an example of what such a definition is good for.

The example consists of a simple physical statement that one would like to prove, and a simple proof based on GRW theory (with flash ontology) as a precisely defined theory. (By the way, this simple proof appears here for the first time in print.) However, from the rules of ordinary quantum mechanics it is impossible to get anywhere near a proof. The statement is this: 
\be\label{mainthmpovm}
\begin{array}{c}
\text{For every conceivable experiment that one could carry out on a physical}\\
\text{system there is a POVM $E(\cdot)$ so that the probability distribution of the}\\
\text{result $R$ is $\scp{\psi}{E(\cdot) \, \psi}$, where $\psi$ is the system's wave function.}
\end{array}
\ee
Below we show that this is true in a (hypothetical) world governed by GRWf, for any choice of Hamiltonian and flash rate operators (while $E(\cdot)$ depends on this choice, of course); we will translate the physical statement \eqref{mainthmpovm} into a mathematical one and give a proof. 

What is the status of \eqref{mainthmpovm} in ordinary quantum mechanics? There, one introduces as an \emph{axiom} (rather than theorem) that observables correspond to self-adjoint operators, and specifies the distribution of the result if an observable is measured, and a formula for the subsequent collapse of the wave function. But it is well known that not every conceivable experiment is the measurement of an observable: Self-adjoint operators correspond to \emph{projection-valued measures} (PVMs), which are POVMs $P(\cdot)$ such that $P(A)$ is a projection for every measurable set $A$; it is easy to name experiments whose POVMs $E(\cdot)$ are not a PVMs, for example a cascade of several measurements corresponding to non-commuting operators, or a ``time-of-arrival measurement'' observing the time a detector clicks. Thus, the usual axioms of quantum mechanics do not exhaust all conceivable experiments. One is tempted to introduce \eqref{mainthmpovm} as a further axiom. 

Let us return to GRWf theories. To translate \eqref{mainthmpovm} into a mathematical statement, we note that the result of an experiment will be read off from the arrangement of matter in space and time, that is, from the primitive ontology. Thus, the result $R$ is a function of the random pattern of flashes $F$, $R=\zeta(F)$. (Note that we do not model a class of experiments, but claim that any experiment deserving the name must be of this form.) We assume that $\zeta$ is a measurable function from the appropriate history space $\Omega$ (such as $M^\NNN$) to the value space $V$ of the experiment. We also assume that the experiment begins at time $t_0$, that the Hilbert space is $\Hilbert = \Hilbert_\mathrm{sys}\otimes \Hilbert_\mathrm{env}$, where $\Hilbert_\mathrm{sys}$ is the Hilbert space of the system and $\Hilbert_\mathrm{env}$ that of its environment, and that the wave function at time $t_0$ is a product, $\Psi_{t_0} = \psi \otimes \phi$ (which expresses that the system and apparatus are initially independent and justifies saying that the system has wave function $\psi$). Finally, the distribution of the GRWf process is given by a history POVM $\povm(\cdot)$ on the appropriate history space. Now, the physical statement \eqref{mainthmpovm} reduces to the following mathematical statement (which is mathematically not deep):

\begin{thm}
Let $\Hilbert = \Hilbert_\mathrm{sys}\otimes \Hilbert_\mathrm{env}$ be a separable Hilbert space, $\povm(\cdot)$ a POVM on $(\Omega,\salg_\Omega)$ acting on $\Hilbert$, $\phi$ a fixed vector in $\Hilbert_\mathrm{env}$ with $\|\phi\|=1$, and $\zeta: (\Omega,\salg_\Omega) \to (V,\salg_V)$ a measurable function. For every $\psi\in\Hilbert_\mathrm{sys}$ with $\|\psi\|=1$, let $\Psi_{t_0} = \psi \otimes \phi$, $F_\psi$ be a random variable in $\Omega$ with distribution $\scp{\Psi_{t_0}}{\povm(\cdot) \, \Psi_{t_0}}$, and $R_\psi=\zeta(F_\psi)$. Then there is a POVM $E(\cdot)$ on $(V,\salg_V)$ acting on $\Hilbert_\mathrm{sys}$ so that the distribution of $R_\psi$ is $\scp{\psi}{E(\cdot) \, \psi}$.
\end{thm}

\begin{proof}
For $A\subseteq V$ with $A\in\salg_V$,
\be
\PPP(R\in A) = \PPP\bigl(F\in \zeta^{-1}(A)\bigr) = 
\scp{\Psi_{t_0}}{\povm\bigl(\zeta^{-1}(A)\bigr) \,\Psi_{t_0}} =
\ee
\be
= \scp{\psi\otimes \phi}{\povm\bigl(\zeta^{-1}(A)\bigr) \,\psi\otimes\phi} = 
\scp{\psi}{E(A)\,\psi}_\mathrm{sys}\,,
\ee
where $\scp{\cdot}{\cdot}_\mathrm{sys}$ denotes the scalar product in $\Hilbert_\mathrm{sys}$, and $E(A):\Hilbert_\mathrm{sys}\to\Hilbert_\mathrm{sys}$ is defined by first mapping $\psi \mapsto \povm\bigl(\zeta^{-1}(A)\bigr) \, \psi\otimes \phi$ and then taking the partial scalar product with $\phi$. The partial scalar product with $\phi$ is the adjoint of $\psi \mapsto \psi\otimes \phi$, indeed the unique bounded linear mapping $L_\phi:\Hilbert_\mathrm{sys}\otimes \Hilbert_\mathrm{env} \to \Hilbert_\mathrm{sys}$ such that
\be\label{partialscp}
L_\phi(\psi\otimes \chi) = \scp{\phi}{\chi}_\mathrm{env} \,\psi\,.
\ee
It has $\|L_\phi\| = \|\phi\|$ and satisfies
\be\label{partialscpscp}
\scp{\psi}{L_\phi\,\Psi}_\mathrm{sys} = \scp{\psi\otimes\phi}{\Psi} \,.
\ee

We check that $E(\cdot)$ is a POVM: For $A=V$ (the entire space), $\zeta^{-1}(V)=\Omega$ and $\povm\bigl(\zeta^{-1}(V)\bigr) = I$, and $E(V)=I$ by \eqref{partialscp}. For every $A$, $E(A)$ is clearly well defined and bounded, and positive by \eqref{partialscpscp}. The weak $\sigma$-additivity follows from that of $\povm(\cdot)$.
\end{proof}

[There does exist, though, another argument yielding \eqref{mainthmpovm}, due to D\"urr et al.~\cite{DGZ04}. It constitutes a proof of \eqref{mainthmpovm} from \emph{Bohmian mechanics}, another proposal for the precise definition of quantum mechanics; but on the basis of ordinary quantum mechanics it remains incomplete. Here is an outline of the argument: Suppose that the experiment begins at time $t_0$ and ends at $t_1$; that, as before, $\Hilbert = \Hilbert_\mathrm{sys}\otimes \Hilbert_\mathrm{env}$ and $\Psi_{t_0} = \psi \otimes \phi$; that the time evolution of the wave function is given by a unitary operator $U_{t_0}^{t_1}$, so that $\Psi_{t_1} = U^{t_1}_{t_0} \, \Psi_{t_0}$. Now assume Born's rule, according to which the probability distribution of the configuration $Q$ at time $t_1$ is $\scp{\Psi_{t_1}}{P(\cdot)\,\Psi_{t_1}}$ for a suitable PVM $P(\cdot)$ on configuration space $\Q$ acting on $\Hilbert$, the ``configuration PVM.'' Finally, assume that $R$ is a function of $Q$, $R=\zeta(Q)$. (Here is where the argument works in Bohmian mechanics but not really in ordinary quantum mechanics, as one assumes that the configuration is part of the primitive ontology.) Then
\be
\PPP(R\in A) = \scp{\psi\otimes\phi}{U_{t_0}^{t_1*} \,P\bigl(\zeta^{-1}(A) \bigr) \, U_{t_0}^{t_1} \, \psi\otimes\phi}=\scp{\psi}{E(A)\,\psi}\,,
\ee
and $E(\cdot)$ is a POVM.]

\bigskip

To sum up, the value of a precise definition of a physical theory is much the same as the value of a precise definition of a mathematical concept: It allows us to provide \emph{proofs} for statements that we are interested in. Without the precise definition, many of these statements remain mere \emph{guesses or intuitions}. And often, the clarity afforded by this precision helps us make new discoveries.

\bigskip
\bigskip

\noindent\textit{Acknowledgments.} I thank Valia Allori (Rutgers University), Angelo Bassi (LMU M\"un\-chen), Fay Dowker (Imperial College London), Detlef D\"urr (LMU M\"unchen), GianCarlo Ghirardi (ICTP Trieste), Sheldon Goldstein (Rutgers University), Frank Loose (T\"ubingen), Tim Maudlin (Rutgers University), Rainer Nagel (T\"ubingen), Travis Norsen (Marlboro College), Philip Pearle (Hamilton College), Peter Pickl (Wien), Reiner Sch\"atzle (T\"ubingen), Luca Tenuta (T\"ubingen), Stefan Teufel (T\"ubingen), Jakob Wachsmuth (T\"ubingen), and Nino Zangh\`\i\ (Genova) for helpful discussions at various times on various topics related to this work.

\end{document}